\documentclass[12pt, a4paper]{article}
\usepackage[ngerman,english]{babel} 

\usepackage[left=2.5cm,right=2.5cm,top=2.5cm,bottom=2.5cm,includefoot]{geometry}
\usepackage{amsmath}
\usepackage{amsthm} 
\usepackage{amsfonts}
\usepackage{booktabs}
\usepackage{dsfont}
\usepackage{graphicx,psfrag,epsf}
\usepackage{enumitem}
\usepackage{url} 
\usepackage{xcolor}
\usepackage{multirow,multicol}
\usepackage{tabularx}
\usepackage[title]{appendix}
\usepackage{placeins}
\usepackage{float}
\usepackage{stackengine}
\usepackage{setspace}
\usepackage[font={footnotesize}]{subcaption}
\usepackage{lscape}
\usepackage{lipsum}

\usepackage{todo}

\usepackage{xr-hyper} 
\usepackage[
breaklinks,
colorlinks=true,
pagebackref=false,
pdfauthor={},%
pdftitle={},%
citecolor=blue,
linkcolor=blue,
urlcolor=blue,
filecolor=blue
]{hyperref}      
\usepackage{cleveref}
\usepackage[authoryear,round]{natbib}

\numberwithin{equation}{section}

\DeclareMathOperator{\ES}{ES}
\DeclareMathOperator{\VaR}{VaR}

\DeclareMathOperator{\Var}{Var}
\DeclareMathOperator*{\argmin}{arg\,min}

\newcommand{\bq}{\hat{\boldsymbol{q}}_t}
\newcommand{\be}{\hat{\boldsymbol{e}}_t}

\newcommand{\tod}{\overset{d}{\longrightarrow}}
\newcommand{\toP}{\overset{P}{\longrightarrow}}

\newcommand{\sumt}{\sum_{t \in \mathfrak{T}}}

\newcommand{\sumtj}{\sum_{t \in \mathfrak{T}_j}}

\allowdisplaybreaks

\newtheoremstyle{boldtitle}
{}
{}
{}
{}
{\bfseries}
{.}
{.5em}
{{\thmname{#1 }}{\thmnumber{#2}}{\thmnote{ (#3)}}}

\theoremstyle{boldtitle}

\newtheorem{theorem}{Theorem}
\newtheorem{lemma}{Lemma}
\newtheorem{definition}{Definition}

\newtheorem{assumption}{Assumption}

\definecolor{TD}{rgb}{0,0.5,0}

\definecolor{JS}{rgb}{0.8,0.3,0.2}

\begin{document}

	\def\spacingset#1{\renewcommand{\baselinestretch}%
		{#1}\small\normalsize} \spacingset{1}

%

	\title{\bf Encompassing Tests for Value at Risk and Expected Shortfall Multi-Step Forecasts based on Inference on the Boundary}
	
	\author{Timo Dimitriadis\thanks{Heidelberg Institute for Theoretical Studies (HITS), Heidelberg, 69118 Heidelberg and
			University of Hohenheim, Institute of Economics, 70599 Stuttgart, Germany, e-mail: \href{mailto:timo.dimitriadis@h-its.org}{timo.dimitriadis@h-its.org}}
		\and Xiaochun Liu\thanks{Department of Economics, Finance and Legal Studies, Culverhouse College of Business, University of Alabama, Tuscaloosa Alabama 35487 USA. e-mail: \href{mailto:xliu121@ua.edu}{xliu121@ua.edu}}
		\and Julie Schnaitmann\thanks{University of Konstanz, Department of Economics, 78457 Konstanz, Germany, e-mail:  \href{mailto:julie.schnaitmann@uni-konstanz.de}{julie.schnaitmann@uni-konstanz.de}}
	}

%
%
%
%
%
%
%

	\maketitle
	
	\bigskip
	\begin{abstract}
		We propose forecast encompassing tests for the Expected Shortfall (ES) jointly with the Value at Risk (VaR) based on flexible link (or combination) functions.
		Our setup allows testing encompassing for convex forecast combinations and for link functions which preclude crossings of the combined VaR and ES forecasts.
		As the tests based on these link functions involve parameters which are on the boundary of the parameter space under the null hypothesis, we derive and base our tests on nonstandard asymptotic theory on the boundary.
		Our simulation study shows that the encompassing tests based on our new link functions outperform tests based on unrestricted linear link functions for one-step and multi-step forecasts.
    	We further illustrate the potential of the proposed tests in a real data analysis for forecasting VaR and ES of the S\&P\,500 index. 
	\end{abstract}
	
	\noindent%
	{\it Keywords}: asymptotic theory on the boundary, joint elicitability, multi-step ahead and aggregate forecasts, forecast evaluation and combinations 
	\\[0.5em]
	{\it JEL}: C12, C52, C58
	
	
	\onehalfspacing

\clearpage		
	
\section{Introduction}
\label{sec:Introduction}

For nearly two decades, financial institutions and regulators have advocated Value at Risk (VaR) as the main tool for risk management and capital allocation. 
Owing to a number of weaknesses, including the failure of capturing (extreme) tail risks and hence discouraging risk diversification \citep{Artzner1999, Acerbi2002, Tasche2002}, the Basel Committee on Banking Supervision (BCBS) has recently adopted Expected Shortfall (ES), complementing and in parts substituting VaR as the fundamental measure for market risk \citep{Basel2013, Basel2016, Basel2017, Basel2019}.

The ES at level $\alpha \in (0,1)$ is defined as the expected return beyond the $\alpha$-quantile and it is widely used as a coherent measure of tail risks \citep{Artzner1999, Tasche2002}. 
Nonetheless, its inherent deficiency is that the ES is not elicitable on its own, meaning that the ES cannot be obtained as the unique minimizer of the expectation of a loss (scoring) function, see e.g., \cite{Gneiting2011}.
However, \cite{Fissler2016} show that the VaR and ES are jointly elicitable (or 2-elicitable). 
This joint elicitability property directly hints towards evaluating the ES jointly with the VaR in a unified framework \citep{FisslerZiegelGneiting2016}, as in the present study concerning forecast encompassing tests.

Forecast encompassing of two competing forecasts tests whether one forecast alone performs not worse than any forecast combination, stemming from some parametric combination formula, also denoted by \textit{link functions} in this article.
If this holds, the rival forecast contains no additionally useful information relative to the first forecast \citep{HendryRichard1982, MizonRichard1986}.
This makes forecast encompassing tests an attractive tool for the empirical comparison of competing forecasts, especially when focusing on  efficiency gains stemming from forecast combinations.\footnote{
For recent empirical applications of forecast encompassing tests, see e.g., \cite{Taylor2005, Busetti2013, FuertesOlmo2013, Costantini2017, Liu2017, Zhao2017, Tsiotas2018, Clements2020, YouLiu2020} among others.}
As meaningful measures of forecast performance are based on strictly consistent loss functions \citep{Gneiting2011}, this forcefully illustrates the importance of the existence of such loss functions for testing forecast encompassing.
Hence, we build our encompassing tests on joint loss functions for the VaR and ES \citep{Fissler2016}, and on recently developed joint semiparametric VaR and ES models \citep{Patton2019, DimiBayer2019, Taylor2019, Barendse2020}.		

As the main methodological contribution of this paper, we introduce encompassing tests for the ES jointly with the VaR based on flexible \textit{link functions} or \textit{combination formulas}, which allow for several important specifications that go beyond those of existing encompassing tests of e.g.\ \cite{GiacominiKomunjer2005} and \cite{DimiSchnaitmann2020}.
While linear forecast combination methods with unrestricted parameters are the most prominent class of link functions used for encompassing tests,\footnote{see e.g., \cite{HendryRichard1982,MizonRichard1986, Diebold1989, GiacominiKomunjer2005, ClementsHarvey2009, ClementsHarvey2010, DimiSchnaitmann2020}.} more flexible approaches are especially important for joint tests of the VaR and ES:
First, unrestricted linear link functions regularly result in \textit{VaR and ES crossings}, i.e.\ days where the optimally combined ES forecast is larger than the VaR forecast, which immediately contradicts their definitions \citep{Taylor2020}.
To this end, we propose the \textit{no-crossing} link functions which impede such crossings.
Second, \textit{convex} forecast combinations present an attractive alternative as their structure can stabilize the forecast performance and reduce the estimation noise \citep{Timmermann2006, Hansen2008, Bayer2018}, which is particularly important for the case of semiparametric models for the VaR and ES with extreme probability levels \citep{DimiBayer2019}.

The link function specifications considered in this article imply that certain model parameters are on the boundary of the parameter space under the null hypothesis.
This boundary issue is exemplified by encompassing tests for convex forecast combinations, which entails testing whether the convex combination parameter is one (or zero).
Under the null hypothesis, this parameter lies on the boundary of the admissible parameter space, i.e., the unit interval. 
Hence,  we derive novel and nonstandard asymptotic theory for the model parameters and the resulting Wald test statistics for semiparametric models for the VaR and ES which allows some (or all) of the true model parameters to be on the boundary of the parameter space.
For this, we follow the approach of \cite{Andrews1999} and \cite{Andrews2001}, where the proofs use empirical process methods of \cite{Andrews1994} and \cite{Doukhan1995}.
To render our tests practically feasible, we draw critical values from the resulting nonstandard asymptotic distributions of the Wald test statistics obtained from simulations involving the solution of quadratic programming problems.

The proposed encompassing tests allow for testing one-step ahead, multi-step ahead and multi-step aggregate forecasts, where the consideration of multi-step forecasts requires the application of a VaR and ES specific adaption of the HAC (Heteroskedasticity and Autocorrelation Consistent) estimator of \cite{NeweyWest1987} and \cite{Andrews1991}.
The examination of multi-step (aggregate) forecasts is particularly relevant for the risk measures VaR and ES due to the explicit calls for 10-day aggregate VaR and ES forecasts of the \cite{Basel2016, Basel2017, Basel2019, Basel2020}.
Furthermore, this goes beyond many recent papers concerning forecast evaluation procedures for the VaR and ES, which mainly focus on one-step ahead forecasts.\footnote{see e.g.\ \cite{Kratz2017, CostanzinoCurran2018, BayerDimi2019, CouperierLeymarie2019, Patton2019, DimiSchnaitmann2020}.}

Our simulations show that the encompassing tests for the VaR and ES based on our new link functions and on inference on the boundary exhibit accurate empirical sizes and good power properties.
In particular, we find that these test specifications outperform classical tests based on unrestricted linear link functions for the VaR and ES of \cite{DimiSchnaitmann2020} throughout all considered simulation designs and for all, one-step ahead, multi-step ahead, and multi-step aggregate forecasts.
However, we find that long forecast horizons (e.g., $10$ days) paired with short evaluation periods of less than $1000$ trading days result in unreliable test decisions.
This simulation result sheds a critical light on the recent evaluation methods based on relatively short evaluation periods proposed by the \cite{Basel2016, Basel2017, Basel2019, Basel2020}.

We empirically illustrate the usefulness of the VaR and ES encompassing tests based on convex link functions by comparing one-day ahead, and 10-day ahead and aggregate VaR and ES forecasts for daily S\&P\,500 index returns.
We estimate eleven risk models in a rolling forecast scheme, including several GARCH specifications and the ES-specific semiparametric models of \cite{Taylor2019} and \cite{Patton2019}.
For the evaluation period from July 2008 to June 2020, we find that the new tests assign much higher optimal weights to models specified with asymmetric volatilities, asymmetric and fat-tailed residual distributions and dynamic higher moments, especially for multi-step forecasts.

Our paper is closely related to the recent work of \cite{DimiSchnaitmann2020}, but differs in the following ways.
First, our proposed encompassing tests extend the ones of \cite{DimiSchnaitmann2020} with more flexible link functions by allowing the true parameters under the null to be on the boundary of the parameter space.
Second, while the theoretical contribution of \cite{DimiSchnaitmann2020} focuses on the inclusion of misspecified models for encompassing tests, we establish inference on the boundary of the parameter space.
Third, while our considered link functions allow for theoretically appealing specifications, they also exhibit clearly superior empirical properties in a wide range of simulations.
Fourth, following the recent regulations of the \citeauthor{Basel2019} (\citeyear{Basel2019}, \citeyear{Basel2020}), we consider multi-step ahead and aggregate VaR and ES forecasts in the simulations and the empirical application of this article.

The remainder of the paper is organized as follows. 
Section \ref{sec:Theory} proposes the joint encompassing tests based on flexible link functions and develops asymptotic theory for the joint VaR and ES models for parameters on the boundary of the parameter space.
Section \ref{sec:Simulations} presents simulations for the encompassing tests and 
Section \ref{sec:Application} applies the tests to VaR and ES forecasts for the S\&P\,500 index. 
Section \ref{sec:Conclusion} concludes.
The proofs are given in Appendix \ref{sec:Proofs}, and a supplementary material document contains additional material for the paper, where references starting with \textcolor{blue}{S.} refer to the supplement.

\section{Theory}
\label{sec:Theory}
	
\subsection{Setup and Notation}
	
	We follow the general setup of \cite{GiacominiKomunjer2005} and \cite{DimiSchnaitmann2020} while further allowing for multi-step forecasts.
	For this, we consider a stationary stochastic process $Z = \left\{ Z_t: \Omega \to \mathbb{R}^{\tilde l+1}, \, t= 1,\dots,R, \,\tilde l \in \mathbb{N}, \,R\in \mathbb{N} \right\}$, which is defined on some common and complete probability space $(\Omega, \mathcal{F}, \mathbb{P})$, where $\mathcal{F} = \left\{ \mathcal{F}_t, t = 1, \dots, R \right\}$ and $\mathcal{F}_t = \sigma \left\{ Z_s, s\le t \right\}$.
	We partition the stochastic process as $Z_t = (Y_t,X_t)$, where $Y_t: \Omega \to \mathbb{R}$ is an absolutely continuous random variable of interest and $X_t: \Omega \to \mathbb{R}^{\tilde l}$ is a vector of explanatory variables.
	For some fixed forecast horizon $h \in \mathbb{N}$, we denote the conditional distribution of $Y_{t+h}$ given the information set $\mathcal{F}_t$ by $F_t$.
	Accordingly, $\mathbb{E}_t$, $\Var_t$ and $h_t$ denote the expectation, variance and density corresponding to $F_t$.
	The conditional VaR of $Y_{t+h}$ given $\mathcal{F}_t$ at probability level $\alpha \in (0,1)$ is formally defined as
	\begin{align}
		\VaR_{t,\alpha} ( Y_{t+h} ) = F_t^{-1}(\alpha) = \inf \{ z \in \mathbb{R}: F_t(z) \ge \alpha \},
	\end{align}
	and given that $F_t$ is continuous at its $\alpha$-quantile, the conditional ES of $Y_{t+h}$ given $\mathcal{F}_t$ at level $\alpha \in (0,1)$ is defined by
	\begin{align}
		\ES_{t,\alpha} ( Y_{t+h} ) = \mathbb{E}_t \big[ Y_{t+h}  \mid Y_{t+h} \le \VaR_{t,\alpha}(Y_{t+h}) \big].
	\end{align}

	In order to allow for forecast evaluation of multi-step (ahead and aggregate) forecasts with horizon $h \in \mathbb{N}$ in an out-of-sample fashion, we split $R = S+T+h-1$,
	where $S \in \mathbb{N}$ denotes the length of the in-sample and $T \in \mathbb{N}$ of the out-of sample window.
	In detail, for all $t \in \mathbb{N}$, such that $S \le t \le S+T-1$, we generate $h$-step ahead VaR and ES forecasts for the random variables $Y_{t+h}$ (i.e.\ for the sequence $(Y_{S+h}, \dots, Y_{S+T+h-1})$) based on the previous $S$ data points.
	For convenience of the notation, we define the set $\mathfrak{T} := \{ t \in \mathbb{N}: S \le t \le S+T-1 \}$ corresponding to the time points the forecasts are issued for the out-of-sample period.
	
	We further denote the competing, $\mathcal{F}_t$-measurable, $h$-step forecasts for the VaR and ES by $\hat q_{j,t}$ and $\hat e_{j,t}$, for $j=1,2$.
	Following \cite{GiacominiKomunjer2005}, we assume that these are generated through a function $f \big( \gamma_{t,S}, Z_t, Z_{t-1},\dots \big)$, which is fixed over time.
	For this, $\gamma_{t,S}$ denotes the (estimated or fixed) model parameters at time $t$, or alternatively the semi- or non-parametric estimator used in the construction of the forecasts, (possibly) estimated by data from the in-sample period of length $S$.
	This construction allows for forecasting schemes with fixed (or no) parameters, forecasting schemes with model parameters $\gamma_{t,S}$ that are estimated only once, and rolling window forecasting schemes where the parameters $\gamma_{t,S}$ are re-estimated in each step \citep{GiacominiKomunjer2005}.
	In our testing approach, we focus on evaluation of the entire \textit{forecasting method}  as e.g.\ in \cite{GiacominiKomunjer2005} and \cite{GiacominiWhite2006}, instead of on a \textit{forecasting model}, as e.g.\ in \cite{West1996, West2001}.
	The stacked forecasts are denoted by $\bq = (\hat q_{1,t}, \hat q_{2,t})$ for the VaR, and by $\be = (\hat e_{1,t}, \hat e_{2,t})$ for the ES.\footnote{A generalization of our framework to test encompassing for \textit{multiple} competing forecasts ($K\ge2$) in the sense of \cite{HarveyNewbold2000} is readily available by generalizing the notation as $\bq = (\hat q_{1,t}, \dots, \hat q_{K,t})$ and $\be = (\hat e_{1,t}, \dots, \hat e_{K,t})$ and by further using suitable specifications for the link functions and the null hypotheses in the subsequent derivations.}
	In our notation of the forecasts, we stress the dependence on $t$, the time-point they are issued, while suppressing the dependence on the forecast horizon $h$ as it is treated as fixed.

	Let $r_t$ denote financial log-returns for day $t$. 
	Then, our theoretical setup allows for the treatment of classical multi-step ($h$-step) ahead forecasts, but also for \textit{$h$-step aggregate forecasts} in the sense of an aggregated return over $h$ days, such as the 10-day aggregate VaR and ES forecasts explicitly stated in the regulatory framework of the \cite{Basel2019, Basel2020}.
	For  classical \textit{$h$-step ahead forecasts}, we use $Y_{t+h} = r_{t+h}$, while for \textit{$h$-step aggregate forecasts} we choose $Y_{t+h} = \sum_{s=1}^{h} r_{t+s}$.

	In the following exposition, all vectors refer to column vectors.
	For splitting of subvectors, we often abuse notation and write $\theta = (\theta_1, \theta_2)$ instead of $\theta = (\theta_1^\top, \theta_2^\top)^\top$.
	The operator $\nabla$ denotes the derivative with respect to $\theta$.
	All limits below are taken ``as $T \to \infty$'' unless stated otherwise and $\toP$ and $\tod$ denote convergence in probability and distribution respectively.
	Let $:=$ denote an equality ``by definition''.
	Furthermore, let $\mathbb{R}_+$ and $\mathbb{R}_-$ denote the non-negative and non-positive real half-lines respectively and we define $\mathbb{R}_C = \{ z \in \mathbb{R}: |z| \le C \}$ to be a sufficiently large compact subset of the real numbers (for some $C \in \mathbb{R}_+$ large enough).

	\subsection{Joint Encompassing Tests for VaR and ES Forecasts}
	
	For the introduction of the joint encompassing tests for VaR and ES forecasts, we follow \cite{DimiSchnaitmann2020} and define the flexible link (or combination) functions
	\begin{alignat}{3}
		&g^q: \mathfrak{Q} \times \mathfrak{E}  \times \Theta \to \mathbb{R}&, \qquad &(\bq, \be, \theta) \mapsto g^q(\bq, \be, \theta), \\
		&g^e: \mathfrak{Q} \times \mathfrak{E}  \times \Theta \to \mathbb{R}&, \qquad &(\bq, \be, \theta) \mapsto g^e(\bq, \be, \theta),
	\end{alignat}
	based on some (compact) parameter space $\Theta \subset \mathbb{R}^k$, where $\mathfrak{Q}$ and $\mathfrak{E}$ denote the random spaces of the VaR and ES forecasts.
	These link functions represent the parametric, functional forms of the forecast combinations we consider.\footnote{The link functions can alternatively be interpreted as (semi-) parametric models for the conditional quantile (VaR) and ES of $F_t$ as in \cite{Patton2019}.} 
	E.g., in the classical case of testing forecast encompassing, these link functions are linear with (essentially) unrestricted parameter spaces.
	For convenience of notation, we henceforth use the short forms
	\begin{align}
		\label{eqn:ParametricModelsEncompassing}
		g_t^q(\theta)  := g^q(\bq, \be, \theta), \qquad \text{ and } \qquad g_t^e(\theta)  := g^e(\bq, \be, \theta).
	\end{align}	
	We further assume that there exists a unique \textit{test parameter value} $\theta_\ast \in \Theta$ such that $g^q(\bq, \be, \theta_\ast) = \hat q_{1,t}$ and $g^e(\bq, \be, \theta_\ast) = \hat e_{1,t}$ almost surely.
	This assumption ensures that the parametric link function allows for the trivial forecast combination of only choosing the first forecast.\footnote{As the encompassing tests in this article are always formulated as forecast (pair) one encompasses forecast (pair) two, we only assume the existence of $\theta_\ast$ corresponding to the first (pair) of forecasts. 
		Testing the inverted encompassing hypothesis that the second pair of forecasts encompasses the first forecast pair can be carried out by interchanging the forecast pairs.
		Alternatively, one could assume that a value $\tilde \theta_\ast$ exists such that $g^q(\bq, \be,  \tilde \theta_\ast) = \hat q_{2,t}$ and $g^e(\bq, \be, \tilde \theta_\ast) = \hat e_{2,t}$ holds almost surely.}
	In the classical case of unrestricted linear link functions, $\theta_\ast$ often corresponds to $(1,0)$ or $(0,1,0)$, depending on whether an intercept is included in the model.
	We refer to Section \ref{sec:LinkSpecifications} for details and examples of these link functions.

	\cite{Gneiting2011} shows that the ES stand-alone is not elicitable, i.e.\ there do not exist suitable (strictly consistent) loss functions, which are a central ingredient for encompassing tests \citep{GiacominiKomunjer2005, DimiSchnaitmann2020}.
	\cite{Fissler2016} overcome this deficiency and show that there exist joint loss functions for the VaR and ES and further characterize this class subject to mild regularity conditions by
	\begin{align} 
		\begin{aligned}
		\label{eqn:JointLossESRegGeneral}
		\rho \big( Y, q, e \big) &=  \big( \mathds{1}_{\{Y \le q\}} - \alpha \big) \mathfrak{g}(q) - \mathds{1}_{\{Y \le q \}}  \mathfrak{g}(Y) \\
		&\quad+ \phi'(e) \left( e -  q + \frac{( q - Y) \mathds{1}_{\{ Y \le  q \}}}{\alpha}  \right) - \phi( e ) + a(Y),
		\end{aligned}
	\end{align}
	where the function $\mathfrak{g}$ is twice continuously differentiable and increasing, $\phi$  is three times continuously differentiable, strictly increasing and strictly convex, and $a$ and $\mathfrak{g}$ are $Y_{t+h}$-integrable functions.
	The most prominent candidate of this class is the zero-homogeneous loss function \citep{NoldeZiegel2017AAS}, sometimes called the FZ0 loss, 
	\begin{align} 
		\begin{aligned}
		\label{eqn:JointLossESReg0Hom}
		\rho^{\text{FZ0}} \big( Y, q, e \big) =
		- \frac{1}{e } \left( e - q + \frac{(q  - Y) \mathds{1}_{\{Y \le q \}}}{\alpha}  \right) + \log(-e),
		\end{aligned}
	\end{align}	
	which is obtained by choosing $\mathfrak{g}(z) = 0$, $a(z) = 0$ and $\phi(z) = -\log(-z)$ in (\ref{eqn:JointLossESRegGeneral}).
	We henceforth often use the short notations $\rho_t(\theta) := \rho \big( Y_{t+h}, g^q_t(\theta) , g^e_t(\theta)  \big)$ and $\rho^{\text{FZ0}}_t(\theta) := \rho^{\text{FZ0}} \big( Y_{t+h}, g^q_t(\theta) , g^e_t(\theta)  \big)$.
	
	Using the general class of loss functions in (\ref{eqn:JointLossESRegGeneral}), we define the \textit{true} regression (or combination) parameter $\theta^0 \in \Theta$ by
	\begin{align}
		\label{eqn:DefTrueParameter}
		\theta^0 :=  \underset{\theta \in \Theta}{\argmin} \, \mathbb{E} \left[ \rho \big( Y_{t+h}, g^q_t(\theta) , g^e_t(\theta)  \big) \right],
	\end{align}
	which is independent of $t$ as we assume stationarity of the process $Z$.\footnote{See e.g.\ \cite{Patton2019}, \cite{DimiBayer2019}, \cite{BayerDimi2019}, \cite{DimiSchnaitmann2020} and \cite{Barendse2020} for details on joint (semi-) parametric models for the VaR and ES.}
	The strict consistency result of the loss function from \cite{Fissler2016} together with further weak regularity conditions on the link functions implies that
	\begin{align}
		Q(Y_{t+h} \mid \mathcal{F}_t) = g_t^q(\theta^0)  \qquad \text{ and } \qquad \ES(Y_{t+h} \mid \mathcal{F}_t) = g_t^e(\theta^0)
	\end{align}
	almost surely, which justifies the notion of the \textit{true regression parameter}.
	
	We now define joint forecast encompassing for the VaR and ES following \cite{GiacominiKomunjer2005} and \cite{DimiSchnaitmann2020}.
	\begin{definition}[Joint VaR and ES Forecast Encompassing]
		\label{def:QuantileESEncompassing}
		We say that the pair $\big( \hat q_{1,t}, \hat e_{1,t}\big)$ jointly encompasses $\big( \hat q_{2,t}, \hat e_{2,t}\big)$  at time $t$ with respect to the link functions $g^q$ and $g^e$ if and only if
		\begin{align}
			\label{eqn:UncondEncompVaRES}
			\mathbb{E} \left[ \rho \big( Y_{t+h}, \hat q_{1,t}, \hat e_{1,t} \big) \right]
			= \mathbb{E} \left[ \rho  \big( Y_{t+1}, g^q(\bq, \be, \theta^0), g^e(\bq, \be, \theta^0) \big) \right],
		\end{align}
		where the loss function $\rho$ is given in (\ref{eqn:JointLossESRegGeneral}).
	\end{definition}

	This holds if and only if $\theta^0 = \theta_\ast$ as we impose uniqueness of the parameter $\theta_\ast$.
	The intuition behind the specification in (\ref{eqn:UncondEncompVaRES}) is that the forecasts $(\hat q_{1,t}, \hat e_{1,t})$ generate the same expected loss as an optimal forecast combination $\big( g^q(\bq, \theta^0), g^e(\be, \theta^0) \big)$ based on the \textit{optimal} combination parameter defined in (\ref{eqn:DefTrueParameter}).
	Hence, using the first pair of forecasts $(\hat q_{1,t}, \hat e_{1,t})$ is the optimal, but trivial forecast combination.
	From a different point of view, this implies that the second pair of forecasts $(\hat q_{2,t}, \hat e_{2,t})$ does not add any useful information which is not already contained in $(\hat q_{1,t}, \hat e_{1,t})$.

	If the interest is mainly placed on the performance of the competing ES forecasts, one can consider the \textit{auxiliary} ES encompassing test in the spirit of \cite{DimiSchnaitmann2020}.\footnote{Application of the \textit{strict} encompassing test of \cite{DimiSchnaitmann2020} in the setting of the present article further requires combining the asymptotic theory under misspecification of \cite{DimiSchnaitmann2020} with the theory of estimation and testing at the boundary of the present article.}
	\begin{definition}[Auxiliary  ES Forecast Encompassing]
		\label{def:AuxESEncompassing}
		We say that the forecast $\hat e_{1,t}$ auxiliarily encompasses its rival $\hat e_{2,t}$ at time $t$ with respect to the link functions $g^q$ and $g^e$ if and only if
		\begin{align}
			\label{eqn:UncondEncompAuixES}
			\mathbb{E} \left[ \rho \big( Y_{t+h}, g^q(\bq, \be, \theta^0), \hat e_{1,t} \big) \right]
			= \mathbb{E} \left[ \rho  \big( Y_{t+1}, g^q(\bq, \be, \theta^0), g^e(\bq, \be, \theta^0) \big) \right],
		\end{align}
		where the loss function $\rho$ is given in (\ref{eqn:JointLossESRegGeneral}).
	\end{definition}
	Finding testable conditions for the auxiliary test, corresponding to the condition $\theta^0 = \theta_\ast$ for the joint test, has to be done on a case-by-case basis for the link functions under consideration, see Section \ref{sec:LinkSpecifications} for further details.
	
	Given a sample of competing forecasts and corresponding realizations, we can test whether the sequence of joint VaR and ES forecasts $( \hat q_{1,t}, \hat e_{1,t} )$ encompasses the sequence $( \hat q_{2,t}, \hat e_{2,t} )$ for all $t \in \mathfrak{T}$ (in the out-of-sample period) by estimating the parameters of the semiparametric models
	\begin{align}
		\begin{aligned}
			\label{eqn:JointRegressionJointESEncTest}
			Y_{t+h} = g_t^q(\theta) + u_{t+h}^q, \qquad \text{ and } \qquad
			Y_{t+h} = g_t^e(\theta) + u_{t+h}^e,
		\end{aligned}
	\end{align}
	where $Q_\alpha(u_{t+h}^q \mid \mathcal{F}_t) = 0$ and $\ES_\alpha(u_{t+h}^e \mid \mathcal{F}_t) = 0$ almost surely  for all $t \in \mathfrak{T}$ by using the M-estimator introduced in \cite{Patton2019} and \cite{DimiBayer2019}, and by testing whether $\theta_\ast = \theta^0$ using a Wald test.

	Differently from \cite{DimiSchnaitmann2020} and the remaining literature on testing forecast encompassing, we allow the true, optimal parameter $\theta^0$ to be on the \textit{boundary} of $\Theta$ under the null hypothesis.
	This facilitates the consideration of several important link function specifications.
	E.g., this enables to test encompassing for link specifications which \textit{theoretically prevent crossings} of the combined VaR and ES forecasts in the sense that $g_t^e(\theta) \le g_t^q(\theta)$ almost surely for all $t \in \mathfrak{T}$ \citep{Taylor2020}.
	Furthermore, we can test forecast encompassing based on \textit{convex} forecast combinations, which stabilizes the parameter estimation.
	While the subsequent section focuses mainly on these two examples, our approach is by no means limited to these link functions.

	\subsection{The Link Function Specifications}
	\label{sec:LinkSpecifications}

	In this section, we introduce three link function specifications which are of interest for this article, where other link functions can be treated along the lines of this section by employing an equivalent split of the parameter vector and by formulating the null hypotheses accordingly.
	The treatment of asymptotic theory on the boundary in the sense of \cite{Andrews2001}, detailed in Section \ref{sec:AsymptoticTheory} of the present article, requires splitting the parameter vector $\theta$ into the following structurally different subvectors,
	\begin{align}
	\label{eqn:PartitioningThetaParameter}
	\theta = \big( \beta_1, \beta_2, \delta, \psi),
	\end{align}
	where $\beta_1 \in \mathcal{B}_1 \subseteq \mathbb{R}^{p_1}$, $\beta_2 \in \mathcal{B}_2 \subseteq \mathbb{R}^{p_2}$, $\delta \in \Delta \subseteq \mathbb{R}^{q}$ and $\psi \in \Psi \subseteq \mathbb{R}^{s}$, where $p_1 + p_2 + q + s = k$, $p := p_1 + p_2$ and $\Theta = \mathcal{B}_1 \times \mathcal{B}_2 \times \Delta \times \Psi$.
	The intuition behind this decomposition is the following:
	(1)  the null hypothesis we test for is based on $\beta_1$ only, and $\beta_1$ may or may not be on the boundary of the parameter space;
	(2) $\beta_2$ may or may not be on the boundary, but it is not tested for;
	(3) $\delta$ is not on the boundary, and it is not tested for;
	(4) $\psi$ is not tested for, it may or may not be on the boundary, and the off-diagonal elements of the matrix $\mathcal{T}$, defined later in (\ref{eqn:T_Matrix}), corresponding to interactions of $\psi$ and $(\beta_1, \beta_2, \delta)$ are zero.
	
	Most importantly, the null hypothesis is based on $\beta_1$ only, while the remaining parameters can be thought of as nuisance parameters, required for the estimation of the model.
	The distinction between $\psi$ and the remaining parameter subvectors (in particular $\beta_2$) is that the imposed nullity of certain off-diagonal elements of $\mathcal{T}$ implies that the asymptotic distribution of $\beta_1$ is not affected by whether $\psi$ is on the boundary or not.
	
	Using the subvector decomposition in (\ref{eqn:PartitioningThetaParameter}), we can formally introduce the link functions and the corresponding null hypotheses of interest for the \textit{joint} and \textit{auxiliary} encompassing tests.
	The subsequent orderings of the parameters $\theta$ follows the ordering in the decomposition in (\ref{eqn:PartitioningThetaParameter}).
	All following encompassing null hypotheses are formulated for the test that the forecast pair $(\hat q_{1,t}, \hat e_{1,t})$ encompasses $(\hat q_{2,t}, \hat e_{2,t})$, whereas the reverse tests can be defined by simply interchanging the forecast pairs.
	\begin{enumerate}[label=(\arabic*)]
		\item 
		\textbf{(Unrestricted) Linear:}
		The unrestricted linear link functions are given by
		\begin{align}
		g_t^q(\theta) &= \theta_6 + \theta_3 \hat q_{1,t} + \theta_4 \hat q_{2,t}, \qquad \text{ and } \\
		g_t^e(\theta)  &= \theta_5 + \theta_1 \hat e_{1,t} + \theta_2 \hat e_{2,t},
		\end{align}
		where the parameter space $\Theta := \mathbb{R}_C^6$ is essentially unrestricted, as the constant $C$ can be chosen sufficiently large.
		We henceforth denote these link functions as \textit{linear}.
		We then test (a) $\mathbb{H}_0^{\text{Joint}}: (\theta_1, \theta_2, \theta_3, \theta_4) = (1,0,1,0)$, and (b) $\mathbb{H}_0^{\text{Aux}}: (\theta_1, \theta_2) = (1,0)$.\footnote{In terms of the subvectors decomposition in (\ref{eqn:PartitioningThetaParameter}), we can assign $\beta_1 := (\theta_1, \theta_2, \theta_3, \theta_4)$ and $\delta := (\theta_5, \theta_6)$ for the joint test and $\beta_1 := (\theta_1, \theta_2)$ and $\delta := (\theta_3, \theta_4, \theta_5, \theta_6)$ for the auxiliary test. 
		As the parameter subvector $\beta_1$ is in the interior of the parameter space under the null for both tests, classical asymptotic theory is sufficient for this unrestricted linear link function specification.
		}
		This corresponds to the standard case of forecast encompassing tests \citep{FairShiller1989, ClementsHarvey2009}, which is already considered by \cite{DimiSchnaitmann2020} for the case of the VaR and ES.
		As none of the parameters are on the boundary under the null, standard asymptotic theory is sufficient here and we use this specification as the benchmark in this paper.

		\item 
		\textbf{Convex Combinations:} 
		We consider the link functions
		\begin{align}
			g_t^q(\theta)  &= \theta_4 + \theta_2 \hat q_{1,t} + (1-\theta_2) \hat q_{2,t}, \qquad \text{ and } \\
			g_t^e(\theta)  &= \theta_3 + \theta_1 \hat e_{1,t} + (1-\theta_1) \hat e_{2,t},
		\end{align}
		where $\Theta :=  [0,1]^2 \times\mathbb{R}_C^2$.
		We then test the following null hypothesis:
		\begin{enumerate}
			\item 
			$\mathbb{H}_0^{\text{Joint}}: (\theta_1, \theta_2) = (1,1)$,
			and we assign
			$\beta_1 := (\theta_1, \theta_2) \in \mathcal{B}_1 :=  [0,1]^2$, and
			$\delta := (\theta_3, \theta_4) \in \Delta := \mathbb{R}_C^2$.

			\item 
				$\mathbb{H}_0^{\text{Aux}}: \theta_1 = 1$,
				and we assign
				$\beta_1 := \theta_1 \in \mathcal{B}_1 :=  [0,1]$,
				$\beta_2 := \theta_2 \in \mathcal{B}_2 :=  [0,1]$ and
				$\delta := (\theta_3, \theta_4) \in \Delta := \mathbb{R}_C^2$.
		\end{enumerate}
		In comparison with the linear link functions, the convex forecast combinations require estimation of less parameters and therefore stabilizes the parameter estimation, especially for highly correlated forecasts.\footnote{Notice that for the estimation of joint VaR and ES models, especially for extreme probabilities such as $\alpha = 2.5\%$, adding additional parameters is costly in terms of both, computation times and estimation noise, see e.g.\ the simulations of \cite{DimiBayer2019} for details.}
		For both hypotheses formulated above, $\theta_1$ and $\theta_2$ are on the boundary under the null, while $\theta_3$ and $\theta_4$ are not.
		The latter parameters are assigned to $\delta$ instead of $\psi$ as the matrix $\mathcal{T}$, given in (\ref{eqn:T_Matrix}), does not have null entries at the respective points.
		As the tested parameters are on the boundary of the parameter space under the null hypotheses of both tests, their corresponding Wald test statistics are subject to a non-standard asymptotic distribution \citep{Andrews1999, Andrews2001}.
		
		\item 
		\textbf{No VaR and ES Crossing:} 
		We consider the link functions
		\begin{align}
			g_t^q(\theta)  &= \theta_3 + \theta_1 \hat e_{1,t} + (1-\theta_1) \hat e_{2,t} + \theta_2 \big(\hat q_{1,t}  -\hat e_{1,t} \big) + (1-\theta_2) \big(\hat q_{2,t}  -\hat e_{2,t} \big), \; \text{ and } \label{eqn:Var_nocrossing}\\ 
			g_t^e(\theta)  &= \theta_3 + \theta_1 \hat e_{1,t} + (1-\theta_1) \hat e_{2,t},  
		\end{align}
		where $\Theta := [0,1]^2 \times \mathbb{R}_C^2$.
		These link functions imply that $g_t^q(\theta) \ge g_t^e(\theta)$ holds almost surely for all $t \in \mathfrak{T}$, which can be interpreted as a necessary condition for sensible (combinations of) VaR and ES forecasts, which is closely related to the issue of \textit{quantile crossings} in quantile regression \citep{Koenker2005book}.
		We then test
		\begin{enumerate}
			\item 
			$\mathbb{H}_0^{\text{Joint}}: (\theta_1, \theta_2) = (1,1)$, and we assign
			$\beta_1 := (\theta_1, \theta_2) \in \mathcal{B}_1 :=  [0,1]^2$, and
			$\delta := \theta_3 \in \Delta := \mathbb{R}_C$.
			
			\item 
			$\mathbb{H}_0^{\text{Aux}}: \theta_1 = 1$, and we assign
			$\beta_1 := \theta_1 \in \mathcal{B}_1 :=  [0,1]$,
			$\beta_2 := \theta_2 \in \mathcal{B}_2 :=  [0,1]$ and
			$\delta := \theta_3 \in \Delta := \mathbb{R}_C$.
		\end{enumerate}
		As in the convex setup, the tested parameters are on the boundary under the null and non-standard asymptotic theory is required.  
\end{enumerate}

While we focus on these examples of link functions in this article, the asymptotic theory presented in the subsequent section is valid for a many other interesting link functions, such as link functions without intercepts, nonlinear functions, and further specifications which prevent a crossing of the VaR and ES forecasts.

\subsection{Asymptotic Theory on the Boundary of the Parameter Space}
\label{sec:AsymptoticTheory}

In this section, we derive the asymptotic theory for the M-estimator\footnote{In order to ensure global convergence of the M-estimator by avoiding local minima, we utilize the implementation of the \texttt{R} package \texttt{esreg} \citep{BayerDimi2019esreg} based on the Iterated Local Search (ILS) meta-heuristic of \cite{Lourenco2003}. See Section 3 of \cite{DimiBayer2019} for further details.}
$\hat \theta_T$, given by 
\begin{align}
	\label{eqn:DefQn}
	\hat \theta_T = \underset{\theta \in \Theta}{\argmin} \; l_T(\theta), \quad \text{ where }  \quad
	l_T(\theta) =\sumt \rho_t \big( Y_{t+h}, g^q_t(\theta) , g^e_t(\theta)  \big).
\end{align}
Classical asymptotic theory for the M-estimator $\hat \theta_T$, as given in \cite{Patton2019}, states that given certain regularity conditions,
\begin{align}
\label{eqn:AsymptoticNormality}
\sqrt{T} \big( \hat \theta_T - \theta^0 \big) \tod \mathcal{N} \big( 0 , \mathcal{T}^{-1} \mathcal{I} \mathcal{T}^{-1} \big),
\end{align}
where
\begin{align}
	&\begin{aligned}
		\mathcal{T} 
		&= - \mathbb{E} \left[ \nabla  g_t^q(\theta^0) \nabla  g_t^q(\theta^0)^\top \left( \mathfrak{g}(g^q_t(\theta^0)) + \frac{\phi'(g^e_t(\theta^0))}{\alpha} \right)  h_t( g_t^q(\theta^0)) \right.  \\
		&\qquad\qquad+ \left. \nabla  g_t^e(\theta^0) \nabla g_t^e(\theta^0)^\top \phi''(g^e_t(\theta^0)) \right], \qquad \text{ and } 
		\label{eqn:T_Matrix}
	\end{aligned} \\
	&\mathcal{I} = \Var \left( T^{-1/2} \sumt \psi_t(\theta^0) \right),
	\label{eqn:I_Matrix}
	\end{align}
	with
	\begin{align}
	\begin{aligned}
		\psi_t(\theta) &=
		\nabla g^q_t(\theta) \left( \mathfrak{g}(g^q_t(\theta)) + \frac{\phi'(g^e_t(\theta))}{\alpha} \right) \left( \mathds{1}_{\{Y_{t+h} \le g^q_t(\theta) \}} - \alpha \right) \\
		& + \nabla g^e_t(\theta)  \phi''(g^e_t(\theta))  \left( g^e_t(\theta) - g^q_t(\theta)+ \frac{1}{\alpha} (g^q_t(\theta) - Y_{t+h}) \mathds{1}_{\{Y_{t+h} \le g^q_t(\theta) \}} \right).
		\label{eqn:Psi_t}
	\end{aligned}
\end{align}
The function $\psi_t(\theta)$ corresponds to the gradient of the loss function $\rho_t(\theta)$ almost surely, i.e.\  on the set $\{ \theta \in \Theta: \; Y_{t+h} \not= g^q_t(\theta) \}$, which has probability one as the distribution $F_t$ is assumed to be absolutely continuous.

The asymptotic normality result in (\ref{eqn:AsymptoticNormality}) crucially depends on the regularity condition that the true parameter $\theta^0$ is in the interior of the parameter space, $\operatorname{int}(\Theta)$.
This condition is violated under the null hypothesis for many interesting specifications of the link functions for the considered encompassing tests, as further outlined in Section \ref{sec:LinkSpecifications}.
\cite{Andrews1999} derives the non-standard asymptotic distribution of the parameter estimates in a general setup, which allows for parameters to be on the boundary and \cite{Andrews2001} extends this result to the asymptotic distribution of the resulting Wald test statistics.

Intuitively, the condition $\theta^0 \in \operatorname{int}(\Theta)$ implies that parameters to all sides (in a neighborhood) of $\theta^0$ are contained in $\Theta$ such that the estimator $\hat \theta_T$ is allowed to \textit{vary} to all sides of $\theta^0$.
The asymptotic normality result in (\ref{eqn:AsymptoticNormality}) formalizes this intuition by quantifying this variation as a limiting normal distribution.
In contrast, if $\theta^0$ is on the boundary of $\Theta$, the estimator $\hat \theta_T$ cannot attain values to all sides of $\theta^0$, as values in some directions are excluded through the boundary.
Consequently, in these cases the asymptotic distribution is more complicated and non-standard, which we formalize through deriving asymptotic theory on the boundary in the following.	
For this, we make the following assumptions.
	\newcounter{AssumptionCounter}
	\begin{assumption}
		\label{assu:AsymptoticTheory}
		$ $
		\begin{enumerate}[label=(\Alph*)]
			\item 
			\label{cond:ParameterSpace}
			The parameter space is given as the product space $\Theta = \mathcal{B}_1 \times \mathcal{B}_2 \times \Delta \times \Psi$, where each of these four spaces is compact and restricted by individual inequality constraints:
			\begin{itemize}
				\item 
				$\mathcal{B}_1 = \big\{ \beta_1 \in \mathbb{R}^{p_1}: \Gamma_{\beta_1} \beta_1 \le r_{\beta_1} \big\}$, where $\Gamma_{\beta_1}$ is a $l_{\beta_1} \times p_1$ matrix and $r_{\beta_1}$ a $l_{\beta_1}$-dimensional vector,
				
				\item 
				$\mathcal{B}_2 = \big\{ \beta_2 \in \mathbb{R}^{p_2}: \Gamma_{\beta_2} \beta_2 \le r_{\beta_2} \big\}$, where $\Gamma_{\beta_2}$ is a $l_{\beta_2} \times p_2$ matrix and $r_{\beta_2}$ a $l_{\beta_2}$-dimensional vector,
								
				\item 
				$\Delta = \big\{ \delta \in \mathbb{R}^{q}: \Gamma_{\delta} \delta \le r_{\delta} \big\}$, where $\Gamma_{\delta}$ is a $l_{\delta} \times q$ matrix and $r_{\delta}$ a $l_{\delta}$-dimensional vector,
				
				\item 
				$\Psi = \big\{ \psi \in \mathbb{R}^{s}: \Gamma_{\psi} \psi \le r_{\psi} \big\}$, where $\Gamma_{\psi}$ is a $l_{\psi} \times s$ matrix and $r_{\psi}$ a $l_{\psi}$-dimensional vector.
			\end{itemize}

			\item
			\label{cond:BetaMixing}
			The process $Z_t$ is stationary and $\beta$-mixing of size $- r/(  r-1)$ for some $ r>1$.

			\item 
			\label{cond:MomentCondition}
			It holds that $\mathbb{E}\big[ \sup_{\theta \in \Theta}  |\rho_t(\theta)|^{2r}  \big] < \infty$ and $\mathbb{E} \left[ \sup_{\theta \in \Theta} ||\psi_t(\theta)||^{2r}  \right] < \infty$  for all $\theta \in \Theta$ and some $\delta > 0$, where $ r>1$ is given in condition \ref{cond:BetaMixing}.\footnote{We state these conditions as high-level moment conditions depending on $\rho_t(\theta)$ and $\psi_t(\theta)$.
			The derivations for primitive moment conditions for the semiparametric models for the VaR and ES for specific choices of the functions $\mathfrak{g}(\cdot)$ and $\phi(\cdot)$ are straight-forward, but the resulting conditions are often rather convoluted, see e.g.\ Appendix A of \cite{DimiBayer2019} and Assumption 2 (C) and (D) of \cite{Patton2019}.}
			
			\item 
			\label{cond:AbsContDistribution}
			The distribution of $Y_{t+h}$ given $\mathcal{F}_t$, denoted by $F_t$, is absolutely continuous with continuous and strictly positive density $h_t$, which is bounded from above almost surely on the whole support of $F_t$ and Lipschitz continuous.
			
			\item 
			\label{cond:FullRankConditionNormality}
			The link functions $g_t^q(\theta)$ and $g_t^e(\theta)$ are $\mathcal{F}_t$-measurable, twice continuously differentiable in $\theta$ on $\operatorname{int}(\Theta)$ almost surely, and directionally differentiable on the boundary of $\Theta$.
			Moreover,  if for some $\theta_1, \theta_2 \in \Theta$, $\mathbb{P} \big( g_t^q(\theta_1) = g_t^q(\theta_2) \cap g_t^e(\theta_1) = g_t^e(\theta_2) \big) = 1$, then $\theta_1 = \theta_2$.

			\item 
			\label{cond:AsyCovPositiveDefinite}
			The matrices $\mathcal{I}$ and $\mathcal{T}$ have full rank.

			\item
			\label{cond:BreadMatrixZeroBlock}
			The matrix-elements of $\mathcal{T}$ governing the dependence of $(\beta_1, \beta_2, \delta)$ and of $\psi$ are zero.
			
			\setcounter{AssumptionCounter}{\value{enumi}}
		\end{enumerate}
	\end{assumption}

	Apart from the conditions \ref{cond:ParameterSpace} and \ref{cond:BreadMatrixZeroBlock}, these assumptions are similar to the ones of \cite{Patton2019} and \cite{DimiSchnaitmann2020}.
	However, as we base our proofs on stochastic equicontinuity and empirical process theory \citep{Andrews1994}, instead of on the approach of \cite{Weiss1991}, some of the conditions differ slightly.
	One main difference is that we assume the slightly stronger dependence condition of $\beta$-mixing (instead of $\alpha$-mixing) in order to show stochastic equicontinuity of the empirical process based on the theory of \cite{Doukhan1995}.
	Notice that the parameter space in condition \ref{cond:ParameterSpace} can  conveniently be expressed through $l$ inequality constraints using an $l \times k$ matrix $ \Gamma_\theta$ and an $l$-dimensional vector $r_\theta$ as\footnote{In fact, $r_\theta = (r_{\beta_1}, r_{\beta_2}, r_{\delta}, r_{\psi})$ and by expressing $\Gamma_\theta$ as a $4 \times 4$ block matrix, the individual blocks $\Gamma_{\beta_1}$,  $\Gamma_{\beta_2}$, $\Gamma_{\delta}$ and $\Gamma_{\psi}$ appear on its \textit{diagonal} with rectangular zero-blocks everywhere else.}
	\begin{align}
		\label{eqn:ParamSpace}
		\Theta = \big\{ \theta \in \mathbb{R}^k: \Gamma_\theta \theta \le r_\theta \big\}.
	\end{align}
	This general formulation allows for flexible product spaces of closed real intervals.

	\begin{theorem}
		\label{thm:GeneralAsymptoticDistribution}
		Suppose Assumption \ref{assu:AsymptoticTheory} holds.
		Then
		\begin{align}
			\label{eqn:AsyDistArgInf}
			\sqrt{T} \big( \hat \theta_T - \theta^0 \big) \tod \hat \lambda, 
			\qquad \text{ where } \qquad
			\hat \lambda = \underset{\lambda \in \Lambda}{\arg \inf}\; (\lambda - Z)^\top \mathcal{T} (\lambda - Z),
		\end{align}
		with $Z= \mathcal{T}^{-1} G$,  $G \sim \mathcal{N}(0, \mathcal{I})$ and $\Lambda = \big\{ \lambda \in \mathbb{R}^k: \Gamma_{\theta}^{(b)} \lambda \le 0 \big\}$, where $\Gamma_{\theta}^{(b)}$ denotes the submatrix of $\Gamma_\theta$ from \eqref{eqn:ParamSpace}, which consists of the rows of $\Gamma_\theta$ for which all inequalities $\Gamma_{\theta}^{(b)} \theta^0 \le r_\theta$ hold as an equality. 
	\end{theorem}	

    The proof of Theorem \ref{thm:GeneralAsymptoticDistribution} verifies the necessary assumptions in \cite{Andrews1999} and \cite{Andrews2001}.\footnote{
	Notice that the notation in \cite{Andrews2001} includes the nuisance parameter $\pi \in \Pi$ which we do not require. 
	Thus, following the comment on p.692 of \cite{Andrews2001}, we simply employ a parameter space $\Pi = \{\pi_0\}$ consisting of a single point $\pi_0$, e.g.\ $\pi_0 = 0$, and suppress the dependency on $\pi$ in the notation.
}
If $\theta^0 \in \operatorname{int}(\Theta)$,  none of the inequalities in (\ref{eqn:ParamSpace}) is binding and $\Lambda = \mathbb{R}^k$.
This implies that $\hat \lambda = Z$ almost surely in (\ref{eqn:AsyDistArgInf}), which results in the classical asymptotic normality result given in (\ref{eqn:AsymptoticNormality}).
In contrast, if $\theta^0$ is on the boundary of $\Theta$, the $\arg \inf$ in (\ref{eqn:AsyDistArgInf}) results in a non-standard asymptotic distribution of the stabilizing transformation $\sqrt{T} \big( \hat \theta_T - \theta^0 \big)$.

\subsubsection*{Subvector Inference}

In the notation of the subvector decomposition of $\theta$ in (\ref{eqn:PartitioningThetaParameter}), we only test parametric restrictions for the subvector $\beta_1$, which might be substantially smaller than $\theta$.
Thus, the formulation of the arg\,inf in (\ref{eqn:AsyDistArgInf}) might be unnecessarily complex in these situations.
To address this issue, we derive inference for the subvector $\beta = (\beta_1, \beta_2)$ of $\theta$ by following the general approach of \cite{Andrews1999, Andrews2001}.
In some instances, this considerably simplifies the solution of the arg\,inf in (\ref{eqn:AsyDistArgInf}).

For this, we define the subvector $\gamma := (\beta, \delta) =  (\beta_1, \beta_2, \delta)$, which contains all parameters in $\theta$ but $\psi$, with the intuition that $\psi$ does not have any influence on the asymptotic distribution of $\gamma$ through the nullity restrictions on $\mathcal{T}$ imposed in condition \ref{cond:BreadMatrixZeroBlock} in Assumption \ref{assu:AsymptoticTheory}.
We define the following quantities for the subvectors $\beta$ and $\gamma$,
\begin{align}
	Z_\gamma := {\mathcal{T}_\gamma}^{-1} G_\gamma, \quad
	Z_\beta := H Z_\gamma, \quad \text{ with } \quad
	H := [I_p, \mathbf{0}_{p\times q}],
\end{align}
where $ {\mathcal{T}_\gamma}$ denotes the upper-left $(p+q) \times (p+q)$ submatrix of $\mathcal{T}$ and $G_\gamma$ the upper $(p+q)$-dimensional subvector of $G$.
The following theorem states the asymptotic distribution of the subvector $\beta$.
\begin{theorem}
	\label{thm:InferenceBeta}
	Given Assumption \ref{assu:AsymptoticTheory}, it holds that
	\begin{align}
		\sqrt{T} \big( \hat \beta_T - \beta^0 \big) \tod \hat \lambda_\beta,
	\end{align}
	where 
	\begin{align}
		\label{eqn:QuadProgrammingBeta}
		\hat \lambda_\beta = \underset{\lambda_\beta \in \Lambda_\beta}{\arg \inf}\; (\lambda_\beta - Z_\beta)^\top \big( H {\mathcal{T}_\gamma}^{-1} H^\top \big)^{-1} (\lambda_\beta - Z_\beta),
	\end{align}
	and $\Lambda_\beta = \big\{ \lambda_\beta \in \mathbb{R}^p: \Gamma_{\beta}^{(b)} \lambda _\beta\le 0 \big\}$.
	The matrix $\Gamma_{\beta}^{(b)}$ denotes the sub-matrix of $\Gamma_{\beta}$, which consists of the rows of $\Gamma_{\beta}$ for which the inequality $\Gamma_{\beta} \beta^0 \le r_\beta$ holds as an equality.
\end{theorem}
Theorem \ref{thm:InferenceBeta} shows that the asymptotic distribution of $\beta$ is entirely unaffected by the parameter $\psi$. In contrast, the subvector $\delta$ (which is contained in $\gamma$) influences the asymptotic distribution of $\beta$ through the weighting matrix in the quadratic programming problem in (\ref{eqn:QuadProgrammingBeta}), even though $\delta$ itself is not on the boundary of the parameter space.

While closed-form representations for the distribution of $\hat \lambda_\beta$ (and of $\hat \lambda$) are only available in special cases \citep{Andrews1999}, we can conveniently simulate from its distribution in a straight-forward fashion by solving a quadratic programming problem.
For this, notice that the minimization problem in (\ref{eqn:QuadProgrammingBeta}) is equivalent to solving
\begin{align}
	\label{eqn:QuadProgrammingBetaImplementation}
	\min_{\lambda_\beta \in \mathbb{R}^p} \frac{1}{2} \lambda_\beta^\top \big( H {\mathcal{T}_\gamma}^{-1} H^\top \big)^{-1} \lambda_\beta - Z_\beta^\top \big( H {\mathcal{T}_\gamma}^{-1} H^\top \big)^{-1} \lambda_\beta 
	\quad \text{ subject to } \quad \Gamma_\beta^{(b)} \lambda_\beta \le 0,
\end{align}
where $\Gamma_\beta^{(b)}$ is given as in Theorem \ref{thm:InferenceBeta} and specifies the binding inequality restrictions of $\Lambda_\beta$.
Consequently, we can draw samples from the Gaussian random variable $G_\gamma$, and for each sampled value, we solve the quadratic programming problem given in (\ref{eqn:QuadProgrammingBetaImplementation}).
The respective solutions then form a sample of the random variable $\hat \lambda_\beta$, whose distribution is asymptotically equivalent to the one of $\sqrt{T} \big( \hat \beta_T - \beta^0 \big)$.

\subsubsection*{The Wald Test Statistic}

We now consider a Wald test for the null hypothesis $\mathbb{H}_0: \beta_1 = \beta_{1 \ast}$ for some $\beta_{1 \ast} \in \mathcal{B}_1$, which may or may not be on the boundary of $\mathcal{B}_1$.
We define the Wald test statistic for the null hypothesis $\mathbb{H}_0: \beta_1 = \beta_{1\ast}$ as
\begin{align}
	\label{eqn:WaldTestStatistic}
	W_T = T \big( \hat \beta_1 - \beta_{1 \ast} \big)^\top \, \hat V_T^{-1} \, \big( \hat \beta_1 - \beta_{1 \ast} \big),
\end{align}
with weighting matrix $\hat V_T^{-1}$, given by
\begin{align}
\label{eqn:WeightingMatrixWaldTest}
\hat V_T := H_1 \hat{\mathcal{T}}_{T \gamma}^{-1}  \hat{\mathcal{I}}_{T \gamma}  \hat{\mathcal{T}}_{T \gamma}^{-1} H_1^\top,
\end{align}  
where $H_1 := [I_p, \mathbf{0}_{p \times q}]$, and where $\hat{\mathcal{T}}_{T \gamma}$ and $\hat{\mathcal{I}}_{T \gamma}$ are the upper left $(p+q) \times (p+q)$ submatrices of $\hat{\mathcal{T}}_{T}$ and $\hat{\mathcal{I}}_{T}$, respectively, which are consistent estimators for the matrices $\mathcal{T}$ and $\mathcal{I}$.
For the matrix $\mathcal{T}$, we use the estimator
\begin{align}\label{eqn:est_T_mat}
	&\begin{aligned}
	\hat{\mathcal{T}}_T
	&= -  \frac{1}{T} \sumt \left( \nabla  g_t^q(\hat \theta_T) \nabla  g_t^q(\hat \theta_T)^\top \left( \mathfrak{g}(g^q_t(\hat \theta_T)) + \frac{\phi'(g^e_t(\hat \theta_T))}{\alpha} \right) \frac{1}{2c_T} \mathds{1}_{\{| Y_{t+h} - g_t^q(\hat \theta_T)| \le c_T\}} \right.  \\
	&\qquad\qquad\qquad+ \left. \nabla  g_t^e(\hat \theta_T) \nabla g_t^e(\hat \theta_T)^\top \phi''(g^e_t(\hat \theta_T)) \right),
	\end{aligned}
\end{align}
where the bandwidth $c_T$ satisfies $c_T = o(1)$ and $c_T^{-1} = o(T^{1/2})$.
In the specification of $\hat{\mathcal{T}}_{T}$, the term $\mathds{1}_{\{| Y_{t+h} - g_t^q(\hat \theta_T))| \le c_T \}}/(2 c_T)$ is a nonparametric estimator of the conditional density $h_t(g^q_t(\theta^0))$, which is also employed in \cite{Engle2004} and  \cite{Patton2019}.

As we allow for multi-step ahead (aggregate) forecasts in this treatment, we employ a HAC estimator \citep{NeweyWest1987, Andrews1991} for the matrix $\mathcal{I}$, 
\begin{align}
	\label{eqn:HACEstimation}
   \hat{\mathcal{I}}_{T} &= \widehat{\Omega}_{T,0} + \sum_{j=1}^{m_T} z(j,m_T) \big( \widehat{\Omega}_{T,j} + \widehat{\Omega}_{T,j}^\top \big),  \; \text{ where } \;
   \widehat{\Omega}_{T,j} =  \frac{1}{T} \sumtj \psi_t(\hat \theta_T)  \psi_{t-j}^\top(\hat \theta_T),
\end{align}
based on some weight functions $z(j,m) \to 1$ and the bound (or bandwidth) $m_T = o(T^{1/4})$.
Furthermore, $\psi_t(\theta)$ is given in (\ref{eqn:Psi_t}) and we define $\mathfrak{T}_j := \{ t \in \mathbb{N}: S+j \le t \le S+T-1 \}$ for all $j \ge 0$.
As the functions $\psi_t( \theta)$ are not continuous in $\theta$, we generalize the consistency proofs of the HAC estimator in \cite{NeweyWest1987} to nonsmooth objective functions in Lemma \ref{lemma:ConsistencyHAC} in the supplementary material. 
For the asymptotic distribution of the Wald test statistic, we impose the following assumptions.   

\begin{assumption}
	\label{assu:AsymptoticTheory2}
	$ $
	\begin{enumerate}[label=(\Alph*)]
		\setcounter{enumi}{\value{AssumptionCounter}}
		
		\item 
		\label{cond:HACWeights}
		$m_T \to \infty$ such that $m_T = o(T^{1/4})$ and $z(j,m) \to 1$ as $m \to \infty$.

		\item 
		\label{cond:DensityBandwidth}
		$c_T = o(1)$ and $c_T^{-1} = o(T^{1/2})$.
		
		\item 
		\label{cond:HACMoments}
		The functions $g_t^q(\theta)$ and $g_t^e(\theta)$ are three times continuously differentiable (in $\theta$) and the following moments are finite, 
		$\mathbb{E} \left[ \sup_{\tilde{\theta} \in U(\theta, \delta) } \left|\left|  \nabla_\theta \tilde A_t(\theta) \right| \right|^{2  r}  \right]$,  		
		$\mathbb{E} \left[ \sup_{\tilde{\theta} \in U(\theta, \delta) } \left|\left|  \nabla_\theta \tilde B_t(\theta) \right| \right|^{2  r} \right]$,
		\\
		$\mathbb{E} \left[ \sup_{\tilde{\theta} \in U(\theta, \delta) } \left| \tilde A_t(\tilde \theta) \right|^{2r}  \times  \sup_{\tilde{\theta} \in U(\theta, \delta) } \left|\left| \nabla_\theta g_t^q(\tilde \theta) h_t(g_t^q(\tilde \theta)) \right| \right|^{2r}  \right]$, \\
		and
		$\mathbb{E} \left[ \sup_{\theta \in \Theta} \left| \left|  \psi_t( \theta) \right| \right|^{2(r+\delta)} \right]$, for some $\delta > 0$,
		where $\tilde A_t(\theta)$ and  $\tilde B_t(\theta)$ are given in (\ref{eqn:DefTildeBt}) and (\ref{eqn:DefTildeAt}) in the supplementary material.
	\end{enumerate}
\end{assumption}
Conditions \ref{cond:HACWeights} and \ref{cond:DensityBandwidth} are standard in the literature on HAC estimators and estimating the conditional density, see e.g., \cite{NeweyWest1987}, \cite{Engle2004} and \cite{Patton2019}.
The strengthened moment conditions \ref{cond:HACMoments} are required to establish stochastic equicontinuity of the discontinuous function $\frac{1}{T} \sumtj \psi_t(\theta) \psi_{t-j}^\top (\theta)$ for consistency of the HAC estimator.

\begin{theorem}
	\label{thm:WaldTestAsymptoticDistribution}
	Suppose Assumption \ref{assu:AsymptoticTheory} and Assumption \ref{assu:AsymptoticTheory2} hold.
	Then
	\begin{align}
		\label{eqn:WaltStatisticAsyDistribution}
		W_T \tod W := \hat \lambda_{\beta_1}^\top V^{-1} \hat \lambda_{\beta_1},
	\end{align}
	where $V$ denotes the probability limit of $\hat V_T$ and $\hat \lambda_{\beta_1}$ is the upper $p_1$-dimensional subvector of $ \hat \lambda_{\beta}$, given in Theorem \ref{thm:InferenceBeta}.
\end{theorem}	

Using the simulation procedure for the distribution of $\hat \lambda_{\beta}$ described after Theorem \ref{thm:InferenceBeta}, we can easily simulate draws from $\hat \lambda_{\beta_1}$ and consequently from the distribution of $W$ by using the formula in (\ref{eqn:WaltStatisticAsyDistribution}).
Hence, we obtain simulated, asymptotic critical values for the Wald test statistic.

We further use a variant of the HAC estimator \citep{NeweyWest1987, Andrews1991}, which is specifically designed for the semiparametric VaR and ES models.
For most classical HAC estimators, estimation of the contemporaneous variance $\mathbb{E} \big[ \psi_t(\theta^0)  \psi_{t}^\top(\theta^0) \big]$ is straight-forward by employing a sample counterpart.
The major challenge in consistently estimating the matrix $\mathcal{I}$ in (\ref{eqn:I_Matrix}) is then the inclusion of the (sample) autocovariances $\mathbb{E} \big[ \psi_t( \theta^0)  \psi_{t-j}^\top( \theta^0) \big]$ such that the resulting estimator is positive definite.

However, for the VaR and ES, and especially for extreme quantile levels, estimation of the contemporaneous variance $\mathbb{E} \big[ \psi_t( \theta^0)  \psi_{t}^\top( \theta^0) \big]$ is cumbersome in itself as it depends on the \textit{conditional truncated variance} $\operatorname{Var}_t( Y_{t+h} | Y_{t+h} \le g_t^q(\theta^0))$, see e.g.\ \cite{DimiBayer2019}.
For this, we employ the \textit{scl-sp} estimator of \cite{DimiBayer2019}, which is based on the regularizing assumption that the quantile residuals $u_{t+h}^q = Y_{t+h} - g_t^q(\theta^0)$ follow a location-scale model, conditional on the employed covariates.
Imposing a location-scale model might cause some misspecification in the estimation, but it allows to use all observations to estimate a conditional variance, and then obtain the conditional truncated variance through a transformation formula for location-scale models.
We obtain this estimator by replacing the outer product estimator of the contemporaneous variance by the \textit{scl-sp} estimator,
\begin{align}
	\label{eqn:HAC-SCLSP-Estimation}
	\tilde{\mathcal{I}}_{T} &= \widetilde{\Omega}_{T,0} + \sum_{j=1}^{m_T} z(j,m_T) \big( \widehat{\Omega}_{T,j} + \widehat{\Omega}_{T,j}^\top \big),
\end{align}
where $\widetilde{\Omega}_{T,0}$ denotes the \textit{scl-sp} estimator of \cite{DimiBayer2019}.

Even though the parametric link functions in (\ref{eqn:ParametricModelsEncompassing}) depend explicitly on the forecasts $\bq$ and $\be$, it is important to note that the  asymptotic theory of this section also holds for general semiparametric models for the VaR and ES in the sense of \cite{Patton2019}.
Consequently, the asymptotic theory and the proposed Wald test can further be employed for testing (the nullity) of coefficients in the dynamic models of \cite{Taylor2019} and \cite{Patton2019}, which are on the boundary of the parameter space under the null hypothesis.
Furthermore, the \textit{strict} ES encompassing test of \cite{DimiSchnaitmann2020} allows for testing encompassing of ES forecasts without their accompanying VaR forecasts, which potentially introduces model misspecification in the parametric models.
The asymptotic theory for the M-estimator presented here can easily be adapted to the misspecified case by replacing the matrices $\mathcal{T}$ and $\mathcal{I}$ with their misspecification-robust counterparts of \cite{DimiSchnaitmann2020}, and by replacing the respective steps in the proof of Theorem \ref{thm:GeneralAsymptoticDistribution}.

\section{Simulations}
\label{sec:Simulations}

In this section, we evaluate the empirical properties of the encompassing tests based on the three different link functions specified in Section \ref{sec:LinkSpecifications}, and on the asymptotic theory of Section \ref{sec:AsymptoticTheory}.
Section \ref{sec:SimAsyDist} numerically illustrates the effect testing on the boundary has on the asymptotic distribution of the parameters.
Subsequently, we analyze the size and power properties of the encompassing tests in Section \ref{sec:onestep_forecasts} for one-step ahead forecasts and in \Cref{sec:multistep_forecasts} for multi-step ahead and aggregate forecasts.

\subsection{The Asymptotic Distribution on the Boundary}
\label{sec:SimAsyDist}

We illustrate how true parameters on the boundary of the parameter space affect the asymptotic distribution of the M-estimator through simulations.
For this, we simulate data according to the standard GARCH model with Gaussian innovations described in (\ref{eqn:GARCHModel}) in Section \ref{sec:onestep_forecasts} with an out-of-sample window length of $T=2500$.
We estimate the parameters of the three considered link functions for the joint encompassing test that tests whether forecasts stemming from the (true) GARCH model encompass forecasts from the GJR-GARCH model given in (\ref{eqn:GJRGARCHModel}).

Figure \ref{fig:IllustrationBoundary} illustrates the distribution of the parameter estimates by plotting histograms over 10000 simulation replications for the intercept and slope parameters of the respective ES link functions $g_t^e$, whose true values equal zero and one respectively throughout all link functions.
For the (unrestricted) linear link function, all true parameters are in the interior of the parameter space and we find that the histograms for both parameters closely approximate the asymptotic normal distribution, derived and employed by \cite{Patton2019} and \cite{DimiSchnaitmann2020}.
In contrast, for the convex and no-crossing link functions, the slope parameter is bounded between zero and one, i.e.\ its true value of one is on the boundary of the parameter space.
This results in the non-standard distributions illustrated by the histograms for the slope parameters in the second and third plot in the lower row of Figure \ref{fig:IllustrationBoundary}.
The histograms approximate the asymptotic distribution consisting of a mixture of a point mass at one and a half-normal distribution, which is considerably different from asymptotic normality.
This behavior directly carries over to the resulting asymptotic distributions of the Wald test statistics which substantiates the necessity of the non-standard asymptotic theory on the boundary presented in Section \ref{sec:AsymptoticTheory}.

\begin{figure}[tb]
	\includegraphics[width=\linewidth]{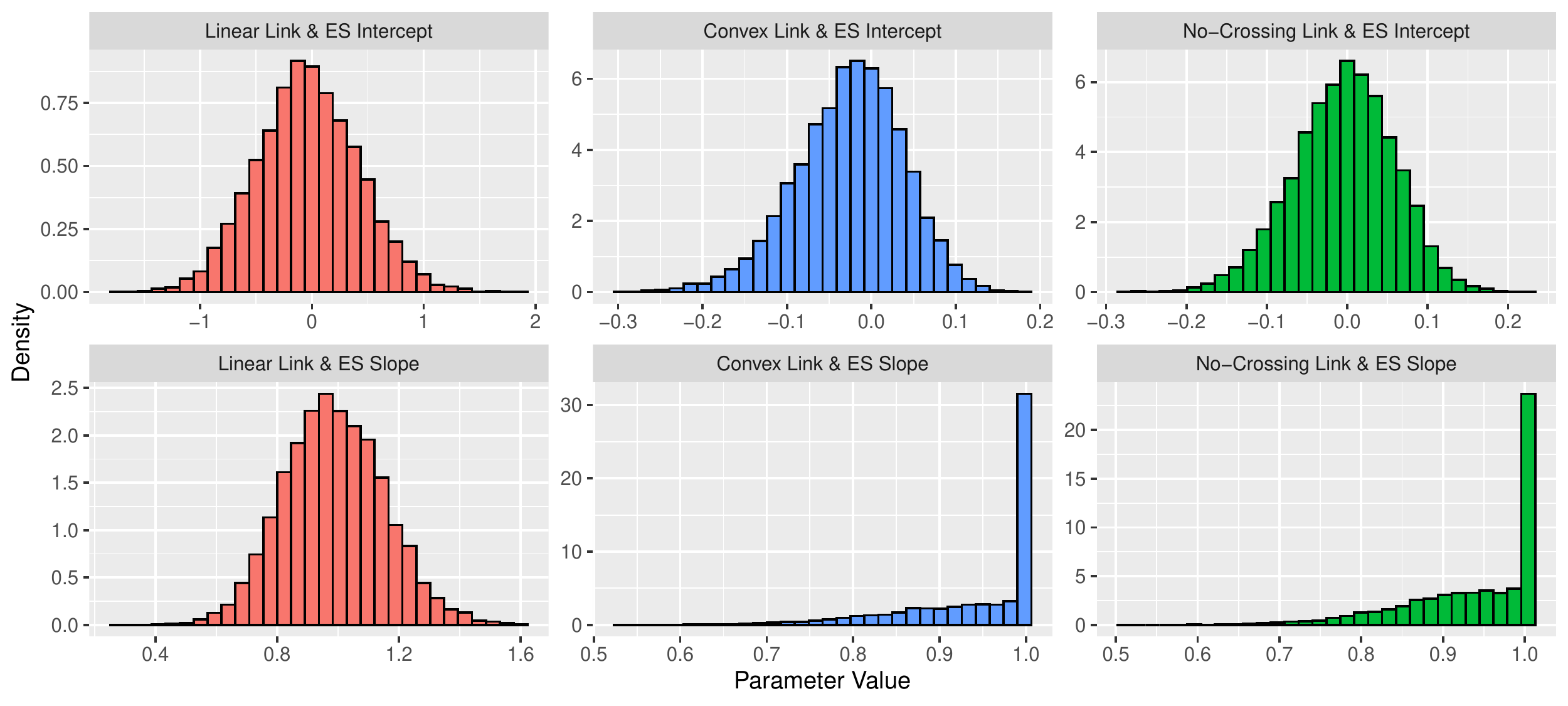}
	\caption{Illustration of the (asymptotic) distributions of the parameter estimates of the ES-specific intercept and slope parameter corresponding to the first ES forecast $\hat{e}_{1,t}$ for the three considered link functions.}
	\label{fig:IllustrationBoundary}
\end{figure}

While this behavior is not unexpected for the parameters on the boundary, the asymptotic distribution of the intercept parameters, which themselves are in the interior of the parameter space, is also affected due to the joint estimation.
For instance, we observe a slight skewness in the distribution of the intercept parameter of the convex link function contrasting the Gaussian distribution of the linear intercept parameter.

\subsection{One-Step Ahead Forecasts}
\label{sec:onestep_forecasts}

In this section, we investigate the empirical performance of our new encompassing tests for one-step ahead forecasts.
For this, we consider encompassing of VaR and ES forecasts stemming from a standard GARCH and a GJR-GARCH model \citep{Bollerslev1986, Glosten1993}, which are given by 
$r_{j,t+1} = \sigma_{j,t+1} u_{t+1}$,
for $j=1,2$, where the two distinct volatility specifications are given by
\begin{align}
	\label{eqn:GARCHModel}
	 \sigma_{1,t+1}^2 &= 0.04 + 0.1 r_{1,t}^2 + 0.85  \sigma_{1,t}^2, \qquad \text{and}  \\
	 \sigma_{2,t+1}^2 &= 0.04 + \big(0.05 + 0.1 \cdot \mathds{1}_{\{ r_{2,t} \le 0 \}} \big) r_{2,t}^2 + 0.8  \sigma_{2,t}^2.
	\label{eqn:GJRGARCHModel}
\end{align}
Furthermore, we employ two different residual distributions,
\begin{align}
	\label{eqn:GARCHresiduals}
	(a) \; u_{t+1} \stackrel{iid}{\sim} \mathcal{N}(0,1)  \qquad \text{ and } \qquad
	(b) \; u_{t+1} \stackrel{iid}{\sim} t(0.8,5),
\end{align}
where the latter denotes a skewed $t$-distribution, parameterized as in \cite{fernandez1998bayesian} and \cite{giot2003value}, with zero mean, unit variance, a skewness parameter of $0.8$ and $5$ degrees of freedom. 
For the two GARCH models paired with the two residual distributions, optimal one-step ahead VaR and ES forecasts are given by 
$\hat q_{j,t} = z_\alpha \sigma_{j,t+1}$ and $\hat e_{j,t} = \xi_\alpha  \sigma_{j,t+1}$
for $j=1,2$, where $z_\alpha$ and $\xi_\alpha$ are the $\alpha$-quantile and $\alpha$-ES of the standard normal and the skewed $t$-distribution, respectively.\footnote{Regarding the time index, notice that $\hat q_{j,t}$ and $\hat e_{j,t}$ represent $\mathcal{F}_t$-measurable forecasts for the return $r_{j,t+1}$, while $\sigma_{j,t+1}$ is equivalently based on time $t$ information and corresponds to the conditional volatility of $r_{j,t+1}$.}
For both distributions, we simulate $Y_{t+1} = r_{t+1}  = \big( (1-\pi)  \sigma_{1,t+1} + \pi  \sigma_{2,t+1} \big) u_{t+1}$ for 11 equally spaced values of $\pi \in [0,1]$, where $u_{t+1}$ is given as in ($a$) and ($b$) in \eqref{eqn:GARCHresiduals}. 

We consider encompassing tests comparing the respective GARCH and GJR-GARCH volatility specifications,
where we analyze the models based on \textit{Gaussian} and \textit{$t$-distributed} residuals in separate simulation setups.
For each forecast pair, we test two null hypotheses: the first tests whether the first forecast encompasses the second, indicated by $\mathbb{H}_0^{(1)}$, whereas the second tests the reverse, i.e.\ that forecast two encompasses forecast one, indicated by $\mathbb{H}_0^{(2)}$.
These two null hypotheses correspond to the cases $\pi = 0$ and $\pi=1$ in the simulation design above.
For all intermediate values of $\pi \in (0,1)$, the returns are generated as linear combinations of the models, and both null hypotheses should be rejected.
For both encompassing tests, we employ the \textit{scl-sp} covariance estimator of \cite{DimiBayer2019} described in Section \ref{sec:AsymptoticTheory}.\footnote{Section \ref{sec:cov_choice} in the supplemental material shows that the results for employing a HAC estimator are qualitatively equivalent for one-step ahead forecasts.}
All following results are based on 2000 Monte Carlo replications.

\begin{table}[t]
	\footnotesize
	\centering
	\caption{Empirical Test Sizes for One-Step Ahead Forecasts.}
	\label{tab:Size_GARCH_1}
	\begin{tabularx}{0.95\linewidth}{ll @{\hspace{0.2cm}} rrl @{\hspace{0.1cm}} rrl @{\hspace{0.5cm}}  rrl @{\hspace{0.1cm}} rr}
		\hline\hline
		\addlinespace
 		& & \multicolumn{2}{c}{$\mathbb{H}_0^{(1)}$} & & \multicolumn{2}{c}{$\mathbb{H}_0^{(2)}$} & &  \multicolumn{2}{c}{$\mathbb{H}_0^{(1)}$} & & \multicolumn{2}{c}{$\mathbb{H}_0^{(2)}$} \\
		\cmidrule(lr){3-4} \cmidrule(lr){6-7} \cmidrule(lr){9-10} \cmidrule(lr){12-13}
		& &  VaR ES &  Aux ES  & &  VaR ES &  Aux ES & &   VaR ES &  Aux ES  & &  VaR ES &  Aux ES  \\ 
		\midrule
   		& & \multicolumn{11}{c}{Linear link function} \\
		\midrule
		$T$ & &\multicolumn{5}{c}{Normal innovations} &   &\multicolumn{5}{c}{Skewed-t innovations}\\
		\cmidrule(lr){3-7} \cmidrule(lr){9-13}
		$250$  &        & 21.45  & 11.20  &        & 19.65  & 10.90 &        & 31.30  & 16.10  &        & 31.15  & 16.30 \\
$500$  &        & 16.60  &  8.60  &        & 15.30  &  9.10 &        & 25.40  & 10.60  &        & 24.25  & 10.45 \\
$1000$ &        & 12.95  &  6.70  &        & 11.80  &  7.05 &        & 22.55  &  7.35  &        & 20.25  &  8.30 \\
$2500$ &        & 11.35  &  6.15  &        &  9.70  &  5.05 &        & 16.80  &  5.45  &        & 15.65  &  4.90 \\
$5000$ &        &  8.65  &  5.00  &        &  8.45  &  5.30 &        & 14.35  &  5.10  &        & 15.00  &  5.40 \\
 
		\midrule
		& & \multicolumn{11}{c}{Convex link function} \\
		\midrule
		$T$ & &\multicolumn{5}{c}{Normal innovations} &   &\multicolumn{5}{c}{Skewed-t innovations}\\
		\cmidrule(lr){3-7} \cmidrule(lr){9-13}
		$250$  &        & 10.35  & 8.70   &        & 7.35   & 6.20  &        & 13.30  & 10.45  &        & 10.50  & 8.10  \\
$500$  &        &  8.10  & 7.50   &        & 5.35   & 5.35  &        & 11.16  &  8.91  &        &  7.80  & 6.90  \\
$1000$ &        &  7.53  & 6.82   &        & 4.75   & 4.40  &        &  9.26  &  7.71  &        &  6.36  & 4.56  \\
$2500$ &        &  5.66  & 5.66   &        & 4.10   & 3.90  &        &  7.14  &  5.78  &        &  5.21  & 3.76  \\
$5000$ &        &  7.02  & 6.77   &        & 4.65   & 4.00  &        &  5.56  &  4.31  &        &  6.16  & 3.91  \\
 
		\midrule
		& & \multicolumn{11}{c}{No-crossing link function} \\
		\midrule
		$T$ & &\multicolumn{5}{c}{Normal innovations} &   &\multicolumn{5}{c}{Skewed-t innovations}\\
		\cmidrule(lr){3-7} \cmidrule(lr){9-13}
		$250$  &        & 3.90   & 9.15   &        & 2.65   & 5.20  &        & 7.45   & 10.95  &        &  8.45  & 5.25  \\
$500$  &        & 2.75   & 8.90   &        & 4.95   & 4.75  &        & 7.95   & 10.80  &        & 10.51  & 4.50  \\
$1000$ &        & 2.75   & 9.05   &        & 7.30   & 3.60  &        & 9.70   &  9.35  &        & 12.76  & 3.90  \\
$2500$ &        & 4.55   & 6.60   &        & 8.55   & 3.85  &        & 9.76   &  7.56  &        &  9.80  & 3.05  \\
$5000$ &        & 4.96   & 6.76   &        & 7.35   & 3.90  &        & 8.47   &  5.96  &        &  9.35  & 3.75  \\

		\addlinespace
		\hline
		\hline 
		\addlinespace
		\multicolumn{13}{p{.93\linewidth}}{\textit{Notes:} This table reports the empirical sizes of the encompassing tests with a nominal size of $5\%$ for one-step ahead forecasts.
		For this, we consider the two DGPs based on different GARCH specifications, the three link functions, the joint VaR and ES (VaR ES) and auxiliary ES (Aux ES) tests and both encompassing null hypotheses.
		The columns denoted by ``Normal innovations'' contain results for the GARCH(1,1) and GJR-GARCH(1,1) in \eqref{eqn:GARCHModel} and \eqref{eqn:GJRGARCHModel} with normal innovations, whereas those labeled ``Skewed-t innovations'' report results for the skewed-$t$ distributed innovations.}
	\end{tabularx}
\end{table}

Table \ref{tab:Size_GARCH_1} reports the empirical test sizes of the joint VaR and ES and the auxiliary ES encompassing tests based on the three link functions described in Section \ref{sec:LinkSpecifications} for a nominal size of $5\%$.
For this, we consider the two GARCH specifications described in \eqref{eqn:GARCHModel} and \eqref{eqn:GJRGARCHModel}  for various out-of-sample sizes ranging from $T=250$ to $T=5000$.
We find that the tests based on the convex and no-crossing link functions outperform the ones build on the linear link function, especially for smaller out-of-sample sizes:
the tests based on the linear link function are in some instances severely oversized, while the other two link functions exhibit empirical sizes generally below $10\%$, even for the smallest of the considered sample sizes.
Note for this that a sample size of $T=250$ is considered to be very small for VaR and ES forecasts at a probability level of $\alpha = 2.5\%$, as this corresponds to only six VaR violations on average.
This result can be explained by the reduced number of estimated parameters for both, the convex and no-crossing link functions, and by the theoretically appealing property of excluding VaR and ES crossings for the no-crossing specification.

We further find that the auxiliary ES encompassing test generally exhibits more accurate (smaller) sizes than the  joint VaR and ES test throughout almost all considered designs.
This behavior is particularly evident for the process with skewed-$t$ innovations. 
As the joint test includes testing of the quantile parameters, the asymptotic covariance matrix additionally contains the density quantile function $h_t( g_t^q(\theta^0))$ in \eqref{eqn:T_Matrix}, which is particularly challenging to estimate for small probability levels (see e.g.\ \citealp{Koenker1978, Koenker2005book, DimiBayer2019}).

\begin{figure}[p!]
	\includegraphics[width=\linewidth]{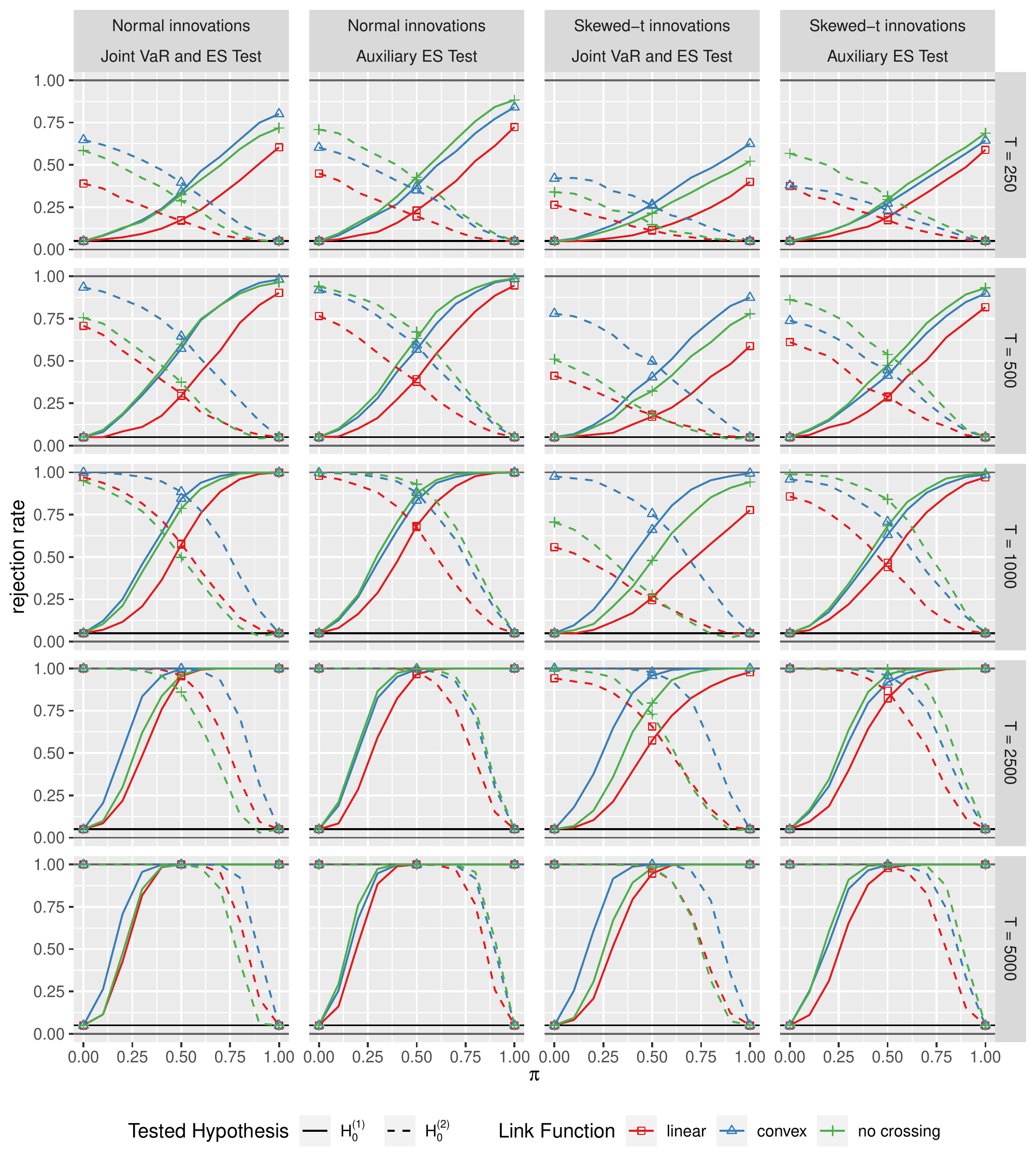}
	\caption{
	This figure shows size-adjusted power curves for the joint VaR and ES and the auxiliary ES encompassing tests with a nominal size of $5\%$.
	The employed link functions are indicated with the line color and symbol shape while the line type refers to the tested null hypothesis.
	The plot rows depict different sample sizes while the plot columns show results for the two innovation distributions described in \eqref{eqn:GARCHModel} - \eqref{eqn:GARCHresiduals} and for the joint and the auxiliary tests.
	An ideal test exhibits a rejection rate of $5\%$ for $\pi=0$ and for $\mathbb{H}_0^{(1)}$ (and inversely for $\pi=1$ and $\mathbb{H}_0^{(2)}$) and as sharply increasing rejection rates as possible for increasing (decreasing) values of $\pi$.}
	\label{fig:sim_VaRES_GARCH}
\end{figure}

Figure \ref{fig:sim_VaRES_GARCH} shows size-adjusted power curves for the joint VaR and ES and the auxiliary ES tests based on the three link functions for a nominal significance level of $5\%$ and for the various settings described above.\footnote{Figure \ref{fig:sim_VaRES_GARCH_rawpower} in the supplementary material shows the corresponding raw power of the tests.}
For computing the size-adjusted power, we follow the approach of \cite{DavidsonMacKinnon1998}.
For an increasing degree of misspecification through $\pi$, we find increasing power throughout all considered tests and processes.
Both, the convex and no-crossing link function specifications exhibit better (size-adjusted) power than the linear link function throughout all considered processes, sample sizes and values of $\pi$.
While the convex link function exhibits a slightly superior performance for the joint test, the no-crossing link function performs slightly better for the auxiliary ES test.
We provide simulation results for two additional processes outside the location-scale family in Section \ref{sec:add_simulation} in the supplementary material, where the results for these forecasts are comparable to those obtained here.

\subsection{Multi-Step Ahead and Aggregate Forecasts}
\label{sec:multistep_forecasts}

In this section, we consider multi-step ahead and multi-step aggregate forecasts for the VaR and ES.
For any $h > 1$, we set $Y_{j,t+h} = r_{j,t+h}$ for multi-step \textit{ahead} forecasts, and $Y_{j,t+h} = \sum_{s=1}^h r_{j,t+s}$ for multi-step \textit{aggregate} forecasts, where the returns $r_{j,t+h}$ are simulated from the respective GARCH specifications in \eqref{eqn:GARCHModel} - \eqref{eqn:GARCHresiduals} for $j=1,2$.
In order to simulate returns which follow a (probabilistic) convex combination of these two processes, we simulate Bernoulli draws $\pi_{t+h} \sim \operatorname{Bern}(\pi)$ for 11 equally spaced values of $\pi \in [0,1]$, and let $Y_{t+h}  = (1-\pi_{t+h})  Y_{1,t+h}  + \pi_{t+h} Y_{2,t+h}$. 

\cite{WongSo2003} and \cite{Loennbark2016} among others illustrate that even though the conditional variance of multi-step ahead (aggregate) forecasts for (quadratic) GARCH models is easily tractable, the entire conditional distribution is not.
This implies that multi-step ahead (aggregate) VaR and ES forecasts cannot be obtained equivalently to one-step ahead forecasts by simply multiplying their conditional multi-step ahead (aggregate) volatilities with the quantile or ES of the residual distribution.
Consequently, we employ a simulation method proposed by \cite{WongSo2003} which yields very accurate approximations of the true VaR and ES forecasts:
for all out-of-sample time points $t \in \mathfrak{T}$, we simulate $R=10000$ sample paths from the respective GARCH model for $h$ days into the future and in order to obtain multi-period ahead (aggregate)  VaR and ES forecasts, we (point-wisely) take the empirical quantile and ES over the $R$ sample paths of the simulated $h$-period ahead (aggregated)  returns.

Here, we restrict attention to the DGP based on Gaussian residuals, the convex link function and on the joint VaR and ES encompassing test as the $t$-distributed residuals and the auxiliary tests perform comparably in the previous section.
However, we consider $h$-step ahead and $h$-step aggregate VaR and ES forecasts with forecasting horizons of $h=1,2,5$ and 10 days. This allows to investigate the properties of the test for increasing forecast horizons $h$.
We employ a HAC estimator with the embedded \textit{scl-sp} estimator of \cite{DimiBayer2019} for the contemporaneous variance as described in Section \ref{sec:AsymptoticTheory}, as in particular the multi-period aggregate forecasts exhibit a correlated behavior due to their inherently overlapping nature.
In Section \ref{sec:cov_choice} in the supplementary material, we discuss four different covariance estimators and show that the HAC estimator augmented with the \textit{scl-sp} estimator performs best.

\begin{table}[th!]
	\footnotesize
	\centering
	\caption{Empirical Test Sizes for Multi-Step Ahead and Aggregate Forecasts.}
	\label{tab:Multistep}
	\begin{tabularx}{0.75\linewidth}{l l @{\hspace{0.5cm}} rrrr l @{\hspace{1cm}} rrrr }
		\hline
		\hline
		\addlinespace
		& & \multicolumn{4}{c}{$\mathbb{H}_0^{(1)}$} & & \multicolumn{4}{c}{$\mathbb{H}_0^{(2)}$} \\
		\cmidrule(lr){3-6} \cmidrule(lr){8-11} 
		$h$ & &  1 &  2  & 5 & 10 &   &  1 &  2  & 5 & 10  \\ 
		\midrule
		$T$ & &\multicolumn{9}{c}{$h$-step ahead forecasts}\\
		\cmidrule(lr){3-11}
		$250$  &        & 10.85  & 11.87  & 10.16  &  7.95  &        & 9.17   & 9.18   & 8.27   & 6.35  \\
$500$  &        &  8.41  &  8.91  & 13.67  & 11.43  &        & 7.01   & 6.12   & 7.59   & 8.43  \\
$1000$ &        &  6.83  &  6.70  & 10.87  & 13.15  &        & 3.51   & 4.62   & 5.52   & 6.51  \\
$2500$ &        &  4.80  &  4.80  &  6.80  & 11.45  &        & 4.21   & 4.52   & 5.92   & 6.83  \\
$5000$ &        &  3.80  &  4.12  &  5.02  &  9.80  &        & 3.90   & 4.30   & 6.61   & 5.81  \\
 
		\midrule
		$T$ & &\multicolumn{9}{c}{$h$-step aggregate forecasts}\\
		\cmidrule(lr){3-11}
		$250$  &        & 11.46  & 13.85  & 22.78  & 31.54  &        &  8.98  & 11.08  & 20.58  & 26.22 \\
$500$  &        &  8.41  & 13.02  & 20.50  & 30.74  &        &  7.01  & 10.10  & 18.07  & 23.01 \\
$1000$ &        &  6.63  &  9.04  & 16.72  & 26.15  &        &  3.61  &  6.46  & 14.04  & 18.69 \\
$2500$ &        &  4.80  &  6.73  & 10.43  & 19.17  &        &  4.11  &  4.52  &  9.28  & 12.24 \\
$5000$ &        &  4.10  &  4.42  &  8.52  & 14.36  &        &  4.40  &  4.02  &  7.02  &  8.72 \\
  
		\addlinespace
		\hline
		\hline
		\addlinespace
		\multicolumn{11}{p{.72\linewidth}}{\textit{Notes:} This table shows test sizes for the joint VaR and ES forecast encompassing test based on the convex link function with a nominal size of $5\%$.
		We simulate data from the two GARCH specifications in \eqref{eqn:GARCHModel} - \eqref{eqn:GARCHresiduals} with normal innovations and consider $h$-step ahead and $h$-step aggregate forecasts for $h=1,2,5,10$.
		}
	\end{tabularx}
\end{table}

Table \ref{tab:Multistep} reports the tests sizes and Figure \ref{fig:sim_multi_VaRES} presents size-adjusted power\footnote{The size-adjusted power plots for $h=10$ and $T \in \{250,500,1000\}$ in Figure \ref{fig:sim_multi_VaRES} exhibit test sizes under the null hypotheses slightly above $5\%$. These are an artifact stemming from the fact that slightly more than $5\%$ of the simulated $p$-values are exactly zero, rendering an \textit{exact} size-adjustment in the sense of \cite{DavidsonMacKinnon1998} infeasible.} plots of the joint VaR and ES encompassing test for multi-step ahead and multi-step aggregate forecasts for a nominal significance level of $5\%$.\footnote{Figure \ref{fig:sim_multi_VaRES_RawPower} in the supplementary material shows the corresponding raw power.
Table \ref{tab:Multistep_AuxES}, Figure \ref{fig:sim_multi_AuxES_SA} and Figure \ref{fig:sim_multi_AuxES} in the supplementary material show test results for the auxiliary ES encompassing test.}
The encompassing tests for $h$-step ahead forecasts are well-sized, especially for larger sample sizes and for small horizons $h$.
The empirical sizes deteriorate slightly with an increasing forecast horizon $h$. 
While the general behavior is similar for $h$-step aggregate forecasts, these tests suffer considerably more from an increase of the forecasting horizon $h$.
The inferior performance of multi-period aggregate forecasts is not surprising given that the moment conditions of the aggregate forecasts are heavily correlated due to the overlapping definition of the aggregate forecasts.

\begin{figure}[p!]
	\includegraphics[width=\linewidth]{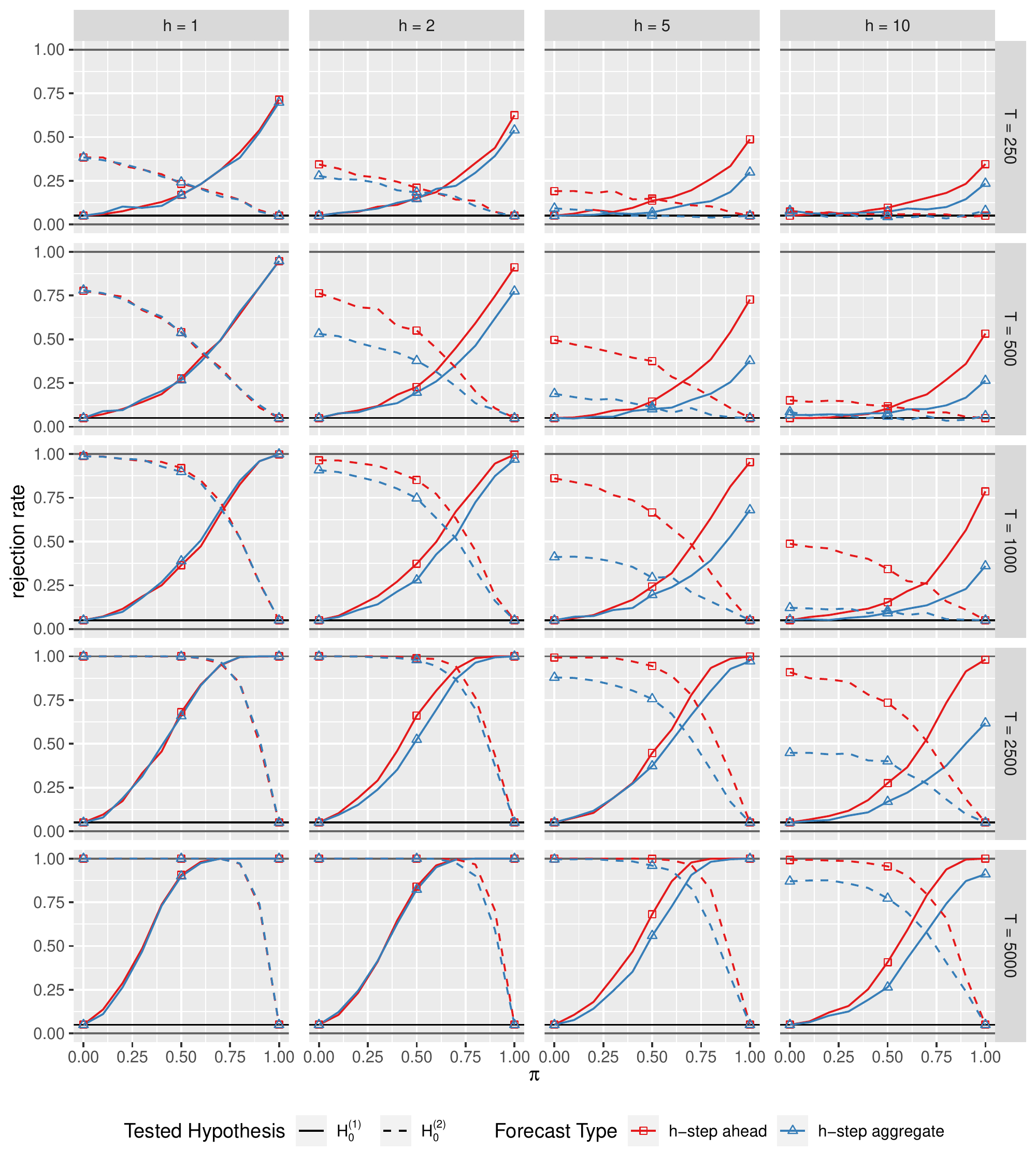}
	\caption{This figure shows size-adjusted power curves for the joint VaR and ES encompassing test with a nominal size of $5\%$, for $h$-step ahead and aggregate forecasts indicated with different colors, and for the two tested null hypotheses indicated with different line types.
		The plot rows depict different sample sizes, while the plot columns refer to different forecast horizons $h$.
		An ideal test exhibits a rejection frequency of $5\%$ for $\pi=0$ and for $\mathbb{H}_0^{(1)}$ (and inversely for $\pi=1$ and $\mathbb{H}_0^{(2)}$) and as sharply increasing rejection rates as possible for increasing (decreasing) values of $\pi$. 
		Note that we use a Bernoulli draw based combination method in this section as opposed to the variance combination in Section \ref{sec:onestep_forecasts} and hence, the results of the one-step ahead forecasts are not necessarily identical.}
	\label{fig:sim_multi_VaRES}
\end{figure}

Concerning the size-adjusted power, depicted in Figure \ref{fig:sim_multi_VaRES}, we observe similar patterns. 
For $h=1,2$, the size-adjusted power increases substantially for an increasing degree of misspecification for all considered settings.
For longer forecast horizons $h=5,10$, the test power is generally lower for both forecast types.
As before, the encompassing tests for $h$-step ahead forecasts exhibit better properties than for $h$-step aggregate forecasts. 
This can again be explained by the inherent correlation in $h$-step aggregate forecasts which necessitates the use of a sufficiently large amount of observations for a consistent estimation of the asymptotic covariance matrix together with its nuisance quantities and autocorrelation structure.
Hence, with an increasing forecast horizon $h$, a larger out-of-sample period is required to obtain encompassing tests with reliable test decisions.  
Furthermore, small sample sizes paired with large forecast horizons (e.g., $T=250$ and $h=10$) yield almost flat (size-adjusted) power curves which implies that the test becomes unreliable and practical applications should be interpreted very carefully in these scenarios.\footnote{Along these lines, \cite{Harvey2017} notice similar small-sample issues for forecast encompassing tests and tests for equal predictive ability \citep{DieboldMariano1995} for multi-step ahead forecasts.}
This negative result is remarkable concerning the planned evaluation of 10-day ahead aggregate ES forecasts \cite[p.89]{Basel2019}.

\section{Empirical Application}
\label{sec:Application}

This section empirically illustrates the usefulness of the proposed encompassing tests by comparing alternative VaR and ES forecasts for daily S\&P\,500 returns from August 4, 2000 to June 19, 2020 including a total of 5000 daily observations.
We conduct a rolling window forecasting scheme with $S=2000$ estimation observations, and $T=3000$ evaluation points starting on July 22, 2008. 
We follow the Basel Accords \citep{Basel2017, Basel2019} and employ $\alpha = 2.5\%$.
Based on the simulation results of Section \ref{sec:Simulations}, we restrict our attention to the tests based on the convex link functions with intercepts in the empirical application.

Particularly, we consider a total of eleven competing risk models for forecasting VaR and ES, including: 
(i) a rolling Historical Simulation using the window length of 250 days, 
(ii) the RiskMetrics model,
(iii) the GARCH(1,1) model with normal innovations (GARCH-N), and the GJR-GARCH(1,1) model of \cite{Glosten1993} with skewed Student-$t$ distributed innovations (GJR-ST),
(iv) the GARCH and GJR-GARCH models with asymmetric Laplace innovations (GARCH-AL and GJR-AL) and the same models with a time varying shape parameter (GARCH-AL-TVP and GJR-AL-TVP) of \cite{Chen2012},
(v) the symmetric absolute value (SAV-) and asymmetric slope (AS-) CAViaR-ES models of \cite{Taylor2019}, 
and (vi) the one factor GAS model (GAS-1F) of \cite{Patton2019}. 
Details for the risk models of \cite{Chen2012}, \cite{Taylor2019} and \cite{Patton2019} are given in Section \ref{sec:RiskModels} in the supplementary material and an additional absolute evaluation in the form of backtests for these models is given in Section \ref{sec:AbsEvaluation} in the supplementary material.

\subsection{One-Step Ahead Forecasts}
\label{sec:EncompOneRes}

In this subsection, we analyze pairwise encompassing for one-step ahead VaR and ES forecasts using the encompassing tests based on the convex link functions.
For each model pair, we estimate the combination weights and test both null hypotheses, i.e.\ that model one encompasses model two and vice versa.
We obtain simulated critical values for the test through Theorem \ref{thm:WaldTestAsymptoticDistribution} and by employing the \textit{scl-sp} estimator of \cite{DimiBayer2019}.
Due to the simulation results of Section \ref{sec:cov_choice} in the supplementary material, we do not consider estimation of HAC-terms in the covariance for one-step ahead forecasts.

\begin{table}[t]
	\footnotesize
	\centering
	\caption{Empirical Encompassing Test Results for One-Step Ahead Forecasts}
	\label{tab:OneSummary}
	\begin{tabularx}{\linewidth}{l @{\hspace{0.5cm}}ccccc @{\hspace{0.4cm}}  ccccc @{\hspace{0.1cm}} c}
	\hline\hline
	\addlinespace
	& \multicolumn{4}{c}{Joint VaR and ES Test} & & \multicolumn{4}{c}{Auxiliary ES Test} \\ 
	\cmidrule(lr){2-5} \cmidrule(lr){7-10} 
	\addlinespace			
	{Models} & {E'ing} & {E'ed} & {Comb} & {Incon}  & & {E'ing} & {E'ed} & {Comb} & {Incon} & & {Avg.\ Weights}\\
	\addlinespace
	\hline 
	\addlinespace
	{GJR-ST} & {10} & {0} & {0} & {0} & & {9} & {0} & {0} & {1} && {(0.82, 0.98)}  \\
	\addlinespace
	{GJR-AL-TVP} & {6} & {1} & {3} & {0} & & {8} & {1} & {0} & {1} && {(0.86, 0.71)} \\
	\addlinespace
	{AS-CAViaR-ES} & {5} & {3} & {1} & {1} & & {6} & {0} & {0} & {4} && {(0.56, 0.66)}  \\
	\addlinespace
	{GARCH-AL} & {4} & {1} & {5} & {0}  & & {2} & {2} & {2} & {4} && {(0.62, 0.63)}  \\
	\addlinespace
	{GARCH-AL-TVP} & {3} & {1} & {5} & {1} & & {2} & {3} & {3} & {2} && {(0.47, 0.60)} \\
	\addlinespace
	{GJR-AL-CP} & {3} & {1} & {3} & {3}  & & {6} & {3} & {0} & {1} && {(0.47, 0.68)} \\
	\addlinespace
	{GARCH-N} & {2} & {6} & {1} & {1} & & {2} & {6} & {0} & {2} && {(0.46, 0.35)} \\
	\addlinespace
	{SAV-CAViaR-ES} & {1} & {4} & {3} & {2} & & {2} & {4} & {1} & {3} && {(0.45, 0.39)}  \\
	\addlinespace
	{RiskMetrics} & {1} & {5} & {1} & {3} & & {1} & {4} & {2} & {3} && {(0.39, 0.21)} \\
	\addlinespace
	{GAS-1F} & {1} & {7} & {1} & {1}  & & {1} & {8} & {0} & {1} && {(0.27, 0.20)}  \\
	\addlinespace
	{Historical Sim} & {0} & {7} & {3} & {0} & & {0} & {8} & {2} & {0} && {(0.09, 0.06)} \\
	\addlinespace
	\hline 
	\hline 
	\addlinespace
	\multicolumn{12}{p{0.98\linewidth}}{\textit{Notes:} This table reports a summary of the test results of the joint VaR and ES and the auxiliary ES encompassing tests based on the convex link function for one-step ahead forecasts. 
	Entries for ``E'ing'' represent the number of occurrences (out of 10) that a row-heading model encompasses a competing model.
	Similarly, ``E'ed'' represent the frequencies that the row-heading model is encompassed, ``Comb'' that neither model encompasses its competitor, and ``Incon'' that both models encompass each other.
	The column ``Avg.\ Weights'' shows the estimated convex combination weights $(\theta_{1},\theta_{2})$, averaged over the 10 estimates for each model.}
	\end{tabularx}
\end{table}

We report the summarized results for the joint VaR and ES and the auxiliary ES encompassing tests with a significance level of $5\%$ for all pairwise combinations of the eleven risk models in Table \ref{tab:OneSummary}, where the models (in the table rows) are sorted according to their encompassing performance.
Out of the ten model combinations each individual model is subject to, we report the instances how often both null hypotheses are rejected (denoted by "Combination" or "Comb"), not rejected ("Inconclusive" or "Incon"), only the first one is rejected ("Encompassed" or "E'ed"), and only the second one is rejected ("Encompassing" or "E'ing").
Notice that the "Combination" column is based on rejecting both null hypotheses, which constitutes a multiple testing problem and the results have to be interpreted at a Bonferroni corrected significance level of $10\%$, while each individual tests are based on a nominal significance level of $5\%$.\footnote{Table \ref{tab:corr_one} in the supplementary material reports the correlations of the VaR and ES forecasts and Table \ref{tab:JointOne} 
additionally reports the estimated (convex) combination weights together with the test decisions for the combinations of the six bestperforming models, chosen by the absolute evaluation in Table \ref{tab:AbsEvaluation}.}

We find that the GJR-GARCH models with Skew-t and asymmetric Laplace innovations achieve the best forecasting performance among the competing models.
Interestingly, the CAViaR-ES models of \cite{Taylor2019} and the GAS-1F model of \cite{Patton2019}, which are specifically developed for jointly forecasting VaR and ES, generally do not perform as good as the GARCH specifications.
As expected, the RiskMetrics and Historical Simulation models perform  worst.
Furthermore, we find many instances of rejections of both encompassing hypotheses, implying that a forecast combination via the estimated encompassing weights is superior to both individual models.
This result justifies the usefulness of the proposed encompassing tests, and is in line with the arguments for forecast combinations of \cite{GiacominiKomunjer2005}, \cite{Timmermann2006}, \cite{Taylor2020} and \cite{DimiSchnaitmann2020}.

\subsection{10-Step Ahead and Aggregate Forecasts}
\label{sec:EncompMultiRes}

In this subsection, we apply the proposed encompassing tests to 10-day ahead and aggregate VaR and ES forecasts.
Note that 10-day aggregate VaR and ES forecasts are required by the Basel Accords for minimal capital requirement and risk weighted assets \citep{Basel2019, Basel2020}.
As it is unclear how to obtain multi-step ahead forecasts from the CAViaR-ES models of \cite{Taylor2019} and the GAS-1F model of \cite{Patton2019}, we reduce the set of evaluation models to the seven members of the GARCH family.
For these models, we obtain multi-step ahead and multi-step aggregate VaR and ES forecasts through the simulation method of  \cite{WongSo2003}, further described in Section \ref{sec:multistep_forecasts}.
Such a simulation-based forecasting is necessary as the conditional distribution of multi-step returns generally differs from the imposed innovation distribution of the model, and thus, VaR and ES forecasts cannot be obtained through classical location-scale formulas as for one-step ahead forecasts.
Based on the results of Section \ref{sec:cov_choice} in the supplementary material, we use a HAC covariance estimator \citep{NeweyWest1987}, augmented with the \textit{scl-sp} estimator of \cite{DimiBayer2019} for the contemporaneous variance component to perform the encompassing tests for multi-step ahead and aggregate forecasts.

\begin{table}[ph!]
	\footnotesize
	\centering
	\caption{Encompassing Test Results for 10-Step Ahead and Aggregate Forecasts}
	\label{tab:Ahead}
	\begin{tabularx}{\linewidth}{l @{\hspace{0.5cm}}ccccc @{\hspace{0.4cm}}  ccccc @{\hspace{0.1cm}} c}
	\hline 
	\hline
	\addlinespace
	\multicolumn{12}{c}{10-Step Ahead Forecasts} \\
	\addlinespace
	\hline
	\addlinespace
	& \multicolumn{4}{c}{Joint VaR and ES Test} & & \multicolumn{4}{c}{Auxiliary ES Test} \\ 
	\cmidrule(lr){2-5} \cmidrule(lr){7-10} 
	\addlinespace			
	{Models} & {E'ing} & {E'ed} & {Comb} & {Incon} & & {E'ing} & {E'ed} & {Comb} & {Incon} & {Avg.\ Weights}\\
	\addlinespace
	\hline 
	\addlinespace
	{GJR-AL-TVP} & {6} & {0} & {0} & {0} &  & {5} & {0} & {0} & {1} & {(0.98, 0.99)} \\
	\addlinespace
	{GARCH-AL-TVP} & {5} & {1} & {0} & {0} & & {3} & {0} & {0} & {3} & {(0.74, 0.84)}  \\
	\addlinespace
	{GJR-ST} & {2} & {2} & {0} & {2} & & {3} & {1} & {0} & {2} & {(0.63, 0.49)}  \\
	\addlinespace
	{GARCH-N} & {1} & {2} & {0} & {3} & &  {1} & {2} & {0} & {3} & {(0.51, 0.40)}  \\
	\addlinespace
	{GARCH-AL} & {1} & {3} & {0} & {2} & &  {0} & {2} & {0} & {4} & {(0.26, 0.30)}  \\
	\addlinespace
	{GJR-AL} & {0} & {5} & {0} & {1} & &  {0} & {4} & {0} & {2} & {(0.24, 0.32)} \\
	\addlinespace
	{RiskMetrics} & {0} & {2} & {0} & {4} & &  {0} & {3} & {0} & {3} & {(0.21, 0.23)}  \\
	\addlinespace
	\hline\hline
	\\
	\addlinespace
	\multicolumn{12}{c}{10-Step Aggregate Forecasts} \\
	\addlinespace
	\hline
	\addlinespace
	& \multicolumn{4}{c}{Joint VaR and ES Test} & & \multicolumn{4}{c}{Auxiliary ES Test} \\ 
	\cmidrule(lr){2-5} \cmidrule(lr){7-10} 
	\addlinespace			
	{Models} & {E'ing} & {E'ed} & {Comb} & {Incon} & & {E'ing} & {E'ed} & {Comb} & {Incon} & {Avg.\ Weights}\\
	\addlinespace
	\hline 
	\addlinespace
	{GJR-AL-TVP} & 6 & 0 & 0 & 0 & & {5} & {0} & {0} & {1} & (0.86, 0.93)  \\
	\addlinespace
	{GARCH-AL-TVP} & 4 & 1 & 0 & 1 & & {5} & {0} & {0} & {1} & (0.95, 0.86)  \\
	\addlinespace
	GJR-ST & 3 & 1 & 0 & 2 & & {3} & {2} & {0} & {1} & (0.51, 0.34) \\
	\addlinespace
	{GARCH-N} & 2 & 3 & 0 & 1 & & {2} & {2} & {0} & {2} & (0.41, 0.47) \\
	\addlinespace
	{GJR-AL} & 1 & 2 & 1 & 2 & & {2} & {3} & {0} & {1} &  (0.37, 0.42) \\
	\addlinespace
	{GARCH-AL} & 1 & 4 & 1 & 0 & & {1} & {5} & {0} & {0} &  (0.33, 0.42) \\
	\addlinespace
	{RiskMetrics} & 0 & 6 & 0 & 0 &  & {0} & {6} & {0} & {0} & (0.04, 0.03) \\
	\addlinespace
	\hline 
	\hline 
	\addlinespace
	\multicolumn{12}{p{0.98\linewidth}}{\textit{Notes:} This table reports a summary of the test results of the joint VaR and ES and the auxiliary ES encompassing tests based on the convex link function, for 10-step ahead forecasts in the upper panel and for 10-step aggregate forecasts in the lower panel.
	Entries for ``E'ing'' represent the number of occurrences (out of 6) that a row-heading model encompasses a competing model.
	Similarly, ``E'ed'' represent the frequencies that the row-heading model is encompassed, ``Comb'' that neither model encompasses its competitor, and ``Incon'' that both models encompass each other.
	The column ``Avg.\ Weights'' shows the estimated convex combination weights $(\theta_{1},\theta_{2})$, averaged over the 10 estimates for each model.}
	\end{tabularx}
\end{table}

Table \ref{tab:Ahead} reports the summarized encompassing test results for 10-step ahead and aggregate forecasts.\footnote{Table \ref{tab:JointDetailsMultiAhead} and Table \ref{tab:JointDetailsMultiAggregate} in the supplementary material report the detailed test results. Table \ref{tab:corr_ahead} additionally reports correlations for the 10-day ahead and aggregate VaR and ES forecasts.}
The test results show that for both, 10-step ahead and aggregate forecasts, the best performing model is the GJR-GARCH model with asymmetric Laplace innovations and a time-varying shape parameter.
We find almost  no cases of double rejections, i.e.\ forecast combinations are not (significantly) preferred over the stand-alone models.
This can be an artifact from the lower power for multi-step forecasts as illustrated in Section \ref{sec:multistep_forecasts} or from the high(er) correlations of the forecasts, reported in Table \ref{tab:corr_ahead} in the supplementary material.

Overall, the empirical results show that a model specified with an asymmetric volatility process and a skewed error distribution, such as the GJR-ST model, outperforms the competing models considered in this paper for one-step ahead VaR and ES forecasts.
Moreover, models based on an asymmetric innovation distribution with time-varying parameters, such as, GJR-AL-TVP and GARCH-AL-TVP models, perform better than the other competing models. 
Note that the time-varying scale parameter of the asymmetric Laplace distribution produces both time-varying skewness and kurtosis for the innovation distribution. 
We find that specifying time-varying higher moments for a risk model substantially improves the model forecasting performance in both multi-step ahead and aggregate risk forecasts, much more than in one-step ahead forecasts.

\section{Conclusion}
\label{sec:Conclusion}

This article proposes joint encompassing tests which compare one-step and multi-step  VaR and ES forecasts based on general semiparametric forecast combination methods (link functions) for the VaR and ES.
While unrestricted linear methods are often employed in encompassing tests for functionals like the mean and quantiles (the VaR) as e.g.\ in \cite{HendryRichard1982, GiacominiKomunjer2005}, different combination methods are of particular interest for the ES.
E.g., our \textit{no-crossing} link specification theoretically circumvents crossings of the predicted VaR and ES, which is conceptually desirable but not straight-forward to achieve \citep{Taylor2020}.

Our employed link functions imply that some of the tested parameters are on the boundary of the parameter space under the null hypothesis, which necessitates non-standard asymptotic theory.
Based on the general framework of \cite{Andrews1999, Andrews2001}, we provide such novel asymptotic theory for the proposed encompassing tests and for the accompanying Wald test statistics, which allows for inference and testing on the boundary.
Our simulations show that the proposed  VaR and ES forecast encompassing tests based on the convex and no-crossing link functions exhibit superior size and power properties than those based on unrestricted linear link functions.
By employing the proposed encompassing tests in a real data analysis, we find that building risk models on specifications including time-varying higher moments substantially improves the model forecasting performance, especially for multi-step ahead and aggregate VaR and ES forecasts.

Our framework allows for several straight-forward extensions.
Encompassing tests for multiple VaR and ES forecasts in the sense of \cite{HarveyNewbold2000} can directly be implemented through our asymptotic theory by adapting the link functions.
Furthermore, the asymptotic theory allows for encompassing tests based on any strictly consistent loss function for the VaR and ES.
Incorporating estimation risk into these tests can be obtained by combining our theory with the work of \cite{EscancianoOlmo2010, Du2017, BarendseKole2019}, and incorporating model misspecification through combining our theory with the one of \cite{DimiSchnaitmann2020}.
Encompassing tests for different functionals such as e.g., the mean, quantiles, expectiles or probability densities based on link functions which require testing on the boundary (e.g., using convex link functions) can be implemented through adapting our asymptotic theory to semiparametric models for the functional under consideration.
Eventually, our asymptotic theory can be used to test (e.g., for nullity of) model parameters on the boundary of the parameter space for the semiparametric VaR and ES models of \cite{Patton2019}, \cite{Taylor2019} or \cite{Gerlach2019}, along the lines of \cite{Francq2009}.

\section*{Acknowledgments}
	Our work has been supported by the University of Hohenheim, the Klaus Tschira Foundation and the University of Konstanz.
	A previous version of this paper circulated with the title "A Regression-based Joint Encompassing Test for Value-at-Risk and Expected Shortfall Forecasts".

\onehalfspacing
\setlength{\bibsep}{5pt}
\bibliographystyle{apalike}	
\bibliography{bib_ConvexESencmp}

\appendix
\doublespacing

\section{Proofs}
\label{sec:Proofs}

	\begin{proof}[Proof of Theorem \ref{thm:GeneralAsymptoticDistribution}]
		For this proof, we employ Theorem 3 of \cite{Andrews1999} (or equivalently Theorem 1 of \cite{Andrews2001}), for which we	verify the necessary Assumptions 1-6 of \cite{Andrews1999} in the following.
		
		We start by showing Assumption 1, i.e.\ the consistency of $\hat \theta_T$.
		For this, we employ Theorem 2.1 of \cite{NeweyMcFadden1994}.
		Assumption $(i)$, i.e.\ that $l(\theta)$ is uniquely minimized by $\theta^0$ follows directly from the identification condition \ref{cond:FullRankConditionNormality} and from the strict consistency result of the loss functions of \cite{Fissler2016}.
		Condition $(ii$) follows directly as we impose that $\Theta$ is compact.
		Condition $(iii)$ holds as $l(\theta)$ is continuous for all $\theta \in \Theta$ as the distribution $F_t$ is absolutely continuous and we use continuously differentiable functions $\mathfrak{g}$, $\phi$, $g^q$ and $g^e$.
		The uniform consistency of $T^{-1} l_T(\theta)$ of condition $(iv)$ is shown by employing Theorem 21.9 of \cite{Davidson1994}.
		For this, we need that a point-wise law of large numbers holds for $T^{-1} l_T(\theta)$ for all $\theta \in \Theta$, which can be verified e.g., by employing Corollary 3.48 of \cite{White2001}.
		This holds as $\rho_t(\theta)$ is $\alpha$-mixing of size $-r/(r-1)$ for $ r > 1$ from condition \ref{cond:BetaMixing} as $\beta$-mixing series are also $\alpha$-mixing of same size by \cite{Bradley2005} and $\mathbb{E}\big[ |\rho_t(\theta)|^{2r} \big] < \infty$ for all $\theta \in \Theta$ by condition \ref{cond:MomentCondition}.
		Furthermore, the sequence $l_T(\theta)$ is stochastically equicontinuous by Lemma \ref{lemma:StochasticEquicontinuity} in the supplementary material.
		Thus, $\sup_{\theta \in \Theta} \left| T^{-1} l_T(\theta) - l(\theta) \right| \toP 0$ and consistency of $\hat \theta_T$ follows from Theorem 2.1 of \cite{NeweyMcFadden1994}.
		
		Assumption 2$^{\ast}$ of \cite{Andrews1999} is shown through the sufficient condition Assumption 1$^{2\ast}$ on page 53 in \cite{Andrews1997}.
		For this, condition $(a)$ holds trivially, and condition $(b)$ follows directly from the uniform consistency result of $T^{-1} l_T(\theta)$.
		For condition $(c)$, we set $\Theta^+ = \Theta$. Given condition \ref{cond:ParameterSpace}, locally to $\theta^0$, $\Theta$ equals a union of (Cartesian) orthants.
		For condition $(d)$, we notice that $l(\theta)$ is twice continuously differentiable on the interior of $\Theta$ and has partial right/left derivatives on the boundary of $\Theta$ of order one and two.
		It further holds that $\nabla_\theta l(\theta^0) = 0$ as the function $l(\theta)$ is uniquely minimized by $\theta^0$.
		Notice that this also holds for the respective directional derivatives if $\theta^0$ lies on the boundary of $\Theta$.
		
		Eventually, for condition $(e)$, Lemma \ref{lemma:TypeIVFunction} in the supplementary material shows that the functions $\rho_t(\theta)$ given in (\ref{eqn:JointLossESRegGeneral}) form a \textit{type IV} class (see \cite{Andrews1994}, p.2278) with index $p=2r$ such that by Theorem 6 in \cite{Andrews1994}, it satisfies Ossiander's $L^{2r}$-entropy condition and consequently has an $L^{2r}$-envelope.
		Furthermore, the moments 	
			$\mathbb{E} \left[ {\sup}_{\tilde{\theta} \in U(\theta, \delta) } \left|\left|  \rho_t (\tilde{\theta}) \right| \right|^{2r}  \right]^{1/{2r} } < \infty$
		are bounded by assumption.
		Consequently,  by Theorem 1 and Application 1 in \cite{Doukhan1995}, we obtain that the empirical process, given by $T^{-1/2} \sum_{t \in \mathfrak{T}} \big( \rho_t(\theta) - \mathbb{E}[ \rho_t(\theta) ] \big)$,	is stochastically equicontinuous (see the remark on p.410 of \cite{Doukhan1995}).
		Hence, the process $T^{-1} l_T(\theta) - l(\theta)$ is stochastically differentiable (see e.g.\ \cite{NeweyMcFadden1994}, p.2187 or the proof after Theorem 3.1 in \cite{Dobric1994}, which does not rely on the imposed iid assumption of that paper).
		Thus, all conditions in Assumption 1$^{2\ast}$ of \cite{Andrews1997} are fulfilled and hence, Assumption 2$^{\ast}$ of \cite{Andrews1999} holds.
		
		In the following, we verify Assumption 3$^\ast$ (which implies Assumption 3) of \cite{Andrews1999}, i.e.\ that $T^{-1/2} \sumt \psi_t(\theta^0) \tod G$, where $G \sim \mathcal{N}(0, \mathcal{I})$, and $\mathcal{I} = \Var \left( T^{-1/2} \sumt \psi_t(\theta^0) \right)$, with $\psi_t(\theta^0)$ given in (\ref{eqn:Psi_t}).
		By using the Cramer-Wold theorem, we instead show that $T^{-1/2} \sumt u^\top \psi_t(\theta^0) \tod u^\top G u$ for all $u \in \mathbb{R}^k$ where $||u|| =1$. 
		This holds as $Z_t$ is assumed to be $\beta$-mixing of size $-r/(r-1)$ for $ r > 1$ from condition \ref{cond:BetaMixing} and $\beta$-mixing implies $\alpha$-mixing of same size \citep{Bradley2005}.
		By Theorem 3.49 in \cite{White2001}, we then get that $u^\top \psi_t(\theta^0)$ are also $\alpha$-mixing of the same size.
		Furthermore, it holds that $\mathbb{E} \left[ \left| u^\top \psi_t(\theta^0) \right|^{2r} \right] < \mathbb{E} \left[ \sup_{\theta \in \Theta} ||\psi_t(\theta)||^{2r}  \right] < \infty$ condition \ref{cond:MomentCondition}  in Assumption \ref{assu:AsymptoticTheory}.
		The matrix $\mathcal{I} = \Var \left( T^{-1/2} \sumt \psi_t(\theta^0) \right)$ does not depend on $T$ as the process is assumed to be stationary.
		As $\mathcal{I}$ has full rank by condition \ref{cond:AsyCovPositiveDefinite}, it holds that $\Var \left( T^{-1} \sumt u^\top \psi_t(\theta^0) \right)
		\ge \lambda_{\text{min}} > 0$, where $\lambda_{\text{min}}$ is the smallest Eigenvalue of  $\mathcal{I}$.
		Consequently, applying Theorem 5.20 in \cite{White2001} delivers the asymptotic normality result.
		
		Following condition \ref{cond:ParameterSpace}, the parameter space is given as the product $\Theta = \mathcal{B}_1 \times \mathcal{B}_2 \times \Delta \times \Psi$, and each of these four spaces is given by (linear) inequality constraints.
		Consequently, $\Theta$ can also be expressed through a system of inequalities of the form $\Gamma_\theta \theta \le r_\theta$, for some matrix $\Gamma_\theta$ and vector $r_\theta$ of appropriate dimensions.
		Then, following equations (4.6) and (4.7) of \cite{Andrews1999}, the cone $\Lambda$ is given by 
		\begin{align}
			\Lambda = \{ \lambda \in \mathbb{R}^{k}:  \Gamma_{\theta}^{(b)} \lambda \le 0 \},
		\end{align}
		where $\Gamma_{\theta}^{(b)}$ consists of the rows of $\Gamma_\theta$ for which the inequality $\Gamma_{\theta} \theta^0 \le r_\theta$ is binding (i.e.\ it holds as an equality).
		As this specification of $\Lambda$ is a convex cone, this shows Assumption 5 and 6 of \cite{Andrews1999}, i.e.\ that $\Theta - \theta^0$ locally equals a convex cone $\Lambda \subset \mathbb{R}^k$.
		
		Consequently, we can apply Theorem 3 of \cite{Andrews1999} (or equivalently Theorem 1 of \cite{Andrews2001}), which completes the proof of this theorem.
	\end{proof}

\begin{proof}[Proof of Theorem \ref{thm:InferenceBeta}]
	The result follows directly from Theorem 2 of \cite{Andrews2001}. 
	Besides the assumptions of Theorem \ref{thm:GeneralAsymptoticDistribution} (of the present article), we further need to verify Assumptions 7 and  8 of \cite{Andrews2001}.
	Assumption 7$(a)$ is fulfilled by the imposed condition \ref{cond:BreadMatrixZeroBlock} in Assumption \ref{assu:AsymptoticTheory2}.
	Furthermore, Assumption 7$(b)$ follows directly from condition \ref{cond:ParameterSpace} and from the specification given in (\ref{eqn:ParamSpace}).
	Assumption 8 also holds trivially as $\delta$ is assumed to be in the interior of $\Theta$, which concludes this proof.
\end{proof}

\begin{proof}[Proof of Theorem \ref{thm:WaldTestAsymptoticDistribution}]
	In order to employ Theorem 6 of \cite{Andrews2001}, we verify the necessary Assumptions 9 and 12 of \cite{Andrews2001}. 
	For the verification of Assumptions 1-8, see the proof of Theorem \ref{thm:GeneralAsymptoticDistribution} and Theorem \ref{thm:InferenceBeta}.
	For Assumption 9, notice that testing $\beta_1 = \beta_1^\ast$ corresponds to the null hypothesis that $\theta \in \Theta_0 = \big\{ \theta = (\beta_1, \beta_2, \delta, \psi) \in \Theta: \beta_1 =\beta_1^\ast \big\}$.
	Consequently, Assumption 9$(a)$ is satisfied.
	Assumption 9$(b)$ holds as throughout the paper $B_T = \sqrt{T} I_k$ and Assumption 9$(c)$ follows directly from condition \ref{cond:ParameterSpace}.
	Eventually, Assumption 9$(d)$ follows as $\mathcal{B}_1$ and  $\mathcal{B}_2$ in condition \ref{cond:ParameterSpace} are given by separate inequality constraints.
	
	Assumption 12$^\ast(a)$ corresponds to Assumption 11$(a)$, which requires that the random variable $G \sim \mathcal{N}(0, \mathcal{I})$ (simplified for the case that the space $\Pi$ is single-valued).
	This follows directly from the proof of Theorem \ref{thm:GeneralAsymptoticDistribution} (where Assumption 3$^\ast$ of \cite{Andrews1999} is verified).
	Assumption 12$^\ast(b)$ follows directly from (\ref{eqn:WeightingMatrixWaldTest}) and the conditions Assumption 12$^\ast(c)$.
	
	Consistency of the "bread" matrix, $\hat{\mathcal{T}}_T \toP \mathcal{T}$ follows directly from Theorem 3 of \cite{Patton2019} and consistency of the HAC estimator $\hat{\mathcal{I}}_T \toP \mathcal{I}$ is shown in Lemma \ref{lemma:ConsistencyHAC} in the supplementary material.
	Notice for this that \textit{joint} convergence in probability $\big( \hat{\mathcal{T}}_T, \hat{\mathcal{I}}_T \big) \toP \big( \mathcal{T}, \mathcal{I} \big)$ follows directly from both variables converging in probability separately.
	Eventually, Assumption 12$^\ast(e)$ follows as the matrix $\mathcal{I}$ has full rank by assumption.
	
	Consequently, the conditions of Theorem 6 of \cite{Andrews2001} are satisfied and part $(d)$ yields that $W_T \tod \hat \lambda_{\beta_1}^\top V^{-1} \hat \lambda_{\beta_1}$,
	where $\hat \lambda = \big( \hat \lambda_{\beta_1}, \hat \lambda_{\beta_2}, \hat \lambda_{\delta}, \hat \lambda_{\psi} \big)$ is given in Theorem \ref{thm:InferenceBeta}.
\end{proof}

\newpage
\appendix
\doublespacing

\setcounter{page}{1}
\begin{center}
	SUPPLEMENTARY MATERIAL FOR   \vspace{10pt} \\
	{\Large\bf {Encompassing Tests for Value at Risk and Expected Shortfall Multi-Step Forecasts based on Inference on the Boundary} \vspace{10pt} } \\
	Timo Dimitriadis \qquad Xiaochun Liu \qquad Julie Schnaitmann \\
	\today \\
	\vspace{1cm} 
\end{center}

\onehalfspacing

All references to equations, sections, tables and figures starting with \textcolor{blue}{S.} refer to this supplement while the remaining references refer to the main document of the article.

\renewcommand{\thesection}{S.\arabic{section}}   
\renewcommand{\thepage}{S.\arabic{page}}  
\renewcommand{\thetable}{S.\arabic{table}}   
\renewcommand{\thefigure}{S.\arabic{figure}}   
\setcounter{section}{0}
\setcounter{table}{0}
\setcounter{figure}{0}

\section{Additional DGPs for the Simulation Study}	
\label{sec:add_simulation}

Following \cite{DimiSchnaitmann2020}, this section provides simulation results for two additional data generating processes (DGPs) outside the class of location-scale models as a robustness check for the proposed encompassing tests.

For the first additional simulation design, we introduce two specifications of generalized autoregressive score (GAS) models proposed by \cite{Creal2013}.
We generate $r_{1,t+1}$, $\hat q_{1,t}$ and $\hat e_{1,t}$ from a GAS model with Gaussian innovations, which corresponds to the standard GARCH(1,1) specification given in \eqref{eqn:GARCHModel}.
We obtain the second sequence of forecasts from a GAS model with Student-$t$ residuals with time-varying variance and degrees of freedom, given by
\begin{align}
\label{eqn:GAS_t}
(\hat \mu_2, \hat \sigma_{2,t}^2 , \hat \nu_{2,t})^\top = \kappa + B \cdot (\hat \mu_2, \hat \sigma_{2,t-1}^2 , \hat \nu_{2,t-1})^\top  + A H_t \nabla_t,
\end{align}
where $H_t \nabla_t$ is the forcing variable of the model, the scaling matrix $H_t$ is the Hessian and $\nabla_t$ the derivative of the log-likelihood function.
We calibrate both models to daily S\&P 500 returns resulting in the parameter values $\kappa = ( 0.0659,  0.00599, -1.737)$, $A =  \operatorname{diag}(0, 0.146, 7.563)$ and  $B =  \operatorname{diag}(0, 0.994, 7.381)$.
This model implies that $r_{2,t+1} \sim t_{\hat \nu_{2,t}} \big( \hat \mu_2, \hat \sigma_{2,t}^2 \big)$ and we obtain one-step ahead VaR and ES forecasts from this $t$-distribution.

In the second additional simulation setup, we implement the one-factor (1F) and two-factor (2F) GAS models for the VaR and ES of \cite{Patton2019}.
The 1F-GAS model evolves as
\begin{align}
\begin{aligned}
\label{eqn:GAS1F}
\hat q_{1,t} &= -1.164 \exp(\hat \kappa_{t}) \qquad \text{ and } \qquad 
\hat e_{1,t} = -1.757 \exp(\hat \kappa_{t}), \quad \text{ where }\\
\hat \kappa_{t} &=  0.995 \hat \kappa_{t-1} + \frac{0.007}{\hat e_{1,t-1}} \left( \frac{r_{1,t}}{\alpha}  \mathds{1}_{\{r_{1,t} \le \hat q_{1,t-1} \}}  -  \hat e_{1,t-1} \right).
\end{aligned}
\end{align}
The 2F-GAS model follows the specification
\begin{align}
\label{eqn:GAS2F}
\begin{pmatrix} \hat q_{2,t} \\ \hat e_{2,t}  \end{pmatrix}
= \begin{pmatrix} -0.009 \\ -0.010  \end{pmatrix}
+ \begin{pmatrix} 0.993 & 0 \\ 0 & 0.994  \end{pmatrix}
\begin{pmatrix} \hat q_{2,t-1} \\ \hat e_{2,t-1}  \end{pmatrix}
+ \begin{pmatrix} -0.358 & -0.351 \\ -0.003 & -0.003  \end{pmatrix} \lambda_t,
\end{align}
where the forcing variable is given by $\lambda_t = \big( \hat q_{2,t-1} ( \alpha - \mathds{1}_{\{r_{2,t} \le \hat q_{2,t-1} \}} ) , \, \mathds{1}_{\{r_{2,t} \le \hat q_{2,t-1} \}}  r_{2,t}  / \alpha - \hat e_{2,t-1} \big)^\top$.
For both models, $j=1,2$, we simulate $r_{j,t+1}  \sim \mathcal{N} \big( \hat \mu_{j,t}, \hat \sigma_{j,t}^2 \big)$, where the conditional mean and standard deviations are given by
$\hat \mu_{j,t} = \hat q_{j,t} - z_\alpha \frac{\hat e_{j,t} -  \hat q_{j,t}}{\xi_\alpha - z_\alpha}$ and $\hat \sigma_{j,t} = \frac{\hat e_{j,t} -  \hat q_{j,t}}{\xi_\alpha - z_\alpha}$,
such that $Q_\alpha(r_{j,t+1}|\mathcal{F}_t) = \hat q_{j,t}$ and  $\ES_\alpha(r_{j,t+1}|\mathcal{F}_t) = \hat e_{j,t}$ almost surely.
The parameter values for this model are obtained from Table 8 of \cite{Patton2019} and correspond to calibrated parameters to daily S\&P\,500 returns.

In order to simulate returns which follow a convex combination of these two conditional distributions (for both DGPs), we simulate Bernoulli draws $\pi_{t+1} \sim \operatorname{Bern}(\pi)$ for 11 equally spaced values of $\pi \in [0,1]$, and let $Y_{t+1} = r_{t+1}  = (1-\pi_{t+1}) r_{1,t+1}  + \pi_{t+1} r_{2,t+1}$.
Thus, for $\pi=0$, $Y_{t+1}$ follows the first model, for $\pi = 1$, $Y_{t+1}$ follows the second model, and for $\pi \in (0,1)$, $Y_{t+1}$ follows some convex combination of the two models.\footnote{
	While generating returns stemming from convex combinations of GARCH-type volatility models is straight-forward by using convex combinations of the conditional volatilities, this is not as simple for the more general GAS models considered in this section. 
	Consequently, we use this more involved approach based on Bernoulli draws in order to generate these convex model combinations. This is comparable to our combination approach of multi-step forecasts.
}

Table \ref{tab:Size_GAS} shows the empirical sizes for the joint VaR and ES and the auxiliary ES encompassing tests for both null hypotheses, both DGPs, the three different link functions and various out-of-sample sizes $T$ and Figure \ref{fig:sim_other} presents the empirical rejection frequencies. 
As for the two GARCH-based DGPs, the tests based on the convex and no-crossing link functions outperform the tests based on the linear link function.
Moreover, the auxiliary ES test performs slightly better than the joint VaR and ES encompassing test, especially in terms of its size properties.
This qualitatively confirms the result for the GARCH DGPs of Section \ref{sec:Simulations}.

\newpage
\section{Covariances Estimation}
\label{sec:cov_choice}

In this section, we compare the performance of the joint VaR and ES encompassing test for four different covariance estimators, where we consider multi-step ahead and multi-step aggregate forecasts at different forecast horizons. 
These estimators differ with respect to the estimation of the "meat" matrix $\mathcal{I}$, given in \eqref{eqn:I_Matrix}, where we consider an estimator based on HAC terms \citep{NeweyWest1987} and one without, paired with either the outer product of the gradient of $\psi_t(\theta)$, or the \textit{scl-sp} estimator of \cite{DimiBayer2019}.

More precisely, the first estimator is given by $\hat{\mathcal{I}}^{(1)}_{T} = \widehat{\Omega}_{T,0}$, where the contemporaneous covariance matrix is estimated by the outer product of the gradient of $\psi_t(\theta)$ as given in \eqref{eqn:HACEstimation}.
The second estimator is specified as $\hat{\mathcal{I}}^{(2)}_{T} = \widetilde{\Omega}_{T,0} $ 
where $\widetilde{\Omega}_{T,0}$ denotes the \textit{scl-sp} estimator of \cite{DimiBayer2019}.
The third specification employs a standard HAC estimator \citep{NeweyWest1987, Andrews1991} as in \eqref{eqn:HACEstimation}, based on an automatic lag selection implemented in the \texttt{R} package \texttt{sandwich} \citep{Zeileis2004, Zeileis2006}.
The last specification combines the HAC estimator with the \textit{scl-sp} estimator of \cite{DimiBayer2019} by replacing the outer product estimator of the contemporaneous variance by the \textit{scl-sp} estimator as in \eqref{eqn:HAC-SCLSP-Estimation}.

Figure \ref{fig:sim_multi_VaRES_iter} shows the empirical rejection frequencies for the joint VaR and ES encompassing test based on the four different covariance estimators for $h$-step ahead forecasts with forecast horizons $h=1,2,5,10$.
Figure \ref{fig:sim_multi_VaRES_aggr} presents equivalent results for $h$-step aggregate forecasts.
For one-step ahead forecasts, the respective lines for the HAC and non-HAC estimators coincide, which stems from the fact that the automatic lag selection almost exclusively chooses no additional lag terms beyond the contemporaneous variance term and hence, for one-step ahead forecasts, our encompassing tests do not require HAC-corrected covariance estimators.
The different performance stems from the estimation of the contemporaneous variance, where the closed-form solution based on the \textit{scl-sp} estimator performs clearly superior to the outer product based version. 

For $h$-step ahead forecasts for larger forecast horizons, the covariance estimator combining the \textit{scl-sp} estimator with additional HAC terms performs (only) slightly superior to the raw \textit{scl-sp} estimator.
However, for inherently correlated multi-step ahead aggregate forecasts, presented in Figure \ref{fig:sim_multi_VaRES_aggr}, this deviance becomes more obvious, especially for increasing forecast horizons.
Consequently, for any $h>1$ for both, $h$-step ahead and aggregate forecasts, we use the \textit{scl-sp} estimator augmented with additional HAC-terms.

\section{Risk Models for the Empirical Application}
\label{sec:RiskModels}

In this section, we describe the (non-standard) risk models of \cite{Chen2012}, \cite{Taylor2019} and \cite{Patton2019} for forecasting VaR and ES in the empirical application in Section \ref{sec:Application}.

\cite{Chen2012} proposes to use GARCH models with innovations which follow an asymmetric Laplace distribution in order to capture potential (dynamic) skewness and heavy tails.
In particular,
\begin{align}
r_{t} & = \sigma_{t} \left(\varepsilon_{t}-\mu_{\varepsilon}\right),\qquad \varepsilon_{t}\overset{iid}{\sim}AL(0,1,p)\label{eqn:Chendist}
\end{align}
where $\sigma_{t}$ follows either a GARCH(1,1) or a GJR-GARCH(1,1) specification, and where $AL\left(0,1,p\right)$ represents the asymmetric Laplace distribution with zero mode, unit variance, and shape parameter $p$, which is defined such that $p=\mathbb{P} \left(\varepsilon_{t}<0\right)$.
The $AL\left(0,1,p\right)$ probability density function has the following form
\begin{equation}
f \left(\varepsilon;p\right)=b_{p} \, \exp \left[-b_{p}\left|\varepsilon\right|\left(\frac{1}{p}  \mathds{1}_{\{ \varepsilon < 0 \}} + \frac{1}{1-p} \mathds{1}_{\{ \varepsilon > 0 \}}  \right)\right],
\label{eqn:ALdensity}
\end{equation}
where $b_{p}=\sqrt{p^{2}+\left(1-p\right)^{2}}$, $\text{Var}[\varepsilon_{t}]=1$ and $\mathbb{E} [\varepsilon_{t}]=\mu_{\varepsilon}=\left(1-2p\right)/b_{p}$.
Thus, $u_{t}=\varepsilon_{t}-\mu_{\varepsilon}$ has an asymmetric Laplace distribution with zero mean, unit variance, and the shape parameter $p$. Note that $p=0.5$ implies a symmetric, standard Laplace distribution. 
If $p<0.5$, the density is skewed to the right, while the opposite applies for $p>0.5$. 
The VaR and ES (in the relevant area\footnote{In practice $p$ is usually close to 0.5 and $\alpha << 0.5$ is often chosen for VaR and ES in financial risk management.} $\alpha \in (0,p)$) can then be obtained analytically as 
\begin{align*}
\hat q_{t} = \sigma_{t} \frac{p}{b_{p}} \log\left(\frac{\alpha}{p}\right) - \mu_{\varepsilon} \sigma_{t}, 
\qquad \text{ and } \qquad 
\hat e_{t} = \hat q_{t} -\frac{\hat q_{t}}{\log\left(\frac{\alpha}{p}\right)}.
\end{align*}
The GARCH and GJR-GARCH models with a constant shape parameter $p$ are denoted by GARCH-AL and GJR-AL, respectively.

\cite{Chen2012} further propose to augment these models with a time-varying shape parameter, which allows for dynamic higher moments for $r_{t}$, and whose dynamics are given by
\begin{equation*}
p_{t} = \frac{1}{1+\sqrt{\frac{\xi_{t}}{\zeta_{t}}}}
\end{equation*}
where $\xi_{t} = \left(1-\lambda\right)\left|u_{t-1}\right| \mathds{1}_{\{ u_{t-1}\ge 0 \}} + \lambda\xi_{t-1}$, and
$\zeta_{t} = \left(1-\lambda\right)\left|u_{t-1}\right| \mathds{1}_{\{ u_{t-1} < 0 \}}  + \lambda\zeta_{t-1}$, for some smoothing parameter $0 \le \lambda \le 1$.
The models with a time-varying shape parameter $p_{t}$ are denoted as GARCH-AL-TVP and GJR-AL-TVP.

\cite{Taylor2019} employs semiparametric models to forecast VaR and ES by augmenting the CAViaR models of \cite{Engle2004} with an additional component for the ES.
In particular, the author assumes that the conditional quantile $\hat q_{t}$ at level $\alpha$ follows either the symmetric absolute value (SAV) or the asymmetric slope (AS) CAViaR models, 
\begin{align}
\text{SAV}: \quad \hat q_{t} & =\beta_{0}+\beta_{1} \hat q_{t-1} + \beta_{2} \left|r_{t-1}\right|, \qquad \text{and} \label{eqn:SAV_CAViaR}\\
\text{AS}: \quad\hat  q_{t} & =\beta_{0} +\beta_{1}  \hat q_{t-1}+\beta_{2} \left|r_{t-1}\right| \mathds{1}_{\{ r_{t-1}\ge 0 \}} + \beta_{3} \left|r_{t-1}\right| \mathds{1}_{\{ r_{t-1} < 0 \}}. \label{eqn:AS_CAViaR}
\end{align}
Since the dynamics of the VaR may not be the same as the dynamics of the ES, \cite{Taylor2019} equips these CAViaR models with the following ES specification
\begin{align}
\hat e_{t} & = \hat q_{t}-x_{t} \label{eqn:TaylorES} \\
x_{t} & =
\begin{cases}
\kappa_{0} + \kappa_{1} \left( \hat q_{t-1} - r_{t-1} \right) + \kappa_{2} x_{t-1}  \quad &\text{if} \, r_{t-1} \le \hat q_{t-1} \\
x_{t-1} \quad &\text{otherwise},
\end{cases}
\nonumber 
\end{align}
where $\kappa_{0} > 0$ and $\kappa_{1},\kappa_{2}\ge 0$ ensure that $\hat e_{t} < \hat q_{t}$ for $\hat q_{t} < 0$. 
The model specification given by \eqref{eqn:SAV_CAViaR} an \eqref{eqn:TaylorES} is denoted as the SAV-CAViaR-ES model, and the model specified by \eqref{eqn:AS_CAViaR} and \eqref{eqn:TaylorES}  as the AS-CAViaR-ES model.
These models are estimated by quasi-maximum likelihood based on the asymmetric Laplace distribution, which corresponds to a special case of the M-estimator considered by \cite{Patton2019}, and given in (\ref{eqn:JointLossESRegGeneral}) and (\ref{eqn:DefQn}) of this article.
In particular, \cite{Taylor2019} shows that under the assumption of a zero (conditional) mean, the (negative) of the asymmetric Laplace log-likelihood corresponds (up to constants) to the loss function in (\ref{eqn:JointLossESRegGeneral}) with $\mathfrak{g}(z) = 0$ and $\phi(z) = -\log(-z)$.

Finally, we consider the one factor GAS model of \cite{Patton2019} (also denoted by GAS-1F) which directly incorporates forcing variables into the dynamic process of the conditional variance in the sense of GAS models of \cite{Creal2013}.
In particular, 
\begin{align*}
\hat q_{t} &= a \exp(\hat \kappa_{t}) \qquad \text{ and } \qquad 
\hat e_{t} = b \exp(\hat \kappa_{t}), \quad \text{ where }\\
\hat \kappa_{t} &= \beta_0 + \beta_1 \hat \kappa_{t-1} + \frac{\beta_2}{\hat e_{t-1}} \left( \frac{r_{t}}{\alpha}  \mathds{1}_{\{r_{t} \le \hat q_{t-1} \}}  -  \hat e_{t-1} \right).
\end{align*}
and $\hat q_{t} = \sigma_{t} \kappa$ and $\hat e_{t} = \sigma_{t} \delta$ where the restrictions $\delta < \kappa < 0$ are imposed in the model estimation to ensure that $\hat e_{t} < \hat q_{t}$.
The model is estimated by M-estimator given in (\ref{eqn:JointLossESRegGeneral}) and (\ref{eqn:DefQn}).

Table \ref{tab:ParamEstimates} in Appendix \ref{sec:additional_tables_figures} reports parameter estimates of the risk models for the full sample.

\section{Absolute Forecast Evaluation}
\label{sec:AbsEvaluation}

Table \ref{tab:AbsEvaluation} shows absolute forecast evaluation criteria, including several backtests, for one-step ahead VaR and ES forecasts.
For this, the VaR Violation Ratio is given by $\hat{\alpha}/\alpha$, where $\hat{\alpha}=T^{-1}\sum_{t\in \mathfrak{T}}\mathds{1}_{\{ Y_{t+1}<\hat{q}_{t} \}}$ and the empirical ES ratio is computed as $\text{ESR} =\sum_{t \in \mathfrak{T}} \left.\left[Y_{t+1}\right.\mathds{1}_{\{ Y_{t+1}<\hat{q}_{t} \}}\right] / \sum_{t \in \mathfrak{T}} \left[ \hat{e}_{t}\mathds{1}_{\{ Y_{t+1}<\hat{q}_{t} \}}\right]$.
Furthermore, we report $p$-values of the unconditional coverage (UC) test of \cite{Kupiec1995}, the conditional coverage (CC) test of \cite{Christoffersen1998}, the dynamic quantile (DQ) test of \cite{Engle2004}, the VQR test of \cite{Gaglianone2011}, the ES backtest of \cite{McNeil2000} (MF), the regression-based ES backtest of \cite{BayerDimi2019} (BD), and for the calibration test of \cite{NoldeZiegel2017AAS} (NZ).
Table \ref{tab:AbsEvaluation} shows that six out of the eleven models pass all (are not rejected by any of the) seven backtests at a $5\%$ significance level, where the $p$-values in bold indicate that the null hypotheses of these tests are not rejected for any of the tests.

\section{Technical Details of the Proofs}

\begin{lemma}[Stochastic Equicontinuity of the Loss Function]
	\label{lemma:StochasticEquicontinuity}
	Given Assumption \ref{assu:AsymptoticTheory}, the function $T^{-1} l_T(\theta)$ is stochastically equicontinuous, i.e.
	for all $\varepsilon > 0$, there exists a $\delta > 0$, such that
	\begin{align}
	\underset{T \to \infty}{\lim \sup} \; \mathbb{P} \left[ \sup_{\{ \theta, \tilde \theta \in \Theta : || \tilde \theta - \theta|| < \delta\}} ||T^{-1}  l_T(\theta) - T^{-1} l_T(\tilde \theta) || > \varepsilon \right] < \varepsilon.
	\end{align}
\end{lemma}

\begin{proof}
	In the following, we show that for all $\theta, \tilde \theta \in \Theta$ and for all $T \in \mathbb{N}$, it holds that
	\begin{align}
	\label{eqn:LipschitzOP1}
	|T^{-1}  l_T(\theta) -T^{-1}  l_T(\tilde\theta) | \le K_T ||\theta- \tilde \theta||,
	\end{align}
	where $K_T =  \mathcal{O}_P(1)$, which implies stochastic equicontinuity by Theorem 21.10 of \cite{Davidson1994}.
	
	For this, we split the loss function
	\begin{align} 
	\begin{aligned}
	\label{eqn:SplitLossLipschitz}
	\rho_t(\theta) &=  \big( \mathds{1}_{\{Y_{t+h} \le g_t^q(\theta) \}} - \alpha \big) \mathfrak{g}(g_t^q(\theta)) - \mathds{1}_{\{Y_{t+h} \le g_t^q(\theta)\}}  \mathfrak{g}(Y_{t+h}) \\
	&+ \phi'(g_t^e(\theta)) \left( g_t^e( \theta) -  g_t^q(\theta) + \frac{( g_t^q(\theta) - Y_{t+h}) \mathds{1}_{\{Y_{t+h} \le  g_t^q(\theta) \}}}{\alpha}  \right) - \phi( g_t^e(\theta)) + a(Y_{t+h}) \\
	&=: A_t(\theta) \mathds{1}_{t+h}(\theta) + B_t(\theta) \mathds{1}_{t+h}(\theta) + C_t(\theta) + a(Y_{t+h}),
	\end{aligned}
	\end{align}
	where we use the short notation $\mathds{1}_{t+h}(\theta) := \mathds{1}_{\{Y_{t+h} \le g_t^q(\theta)\}}$ and
	\begin{align}
	A_t(\theta) &:= \mathfrak{g}(g_t^q(\theta))  - \mathfrak{g}(Y_{t+h}), \\
	B_t(\theta) &:=   \phi'(g_t^e(\theta))/\alpha \big( g_t^q(\theta) - Y_{t+h} \big), \qquad \text{and} \\
	C_t(\theta) &:=  \phi'(g_t^e(\theta)) \big(  g_t^e(\theta)  -  g_t^q(\theta)   \big) - \phi(g_t^e(\theta))  - \alpha \mathfrak{g}(g_t^q(\theta)).
	\end{align}
	It holds that
	\begin{align}
	\begin{aligned}
	\label{eqn:LipschitzSplit}
	| l_T(\theta) - l_T(\tilde\theta) | 
	&\le \big| A_t(\theta) \mathds{1}_{t+h}(\theta) - A_t(\tilde\theta) \mathds{1}_{t+h}(\tilde\theta) \big| \\
	&+  \big| B_t(\theta) \mathds{1}_{t+h}(\theta) - B_t(\tilde\theta) \mathds{1}_{t+h}(\tilde\theta) \big| \\
	&+  \big| C_t(\theta) - C_t(\tilde\theta) \big|.
	\end{aligned}
	\end{align}
	As $C_t$ is continuously differentiable, for the third term in (\ref{eqn:LipschitzSplit}) we get that
	\begin{align}
	\big| C_t(\theta) - C_t(\tilde\theta) \big| 
	\le \left( \sup_{\theta \in \Theta} ||\nabla_\theta  C_t( \theta) || \right) \cdot ||\theta- \tilde \theta||
	\end{align}
	where $\sup_{\theta \in \Theta} ||\nabla_\theta  C_t(\theta) || = \mathcal{O}_P(1)$ as $\mathbb{E} \left[ \sup_{\theta \in \Theta} ||\psi_t(\theta)||^{2r} \right] < \infty$ by condition \ref{cond:MomentCondition}.
	For the first term in (\ref{eqn:LipschitzSplit}), first notice that
	\begin{align}
	\big(\mathfrak{g}(g_t^q(\theta))  - \mathfrak{g}(Y_{t+h}) \big) \mathds{1}_{\{Y_{t+h} \le g_t^q(\theta)\}}
	= \frac{1}{2} \left( \mathfrak{g}(g_t^q(\theta))  - \mathfrak{g}(Y_{t+h}) +  \big| \mathfrak{g}(g_t^q(\theta))  - \mathfrak{g}(Y_{t+h}) \big| \right).
	\end{align}
	Thus, it holds that
	\begin{align}
	&\big| A_t(\theta) \mathds{1}_{t+h}(\theta) - A_t(\tilde\theta) \mathds{1}_{t+h}(\tilde\theta) \big| \\
	= \, &\frac{1}{2} \left| \left( \mathfrak{g}(g_t^q(\theta))  - \mathfrak{g}(Y_{t+h}) +  \big| \mathfrak{g}(g_t^q(\theta))  - \mathfrak{g}(Y_{t+h}) \big| \right) - \left( \mathfrak{g}(g_t^q(\tilde \theta))  - \mathfrak{g}(Y_{t+h}) +  \big| \mathfrak{g}(g_t^q(\tilde \theta))  - \mathfrak{g}(Y_{t+h}) \big| \right) \right| \\
	\le \, & \frac{1}{2} \left| \left( \mathfrak{g}(g_t^q(\theta))  - \mathfrak{g}(Y_{t+h})  \right) - \left( \mathfrak{g}(g_t^q(\tilde \theta))  - \mathfrak{g}(Y_{t+h})  \right) \right|
	+  \frac{1}{2} \left| \left| \mathfrak{g}(g_t^q(\theta))  - \mathfrak{g}(Y_{t+h})  \right| - \left| \mathfrak{g}(g_t^q(\tilde \theta))  - \mathfrak{g}(Y_{t+h})  \right| \right|\\
	\le \, &\left| \mathfrak{g}(g_t^q(\theta))  - \mathfrak{g}(g_t^q(\tilde \theta))   \right|  \\
	\le  \, &\left( \sup_{\theta \in \Theta} ||\nabla_\theta  \mathfrak{g}(g_t^q(\theta))|| \right) \cdot ||\theta- \tilde \theta||,
	\end{align}
	where $\sup_{\theta \in \Theta} ||\nabla_\theta  \mathfrak{g}(g_t^q(\theta))|| = \mathcal{O}_P(1)$ as $\mathbb{E} \left[ \sup_{\theta \in \Theta} ||\psi_t(\theta)||^{2r} \right] < \infty$.
	Equivalently, for the second term in (\ref{eqn:LipschitzSplit}), it holds that
	\begin{align}
	\frac{\phi'(g_t^e(\theta))}{\alpha}   \big(g_t^q(\theta)  - Y_{t+h} \big) \mathds{1}_{\{Y_{t+h} \le g_t^q(\theta)\}}
	= \frac{\phi'(g_t^e(\theta))}{2\alpha}  \left(  \big(g_t^q(\theta)  - Y_{t+h} \big)  - \big| g_t^q(\theta)  - Y_{t+h} \big| \right).
	\end{align}
	Consequently,
	\begin{align}
	&\big| B_t(\theta) \mathds{1}_{t+h}(\theta) - B_t(\tilde\theta) \mathds{1}_{t+h}(\tilde\theta) \big| \\
	= \, &\left|  \frac{\phi'(g_t^e(\theta))}{2\alpha}  \left(  \big(g_t^q(\theta)  - Y_{t+h} \big)  - \big| g_t^q(\theta)  - Y_{t+h} \big| \right) \right. \\
	&\quad - \left. \frac{\phi'(g_t^e(\tilde \theta))}{2\alpha}  \left(  \big(g_t^q(\tilde \theta)  - Y_{t+h} \big)  - \big| g_t^q(\tilde \theta)  - Y_{t+h} \big| \right) \right| \\
	\le \, &\left| \frac{\phi'(g_t^e(\theta))}{2\alpha}  \big(g_t^q(\theta)  - Y_{t+h} \big)  -\frac{\phi'(g_t^e(\tilde \theta))}{2\alpha}  \big(g_t^q(\tilde \theta)  - Y_{t+h} \big) \right| \\
	&\quad + \left| \frac{\phi'(g_t^e(\theta))}{2\alpha}  \big| g_t^q(\theta)  - Y_{t+h} \big|  -\frac{\phi'(g_t^e(\tilde \theta))}{2\alpha}  \big| g_t^q(\tilde \theta)  - Y_{t+h} \big| \right| \\
	\le \, &\left| \frac{\phi'(g_t^e(\theta))}{\alpha}  \big(g_t^q(\theta)  - Y_{t+h} \big)  -\frac{\phi'(g_t^e(\tilde \theta))}{\alpha}  \big(g_t^q(\tilde \theta)  - Y_{t+h} \big) \right| \\
	\le \, & \left( \sup_{\theta \in \Theta} \left| \left| \nabla_\theta  \left( \frac{\phi'(g_t^e(\theta))}{\alpha}  g_t^q( \theta) \right) +  \frac{ \nabla_\theta (\phi'(g_t^e( \theta)))}{\alpha}  Y_{t+h} \right| \right| \right) \cdot ||\theta- \tilde \theta||.
	\end{align}
	and $\sup_{\theta \in \Theta} \left| \left| \nabla_\theta  \left( \frac{\phi'(g_t^e(\theta))}{\alpha}  g_t^q( \theta) \right) +  \frac{ \nabla_\theta (\phi'(g_t^e( \theta)))}{\alpha}  Y_{t+h} \right| \right| = \mathcal{O}_P(1)$ as $\mathbb{E} \left[ \sup_{\theta \in \Theta} ||\psi_t(\theta)||^{2r} \right] < \infty$.
	Eventually, as $T^{-1}  l_T(\theta)  = T^{-1} \sumt \rho_t(\theta)$, the Lipschitz condition in (\ref{eqn:LipschitzOP1}) holds with $K_T = \mathcal{O}_P(1)$, which concludes this proof.
\end{proof}

\begin{lemma}(Type IV Class for Stochastic Equicontinuity of the Empirical Process)
	\label{lemma:TypeIVFunction}
	The class of functions given by $\rho_t(\theta) := \rho \big( Y_{t+h}, g^q_t(\theta) , g^e_t(\theta)  \big)$ in (\ref{eqn:JointLossESRegGeneral}) is a \textit{type IV class} (see \cite{Andrews1994}, p. 2278) with index $p=2r$ (in the notation of \cite{Andrews1994} and where $r>1$ from condition \ref{cond:BetaMixing}), i.e. it holds that
	\begin{align}
	\underset{1\le t \le T, \, T \ge 1 }{\sup}  \mathbb{E} \left[ \underset{\tilde{\theta} \in U(\theta, \delta) }{\sup} \left|  \rho_t(\theta)  - \rho_t(\tilde\theta)  \right|^{2r} \right]^{1/2r}
	\le C \delta,
	\end{align}
	for all $\theta \in \Theta$, for all $\delta > 0$ in a neighborhood of zero, and for some positive constant $C$.
\end{lemma}

\begin{proof}
	For this proof, we split the loss function
	\begin{align} 
	\begin{aligned}
	\label{eqn:SplitLossEquicont}
	l_t(\theta) &=  \big( \mathds{1}_{\{Y_{t+h} \le g_t^q(\theta) \}} - \alpha \big) \mathfrak{g}(g_t^q(\theta)) - \mathds{1}_{\{Y_{t+h} \le g_t^q(\theta)\}}  \mathfrak{g}(Y_{t+h}) \\
	&+ \phi'(g_t^e(\theta)) \left( g_t^e( \theta) -  g_t^q(\theta) + \frac{( g_t^q(\theta) - Y_{t+h}) \mathds{1}_{\{Y_{t+h} \le  g_t^q(\theta) \}}}{\alpha}  \right) - \phi( g_t^e(\theta)) + a(Y_{t+h}) \\
	&=: A_t(\theta) \mathds{1}_{t+h}(\theta) + B_t(\theta) \mathds{1}_{t+h}(\theta) Y_{t+h} + D_t(\theta) - \mathds{1}_{t+h}(\theta) \mathfrak{g}(Y_{t+h}) + a(Y_{t+h}),
	\end{aligned}
	\end{align}
	where $\mathds{1}_{t+h}(\theta) := \mathds{1}_{\{Y_{t+h} \le g_t^q(\theta)\}}$ and
	\begin{align}
	A_t(\theta) &:= \mathfrak{g}(g_t^q(\theta))  + \phi'(g_t^e(\theta)) g_t^q(\theta)/\alpha, \\
	B_t(\theta) &:=  - \phi'(g_t^e(\theta))/\alpha, \qquad \text{and}  \\	
	C_t(\theta) &:= - \alpha \mathfrak{g}(g_t^q(\theta)) + \phi'(g_t^e(\theta)) \big( g_t^e(\theta) - g_t^q(\theta) \big) - \phi(g_t^e(\theta)).
	\end{align}
	Thus, for all $\theta \in \Theta$, it holds that
	\begin{align}
	\begin{aligned}
	\label{eqn:SplitLossSquaredEquicont}
	\mathbb{E} \left[ \underset{\tilde{\theta} \in U(\theta, \delta) }{\sup} \left|  l_t(\theta)  - l_t(\tilde\theta)  \right|^{2r}  \right]^{1/{2r} }
	&\le \mathbb{E} \left[ \underset{\tilde{\theta} \in U(\theta, \delta) }{\sup} \left|  A_t(\theta) \mathds{1}_{t+h}(\theta)  - A_t(\tilde\theta) \mathds{1}_{t+h}(\tilde\theta)  \right|^{2r}  \right]^{1/{2r} } \\
	&+ \mathbb{E} \left[ \underset{\tilde{\theta} \in U(\theta, \delta) }{\sup} \left|  B_t(\theta) \mathds{1}_{t+h}(\theta) - B_t(\tilde\theta) \mathds{1}_{t+h}(\tilde\theta) \right|^{2r}  Y_{t+h}^{2r}   \right]^{1/{2r} } \\
	&+ \mathbb{E} \left[ \underset{\tilde{\theta} \in U(\theta, \delta) }{\sup} \left|  C_t(\theta)  - C_t(\tilde\theta)  \right|^{2r}  \right]^{1/{2r} } \\
	&+ \mathbb{E} \left[ \underset{\tilde{\theta} \in U(\theta, \delta) }{\sup} \left|  \mathds{1}_{t+h}(\theta)   - 	\mathds{1}_{t+h}(\tilde\theta) \right|^{2r}  |\mathfrak{g}(Y_{t+h})|^{2r}   \right]^{1/{2r} },
	\end{aligned}
	\end{align}
	by Minkowski's inequality (and as the $\sup$-operator follows the triangle inequality).
	We start by considering the first term in (\ref{eqn:SplitLossSquaredEquicont})
	\begin{align}
	&\mathbb{E} \left[ \underset{\tilde{\theta} \in U(\theta, \delta) }{\sup} \left|  A_t(\theta) \mathds{1}_{t+h}(\theta)  - A_t(\tilde\theta) \mathds{1}_{t+h}(\tilde\theta)  \right|^{2r}  \right]^{1/{2r} } \\
	\le \, &\mathbb{E} \left[ \underset{\tilde{\theta} \in U(\theta, \delta) }{\sup} \left|  A_t(\theta) \mathds{1}_{t+h}(\theta)  - A_t(\tilde\theta) \mathds{1}_{t+h}(\theta)  \right|^{2r}  \right]^{1/{2r} }
	+\mathbb{E} \left[ \underset{\tilde{\theta} \in U(\theta, \delta) }{\sup} \left|  A_t(\tilde \theta) \mathds{1}_{t+h}(\theta)  - A_t(\tilde\theta) \mathds{1}_{t+h}(\tilde\theta)  \right|^{2r}  \right]^{1/{2r} },
	\end{align}
	where the first term is bounded from above by 
	$\mathbb{E} \left[ \underset{\tilde{\theta} \in U(\theta, \delta) }{\sup} \left|\left|  \nabla_\theta A_t(\theta) \right| \right|^{2r}  \right]^{1/{2r} } \delta$.
	For the second term, we get that
	\begin{align}
	&\mathbb{E} \left[ \underset{\tilde{\theta} \in U(\theta, \delta) }{\sup} \left|  A_t(\tilde \theta) \mathds{1}_{t+h}(\theta)  - A_t(\tilde\theta) \mathds{1}_{t+h}(\tilde\theta)  \right|^{2r}  \right]^{1/{2r} } \\
	\le \, &\mathbb{E} \left[ \underset{\tilde{\theta} \in U(\theta, \delta) }{\sup} \left| A_t(\tilde \theta) \right|^{2r}  \mathbb{E}_t \left[  \underset{\tilde{\theta} \in U(\theta, \delta) }{\sup}  \left| \mathds{1}_{t+h}(\theta)  -  \mathds{1}_{t+h}(\tilde\theta)  \right|^{2r}  \right] \right]^{1/{2r} } \\
	\le \, &\mathbb{E} \left[ \underset{\tilde{\theta} \in U(\theta, \delta) }{\sup} \left| A_t(\tilde \theta) \right|^{2r}  \mathbb{E}_t \left[  \underset{\tilde{\theta} \in U(\theta, \delta) }{\sup}  \left|\left| \nabla_\theta g_t^q(\tilde \theta) h_t(g_t^q(\tilde \theta)) \right| \right|^{2r}  \right] \right]^{1/{2r} }  \delta.
	\end{align}
	by arguments as in the proof of Lemma B.1 of \cite{DimiBayer2019}.
	Similar reasons apply to the second term in in (\ref{eqn:SplitLossSquaredEquicont}), where by argument similar to equation (58) of \cite{DimiBayer2019}, 
	\begin{align}
	\mathbb{E}_t \left[  \underset{\tilde{\theta} \in U(\theta, \delta) }{\sup}  \left| \mathds{1}_{t+h}(\theta) Y_{t+h}  -  \mathds{1}_{t+h}(\tilde\theta) Y_{t+h} \right|^{2r}  \right]
	\le
	\underset{\tilde{\theta} \in U(\theta, \delta) }{\sup}  \left| \nabla_\theta g_t^q(\tilde \theta)  \big(g_t^q(\tilde \theta)\big)^{2r}   h_t(g_t^q(\tilde \theta)) \right| \delta.
	\end{align}
	Consequently, the second term is bounded by
	\begin{align}
	&\mathbb{E} \left[ \underset{\tilde{\theta} \in U(\theta, \delta) }{\sup} \left|  B_t(\theta) \mathds{1}_{t+h}(\theta) Y_{t+h}  - B_t(\tilde\theta) \mathds{1}_{t+h}(\tilde\theta)  Y_{t+h}\right|^{2r}  \right]^{1/{2r} } \\
	\le \, &\mathbb{E} \left[ \underset{\tilde{\theta} \in U(\theta, \delta) }{\sup} \left|  \nabla_\theta B_t(\tilde \theta)  Y_{t+h}  \right|^{2r} 
	+ \underset{\tilde{\theta} \in U(\theta, \delta) }{\sup} \left|  B_t(\tilde \theta) \nabla_\theta g_t^q(\tilde \theta)  g_t^q(\tilde \theta)^{2r}  h_t(g_t^q(\tilde \theta)) \right|^{2r}  \right]^{1/{2r} } d.
	\end{align}
	Equivalent argument apply to the fourth term in (\ref{eqn:SplitLossSquaredEquicont}), which is bounded by
	\begin{align}
	\mathbb{E} \left[  \underset{\tilde{\theta} \in U(\theta, \delta) }{\sup} \left|  B_t(\tilde \theta) \nabla_\theta g_t^q(\tilde \theta)  \mathfrak{g}(g_t^q(\tilde \theta))  h_t(g_t^q(\tilde \theta)) \right|^{2r}  \right]^{1/{2r} } d.
	\end{align}
	Eventually, for the third term in (\ref{eqn:SplitLossSquaredEquicont}) is bounded from above by 
	\begin{align}
	\mathbb{E} \left[ \underset{\tilde{\theta} \in U(\theta, \delta) }{\sup} \left|\left|  \nabla_\theta C_t(\theta) \right| \right|^{2r}  \right]^{1/{2r} } \delta.
	\end{align}
	As the respective moments are finite by condition \ref{cond:MomentCondition} for all $1\le t \le T$ and all $T \ge 1$, it follows that
	\begin{align}
	\underset{1\le t \le T, \, T \ge 1 }{\sup}  \mathbb{E} \left[ \underset{\tilde{\theta} \in U(\theta, \delta) }{\sup} \left|  l_t(\theta)  - l_t(\tilde\theta)  \right|^{2r} \right]^{1/{2r} }
	\le C \delta,
	\end{align}
	which concludes this proof.
\end{proof}

\begin{lemma}[Consistency of the HAC Estimator]
	\label{lemma:ConsistencyHAC}
	Given Assumption \ref{assu:AsymptoticTheory} and Assumption \ref{assu:AsymptoticTheory2}, it holds that $\hat{\mathcal{I}}_T \toP \mathcal{I}$.
\end{lemma}

\begin{proof}
	For this proof, we adapt the proof of \cite{NeweyWest1987} such that it allows for the discontinuity in $\psi_t(\theta)$.
	For this, we use a slightly different expansion than in equation (9) of \cite{NeweyWest1987} and we have to rely on a \textit{uniform} law of large numbers in order to establish the desired convergence.
	
	We start by showing the following uniform convergence for all $j \le T$,
	\begin{align}
	\label{eqn:HACUniformConvergence}
	\sup_{\theta \in \Theta} \left| \frac{1}{T} \sumtj \psi_t(\theta) \psi_{t-j}^\top (\theta)  - \mathbb{E} \big[ \psi_t(\theta) \psi_{t-j}^\top (\theta)  \big] \right| \toP 0.
	\end{align}
	For this, a pointwise law of large numbers (e.g., Corollary 3.48 of \cite{White2001}) holds as $\mathbb{E} \big[ ||\psi_t(\theta)||^{2 (\tilde r + \delta) } \big] < \infty$ for some $\delta > 0$ and the process follows the mixing condition from Assumption \ref{assu:AsymptoticTheory2}.
	Furthermore, Lemma \ref{lemma:HACStochasticEquicontinuous} shows that the function  $\frac{1}{T} \sumtj \psi_t(\theta) \psi_{t-j}^\top (\theta) $  is stochastically equicontinuous.
	Consequently, a uniform law of large numbers holds, see e.g. \cite{Andrews1992} for details.
	
	Consequently, by defining $\Psi_j^2(\theta) :=  \mathbb{E} \big[ \psi_t(\theta) \psi_{t-j}^\top (\theta)  \big]$, we get that
	\begin{align}
	&\left| \frac{1}{T} \sumtj \psi_t(\hat \theta_T) \psi_{t-j}^\top (\hat \theta_T)  - \mathbb{E} \big[ \psi_t(\theta^0) \psi_{t-j}^\top (\theta^0)  \big] \right| \\
	\le 	\, &\left| \frac{1}{T} \sumtj \psi_t(\hat \theta_T) \psi_{t-j}^\top (\hat \theta_T)  - \Psi_j^2(\hat \theta_T) \right|
	+ \left| \frac{1}{T} \sumtj \Psi_j^2(\hat \theta_T)  -\Psi_j^2(\theta^0) \right| \\
	\le 	\, & \sup_{\theta \in \Theta} \left| \frac{1}{T} \sumtj \psi_t( \theta) \psi_{t-j}^\top ( \theta)  - \Psi_j^2( \theta) \right|
	+ \left| \frac{1}{T} \sumtj \Psi_j^2(\hat \theta_T) -\Psi_j^2(\theta^0) \right|.
	\end{align}
	The first term converges to zero by (\ref{eqn:HACUniformConvergence}) and as the function $\Psi_j^2$ is continuous in $\theta$, the second term converges to zero by the continuous mapping theorem and as $\hat \theta_T$ is consistent.
	This also implies that for $T$ sufficiently large enough, it holds (with probability approaching one) that 
	\begin{align}
	\label{eqn:HACSupInequality}
	\left| \frac{1}{T} \sumtj \psi_t(\hat \theta_T) \psi_{t-j}^\top (\hat \theta_T)  - \mathbb{E} \big[ \psi_t(\theta^0) \psi_{t-j}^\top (\theta^0)  \big] \right|
	\le 2  \sup_{\theta \in \Theta} \left| \frac{1}{T} \sumtj \psi_t( \theta) \psi_{t-j}^\top ( \theta)  - \Psi_j^2( \theta) \right|.
	\end{align}
	Furthermore,  as $\mathbb{E} \left[ \sup_{\theta \in \Theta} \left| \left|  \psi_t( \theta) \psi_{t-j}^\top ( \theta)  - \Psi_j^2( \theta)  \right| \right|^{2(\tilde r+\delta)} \right] < \infty$ by assumption and as in equation (10) in the proof of \cite{NeweyWest1987}, we get that for all $j \ge 0$,
	\begin{align}
	\label{eqn:HACMomentInequality}
	\mathbb{E} \left[ \left( \sumtj \sup_{\theta \in \Theta} \left| \left|  \psi_t( \theta) \psi_{t-j}^\top ( \theta)  - \Psi_j^2( \theta) \right| \right| \right)^2 \right] 
	\le T (j+1) D^\ast,
	\end{align}
	for some finite constant $ D^\ast$.
	Consequently, for all $j \ge 1$,
	\begin{align}
	\begin{aligned}
	\label{eqn:HACNonsmoothTrick}
	&\mathbb{P} \left( \sum_{j=1}^{m_T} z(j,m_T)  \left| \left| \frac{1}{T} \sumtj \psi_t(\hat \theta_T) \psi_{t-j}^\top (\hat \theta_T)  - \mathbb{E} \big[ \psi_t(\theta^0) \psi_{t-j}^\top (\theta^0)  \big] \right|  \right| > \varepsilon \right) \\
	\le \, &\sum_{j=1}^{m_T} \mathbb{P} \left( \left| \left| \frac{1}{T} \sumtj \psi_t(\hat \theta_T) \psi_{t-j}^\top (\hat \theta_T)  - \mathbb{E} \big[ \psi_t(\theta^0) \psi_{t-j}^\top (\theta^0)  \big] \right| \right|  > \frac{\varepsilon}{C m_T} \right) \\
	\le \, &\sum_{j=1}^{m_T} \mathbb{P} \left( \sumtj \sup_{\theta \in \Theta}  \left| \left|  \psi_t( \theta) \psi_{t-j}^\top (\theta)  - \mathbb{E} \big[ \psi_t(\theta) \psi_{t-j}^\top (\theta)  \big] \right|  \right|  > \frac{\varepsilon T }{2 C m_T} \right) \\
	\le \, &\sum_{j=1}^{m_T} \mathbb{E} \left[ \left( \sumtj \sup_{\theta \in \Theta}   \left| \left| \psi_t( \theta) \psi_{t-j}^\top (\theta)  - \mathbb{E} \big[ \psi_t(\theta) \psi_{t-j}^\top (\theta)  \big] \right|  \right|  \right)^2 \right]    \frac{4 C^2 m_T^2}{T^2 \varepsilon^2}, \\
	\le \, &\sum_{j=1}^{m_T} T(j+1) D^\ast  \frac{4 C^2 m_T^2}{\varepsilon^2} =   \frac{4 D^\ast  C^2 }{\varepsilon^2} \, \frac{m_T^3 (m_T+3)}{T},
	\end{aligned}
	\end{align}
	where we employ (\ref{eqn:HACSupInequality}) in the second inequality, Markov's inequality in the penultimate line and (\ref{eqn:HACMomentInequality}) in the last line.
	The term in (\ref{eqn:HACNonsmoothTrick}) converges to zero as  $m_T^3 (m_T+3)/T \to 0$ as $m_T = o(T^{1/4})$.
	Now, similar to \cite{NeweyWest1987}, we split
	\begin{align}
	\left| \left| \hat{\mathcal{I}}_{T} (\theta) - \mathcal{I} \right|\right|
	&\le \left| \left| \frac{1}{T} \sumtj \psi_t(\hat \theta_T) \psi_{t}^\top (\hat \theta_T)  - \mathbb{E} \big[ \psi_t(\theta^0) \psi_{t}^\top (\theta^0)  \big] \right| \right| \\
	&+ 2 \sum_{j=1}^{m_T} z(j,m_T) \left( \left| \left| \frac{1}{T} \sumtj \psi_t(\hat \theta_T) \psi_{t-j}^\top (\hat \theta_T)  - \mathbb{E} \big[ \psi_t(\theta^0) \psi_{t-j}^\top (\theta^0)  \big] \right|  \right| \right) \\
	&+ 2 \sum_{j=1}^{m_T} |z(j,m_T) -1| \left| \left| \mathbb{E} \big[ \psi_t(\theta^0) \psi_{t-j}^\top (\theta^0) \big] \right|  \right|  \\
	&+ 2 \sum_{j=m_t+1}^{T} \left| \left| \mathbb{E} \big[ \psi_t(\theta^0) \psi_{t-j}^\top (\theta^0)  \big] \right|  \right|.
	\end{align}
	The terms in the first two lines converge to zero in probability by (\ref{eqn:HACSupInequality}) and (\ref{eqn:HACNonsmoothTrick}).
	The proofs for the terms in the last two lines equal the approach in the proof of Theorem 2 in \cite{NeweyWest1987}. 
	This concludes this proof.
\end{proof}

\begin{lemma}[Stochastic Equicontinuity for the HAC Estimator]
	\label{lemma:HACStochasticEquicontinuous}
	Given Assumption \ref{assu:AsymptoticTheory} and Assumption \ref{assu:AsymptoticTheory2}, the function $\frac{1}{T} \sumtj \psi_t(\theta) \psi_{t-j}^\top (\theta)$ is stochastically equicontinuous, where
	\begin{align}
	\psi_t(\theta) &=
	\nabla g^q_t(\theta) \left( \mathfrak{g}(g^q_t(\theta)) + \frac{\phi'(g^e_t(\theta))}{\alpha} \right) \left( \mathds{1}_{\{Y_{t+h} \le g^q_t(\theta) \}} - \alpha \right) \\
	& + \nabla g^e_t(\theta)  \phi''(g^e_t(\theta))  \left( g^e_t(\theta) - g^q_t(\theta)+ \frac{1}{\alpha} (g^q_t(\theta) - Y_{t+h}) \mathds{1}_{\{Y_{t+h} \le g^q_t(\theta) \}} \right).
	\end{align}
\end{lemma}

\begin{proof}
	We start by showing that the class of functions given by $\frac{1}{T} \sumtj \psi_t(\theta) \psi_{t-j}^\top (\theta)$ is a \textit{type IV class} (see \cite{Andrews1994}, p. 2278) with index $p=2 r$ (in the notation of \cite{Andrews1994} and where $\tilde r>1$ from condition \ref{cond:BetaMixing}), i.e. it holds that
	\begin{align}
	\underset{t,T}{\sup}  \; \mathbb{E} \left[ \underset{\tilde{\theta} \in U(\theta, \delta) }{\sup} \left| \left| \psi_t(\theta) \psi_{t-j}^\top (\theta)  - \psi_t(\tilde \theta) \psi_{t-j}^\top (\tilde \theta)  \right| \right|^{2r} \right]^{1/2r}
	\le C \delta,
	\end{align}
	for all $\theta \in \Theta$, for all $\delta > 0$ in a neighborhood of zero, and for some positive constant $C$.
	
	First notice that (for $j=0$),
	\begin{align}
	&\qquad \psi_t(\theta) \psi_{t}^\top (\theta) \\
	&= \big( \nabla g^q_t(\theta) \nabla^\top g^q_t(\theta) \big)   \left( \mathfrak{g}(g^q_t(\theta)) + \frac{\phi'(g^e_t(\theta))}{\alpha} \right)^2 \left( \mathds{1}_{\{Y_{t+h} \le g^q_t(\theta) \}} (1 - 2 \alpha) + \alpha^2 \right) \\
	&+  \big( \nabla g^e_t(\theta) \nabla^\top g^e_t(\theta) \big) \phi''(g^e_t(\theta))^2  \left( g^e_t(\theta) - g^q_t(\theta)+ \frac{1}{\alpha} (g^q_t(\theta) - Y_{t+h}) \mathds{1}_{\{Y_{t+h} \le g^q_t(\theta) \}} \right)^2  \\
	&+ 2 \big( \nabla g^q_t(\theta) \nabla^\top g^e_t(\theta) \big) \left( \mathfrak{g}(g^q_t(\theta)) + \frac{\phi'(g^e_t(\theta))}{\alpha} \right) \left( \mathds{1}_{\{Y_{t+h} \le g^q_t(\theta) \}} - \alpha \right) \times \\
	&\qquad \phi''(g^e_t(\theta))  \left( g^e_t(\theta) - g^q_t(\theta)+ \frac{1}{\alpha} (g^q_t(\theta) - Y_{t+h}) \mathds{1}_{\{Y_{t+h} \le g^q_t(\theta) \}} \right) \\
	&=: \tilde A_t(\theta) \mathds{1}_{\{Y_{t+h} \le g^q_t(\theta) \}} + \tilde B_t(\theta),
	\end{align}
	where 
	\begin{align}
	\begin{aligned}
	\label{eqn:DefTildeBt}
	\tilde B_t(\theta) 
	&:=  \big( \nabla g^q_t(\theta) \nabla^\top g^q_t(\theta) \big) \alpha^2 \left( \mathfrak{g}(g^q_t(\theta)) + \frac{\phi'(g^e_t(\theta))}{\alpha} \right)^2 \\
	& + \big( \nabla g^e_t(\theta) \nabla^\top g^e_t(\theta) \big) \phi''(g^e_t(\theta))^2  \big(g^e_t(\theta) - g^q_t(\theta)\big)^2 \\
	& + 2 \big( \nabla g^q_t(\theta) \nabla^\top g^e_t(\theta) \big)  \phi''(g^e_t(\theta))\left( \mathfrak{g}(g^q_t(\theta)) + \frac{\phi'(g^e_t(\theta))}{\alpha} \right)  \alpha \big(g^q_t(\theta) - g^e_t(\theta)\big),
	\end{aligned}
	\end{align}
	and 
	\begin{align}
	\begin{aligned}
	\label{eqn:DefTildeAt}
	\tilde A_t(\theta) 
	&:= \big( \nabla g^q_t(\theta) \nabla^\top g^q_t(\theta) \big) \left( \mathfrak{g}(g^q_t(\theta)) + \frac{\phi'(g^e_t(\theta))}{\alpha} \right)^2 ( 1-2\alpha) \\
	& + \big( \nabla g^e_t(\theta) \nabla^\top g^e_t(\theta) \big) \phi''(g^e_t(\theta))^2 \left[ \frac{1}{\alpha} (g^q_t(\theta) - Y_{t+h})^2 + \frac{2}{\alpha} \big(g^e_t(\theta) - g^q_t(\theta)\big) (g^q_t(\theta) - Y_{t+h}) \right] \\
	& + 2 \big( \nabla g^q_t(\theta) \nabla^\top g^e_t(\theta) \big) \phi''(g^e_t(\theta))\left( \mathfrak{g}(g^q_t(\theta)) + \frac{\phi'(g^e_t(\theta))}{\alpha} \right)
	\left[ \big(g^e_t(\theta) - g^q_t(\theta)\big) + \frac{1-\alpha}{\alpha} (g^q_t(\theta) - Y_{t+h}) \right].
	\end{aligned}
	\end{align}
	Further notice that both, $\tilde A_t(\theta)$ and $\tilde B_t(\theta)$ are continuously differentiable.
	In the following, we use the short notation $\mathds{1}_{t+h}(\theta) = \mathds{1}_{\{Y_{t+h} \le g^q_t(\theta) \}}$.
	For all $\theta \in \Theta$, it holds that
	\begin{align}
	\begin{aligned}
	\label{eqn:SplitSquaredEquicont}
	\mathbb{E} \left[ \underset{\tilde{\theta} \in U(\theta, \delta) }{\sup} \left|  \psi_t(\theta) \psi_{t}^\top (\theta)  - \psi_t(\tilde \theta) \psi_{t}^\top (\tilde \theta) \right|^{2r}  \right]^{1/{2r} }
	&\le \mathbb{E} \left[ \underset{\tilde{\theta} \in U(\theta, \delta) }{\sup} \left|  \tilde A_t(\theta) \mathds{1}_{t+h}(\theta)  - \tilde A_t(\tilde\theta) \mathds{1}_{t+h}(\tilde\theta)  \right|^{2r}  \right]^{1/{2r} } \\
	&+ \mathbb{E} \left[ \underset{\tilde{\theta} \in U(\theta, \delta) }{\sup} \left| \tilde B_t(\theta)  - \tilde B_t(\tilde\theta) \right|^{2r}    \right]^{1/{2r} }.
	\end{aligned}
	\end{align}
	
	We start by considering the first term in (\ref{eqn:SplitSquaredEquicont}),
	\begin{align}
	&\mathbb{E} \left[ \underset{\tilde{\theta} \in U(\theta, \delta) }{\sup} \left|  \tilde A_t(\theta) \mathds{1}_{t+h}(\theta)  - \tilde A_t(\tilde\theta) \mathds{1}_{t+h}(\tilde\theta)  \right|^{2r}  \right]^{1/{2r} } \\
	\le \, &\mathbb{E} \left[ \underset{\tilde{\theta} \in U(\theta, \delta) }{\sup} \left|  \tilde A_t(\theta) \mathds{1}_{t+h}(\theta)  - \tilde A_t(\tilde\theta) \mathds{1}_{t+h}(\theta)  \right|^{2r}  \right]^{1/{2r} }
	+\mathbb{E} \left[ \underset{\tilde{\theta} \in U(\theta, \delta) }{\sup} \left|  \tilde A_t(\tilde \theta) \mathds{1}_{t+h}(\theta)  - \tilde A_t(\tilde\theta) \mathds{1}_{t+h}(\tilde\theta)  \right|^{2r}  \right]^{1/{2r} },
	\end{align}
	where the first term is bounded from above by 
	$\mathbb{E} \left[ \underset{\tilde{\theta} \in U(\theta, \delta) }{\sup} \left|\left|  \nabla_\theta \tilde A_t(\theta) \right| \right|^{2r}  \right]^{1/{2r} } \delta$.
	For the second term, we get that
	\begin{align}
	&\mathbb{E} \left[ \underset{\tilde{\theta} \in U(\theta, \delta) }{\sup} \left|  \tilde A_t(\tilde \theta) \mathds{1}_{t+h}(\theta)  - A_t(\tilde\theta) \mathds{1}_{t+h}(\tilde\theta)  \right|^{2r}  \right]^{1/{2r} } \\
	\le \, &\mathbb{E} \left[ \underset{\tilde{\theta} \in U(\theta, \delta) }{\sup} \left| \tilde A_t(\tilde \theta) \right|^{2r}  \mathbb{E}_t \left[  \underset{\tilde{\theta} \in U(\theta, \delta) }{\sup}  \left| \mathds{1}_{t+h}(\theta)  -  \mathds{1}_{t+h}(\tilde\theta)  \right|^{2r}  \right] \right]^{1/{2r} } \\
	\le \, &\mathbb{E} \left[ \underset{\tilde{\theta} \in U(\theta, \delta) }{\sup} \left| \tilde A_t(\tilde \theta) \right|^{2r}  \mathbb{E}_t \left[  \underset{\tilde{\theta} \in U(\theta, \delta) }{\sup}  \left|\left| \nabla_\theta g_t^q(\tilde \theta) h_t(g_t^q(\tilde \theta)) \right| \right|^{2r}  \right] \right]^{1/{2r} }  \delta.
	\end{align}
	by arguments as in the proof of Lemma B.1 of \cite{DimiBayer2019}.
	Eventually, for the second term in (\ref{eqn:SplitSquaredEquicont}) is bounded from above by 
	\begin{align}
	\mathbb{E} \left[ \underset{\tilde{\theta} \in U(\theta, \delta) }{\sup} \left|\left|  \nabla_\theta \tilde B_t(\theta) \right| \right|^{2r}  \right]^{1/{2r} } \delta.
	\end{align}
	The proofs for $j \ge 1$ are equivalent and omitted here.
	Consequently, the set of functions by given $\frac{1}{T} \sumtj \psi_t(\theta) \psi_{t-j}^\top (\theta)$ are a \textit{type IV class} of \cite{Andrews1994} with index $p=2r > 2$.
	Consequently, by Theorem 5 of \cite{Andrews1994}, it satisfies "Ossiander's $L^{2r}$-entropy" condition and thus, it has a "$L^{2 \tilde r}$-envelope" given by their supremum.
	Consequently, we can apply Theorem 1 (and Application 1) of \cite{Doukhan1995} and obtain that $\frac{1}{T} \sumtj \psi_t(\theta) \psi_{t-j}^\top (\theta)$ is stochastically equicontinuous (see the Remark on p.410 of \cite{Doukhan1995}).
\end{proof}

\pagebreak	
\section{Additional Tables and Figures}
\label{sec:additional_tables_figures}

\begin{table}[tbh]
	\footnotesize
	\centering
	\caption{Empirical Sizes for the GAS processes}
	\label{tab:Size_GAS}
	\begin{tabularx}{0.95\linewidth}{l l @{\hspace{0.2cm}} rr l @{\hspace{0.1cm}} rr l @{\hspace{0.5cm}}  rr  l @{\hspace{0.1cm}} rr }
		\hline
		\hline
		\addlinespace
		& & \multicolumn{2}{c}{$\mathbb{H}_0^{(1)}$} & & \multicolumn{2}{c}{$\mathbb{H}_0^{(2)}$} & &  \multicolumn{2}{c}{$\mathbb{H}_0^{(1)}$} & & \multicolumn{2}{c}{$\mathbb{H}_0^{(2)}$}\\
		\cmidrule(lr){3-4} \cmidrule(lr){6-7} \cmidrule(lr){9-10} \cmidrule(lr){12-13}
		& &  VaR ES &  Aux ES  & &  VaR ES &  Aux ES & &   VaR ES &  Aux ES  & &  VaR ES &  Aux ES  \\ 
		\midrule
		& & \multicolumn{11}{c}{Linear link function} \\
		\midrule
		$T$ & &\multicolumn{5}{c}{GAS-t}   &  &\multicolumn{5}{c}{VaR/ES GAS}\\
		\cmidrule(lr){3-7} \cmidrule(lr){9-13}
		$250$  &        & 31.30  & 20.95  &        & 23.30  & 15.00 &        & 30.75  & 21.80  &        & 27.00  & 17.40 \\
$500$  &        & 23.45  & 15.25  &        & 15.75  & 12.25 &        & 24.30  & 19.35  &        & 19.55  & 12.35 \\
$1000$ &        & 14.80  &  9.45  &        & 12.10  &  9.25 &        & 18.70  & 13.05  &        & 15.90  & 10.60 \\
$2500$ &        & 12.70  &  7.85  &        &  8.80  &  5.85 &        & 11.75  &  9.90  &        & 12.30  &  7.25 \\
$5000$ &        &  9.60  &  5.35  &        &  8.60  &  5.75 &        &  8.60  &  7.85  &        &  9.30  &  5.85 \\
 
		\midrule
		& & \multicolumn{11}{c}{Convex link function} \\
		\midrule
		$T$ & &\multicolumn{5}{c}{GAS-t} &   &\multicolumn{5}{c}{VaR/ES GAS}\\
		\cmidrule(lr){3-7} \cmidrule(lr){9-13}
		$250$  &        & 20.62  & 17.07  &        & 12.86  & 9.80  &        &  8.30  &  7.99  &        & 10.26  & 8.05  \\
$500$  &        & 17.87  & 16.57  &        &  9.70  & 8.55  &        & 15.19  & 17.41  &        & 10.31  & 7.85  \\
$1000$ &        & 11.93  & 10.88  &        &  6.83  & 6.38  &        & 15.96  & 16.79  &        &  8.00  & 6.20  \\
$2500$ &        & 10.71  & 10.61  &        &  5.90  & 5.60  &        & 10.67  & 12.02  &        &  7.30  & 4.65  \\
$5000$ &        &  8.69  &  7.94  &        &  5.47  & 5.82  &        & 10.86  & 11.37  &        &  4.85  & 3.15  \\
  
		\midrule
		& & \multicolumn{11}{c}{No-crossing link function} \\
		\midrule
		$T$ & &\multicolumn{5}{c}{GAS-t} &   &\multicolumn{5}{c}{VaR/ES GAS}\\
		\cmidrule(lr){3-7} \cmidrule(lr){9-13}
		$250$  &        & 13.13  & 11.77  &        &  9.95  & 7.25  &        &  6.87  &  5.87  &        &  8.00  & 7.55  \\
$500$  &        & 11.44  & 10.34  &        & 10.55  & 6.95  &        & 10.34  & 11.05  &        &  9.55  & 8.85  \\
$1000$ &        &  7.87  &  7.72  &        &  9.90  & 6.25  &        & 10.87  & 10.82  &        & 10.25  & 9.05  \\
$2500$ &        &  7.18  &  6.38  &        & 10.10  & 6.35  &        &  8.15  &  8.05  &        & 11.95  & 8.50  \\
$5000$ &        &  7.47  &  6.11  &        &  8.55  & 6.10  &        &  8.36  &  8.31  &        &  9.00  & 6.55  \\ 
		\addlinespace
		\hline
		\hline 
		\addlinespace
		\multicolumn{13}{p{.90\linewidth}}{\textit{Notes:} This table shows the empirical sizes for the encompassing tests for one-step ahead forecasts stemming from the two additional DGPs described in Section \ref{sec:add_simulation}, the three link functions, the joint VaR and ES (VaR ES) and auxiliary ES (Aux ES) test and both null hypotheses with a nominal size of $5\%$. 
			The columns denoted by ``GAS-t'' contain results for the GARCH(1,1) model with normal innovations and a GAS-$t$ model, whereas those labeled ``VaR/ES GAS'' report results for the one and two factor GAS models introduced by \cite{Patton2019}.}
	\end{tabularx}
\end{table}

\begin{table}[tbh]
	\footnotesize
	\centering
	\caption{Empirical Sizes for Multi-Step Forecasts}
	\label{tab:Multistep_AuxES}
	\begin{tabularx}{0.75\linewidth}{l l @{\hspace{0.5cm}} rrrr l @{\hspace{1cm}} rrrr }
		\hline\hline
		\addlinespace
		& & \multicolumn{4}{c}{$\mathbb{H}_0^{(1)}$} & & \multicolumn{4}{c}{$\mathbb{H}_0^{(2)}$} \\
		\cmidrule(lr){3-6} \cmidrule(lr){8-11} 
		$h$ & &  1 &  2  & 5 & 10 &   &  1 &  2  & 5 & 10  \\ 
		\midrule
		$T$ & &\multicolumn{9}{c}{$h$-step ahead forecasts}\\
		\cmidrule(lr){3-11}
		$250$  &        &  8.34  & 10.26  &  9.26  &  5.53  &        & 6.96   & 6.96   & 7.26   & 4.64  \\
$500$  &        &  8.31  &  8.71  & 12.56  & 10.63  &        & 4.70   & 5.42   & 6.07   & 6.43  \\
$1000$ &        &  6.83  &  7.00  & 10.16  & 13.25  &        & 3.21   & 4.22   & 4.81   & 6.71  \\
$2500$ &        &  4.30  &  4.50  &  6.60  & 10.94  &        & 3.91   & 3.92   & 4.91   & 7.03  \\
$5000$ &        &  3.60  &  4.83  &  5.82  &  9.50  &        & 3.50   & 3.30   & 5.11   & 6.31  \\
 
		\midrule
		$T$ & &\multicolumn{9}{c}{$h$-step aggregate forecasts}\\
		\cmidrule(lr){3-11}
		$250$  &        &  8.44  & 13.65  & 21.30  & 29.91  &        &  6.96  & 10.14  & 19.80  & 25.76 \\
$500$  &        &  8.81  & 12.51  & 20.08  & 30.74  &        &  4.90  &  8.27  & 16.37  & 22.25 \\
$1000$ &        &  6.93  &  8.63  & 16.82  & 25.63  &        &  3.31  &  5.45  & 13.33  & 18.07 \\
$2500$ &        &  4.20  &  7.04  & 10.53  & 20.08  &        &  3.81  &  4.02  &  7.37  & 10.63 \\
$5000$ &        &  3.50  &  4.92  &  9.12  & 15.36  &        &  3.30  &  2.91  &  5.22  &  7.92 \\
 
		\addlinespace
		\hline\hline
		\addlinespace
		\multicolumn{11}{p{.72\linewidth}}{\textit{Notes:} This table shows the empirical sizes for the auxiliary ES encompassing test for the $h$-step ahead and the $h$-step aggregate forecasts and both null hypotheses with a nominal size of $5\%$. It shows the results for the GARCH specification with normal innovations and the convex link function.
		}
	\end{tabularx}
\end{table}

\begin{table}[tbh]
	\footnotesize
	\centering
	\caption{Parameter Estimates of the Risk Models for the Empirical Application}
	\label{tab:ParamEstimates}
	\begin{tabular}{lccccccccc}
		\hline\hline
		\addlinespace
		{Volatility Models} & $\beta_{0}$ & $\beta_{1}$ & $\beta_{2}$ & $\beta_{3}$ & $v$ & $\lambda$ & $p$ & $a$ & $b$\\
		\addlinespace
		\hline 
		\addlinespace
		GARCH-N & {0.023} & {0.859} & {0.125} &  &  &  &  &  & \\
		\addlinespace 
		GJR-ST & {0.018} & {0.879} & {0.001} & {0.218} & {7.364} & {0.869} &  &  & \\
		\addlinespace
		GARCH-AL & {0.020} & {0.871} & {0.129} &  &  &  & {0.545} &  & \\
		\addlinespace 
		GJR-AL & {0.022} & {0.887} & {-0.021} & {0.267} &  &  & {0.560} &  & \\
		\addlinespace
		GARCH-AL-TVP & {0.020} & {0.870} & {0.130} &  &  & {0.980} &  &  & \\
		\addlinespace
		GJR-AL-TVP & {0.022} & {0.889} & {-0.020} & {0.262} &  & {0.979} &  &  & \\
		\addlinespace
		GAS-1F &  & {0.930} & {-0.003} & {0.034} &  &  &  & {-1.449} & {-1.848}\\
		\addlinespace
		\hline\hline
		\addlinespace
		{CAViaR-ES Models} & $\beta_{0} $ & $\beta_{1} $ & $\beta_{2}$ & $\beta_{3}$ & $\kappa_{0}$ & $\kappa_{1}$ & $\kappa_{2}$ &  & \\
		\addlinespace
		\hline 
		\addlinespace
		SAV & {-0.099} & {0.841} & {-0.337} &  & {-1.233} &  &  &  & \\
		\addlinespace
		AS & {-0.072} & {0.889} & {-0.004} & {-0.436} & {0.006} & {0.890} & {0.113} &  & \\
		\hline 
		\hline
		\addlinespace
		\multicolumn{10}{p{0.92\linewidth}}{\textit{Notes:} The entries in this table show parameter estimates from the risk models described in Section \ref{sec:onestep_forecasts} and Appendix \ref{sec:RiskModels} for the full sample.}
	\end{tabular}
\end{table}

\begin{landscape}
	\begin{table}[tbh]
		\footnotesize
		\centering
		\caption{Correlations of VaR and ES One-Step Ahead Forecasts}
		\label{tab:corr_one}
		\begin{tabularx}{0.92\linewidth}{l @{\hspace{0.3cm}}ccccccccccc}
			\hline\hline
			\addlinespace
			& \multicolumn{11}{c}{Correlations of VaR Forecasts} \\
			\addlinespace
			\cmidrule(lr){2-12}
			\addlinespace
			& {Hist} & {Risk} & GARCHN & GJR & {GARCH} & {GJR} & {GARCH} & {GJR} & & SAV  & AS \\
			& {Sim} & {Metrics} & {-N} & {-ST} & {-AL} & {-AL} & {-AL-TVP} & {-AL-TVP} & {GAS-1F} & {CAViaR-ES} & {CAViaR-ES} \\
			\addlinespace
			\hline 
			\addlinespace
			{Hist Sim} & {1.000} & {0.439} & {0.340} & {0.244} & {0.360} & {0.256} & {0.375} & {0.222} & {0.345} & {0.292} & {0.200}\\
			{RiskMetrics} &  & {1.000} & {0.944} & {0.866} & {0.954} & {0.862} & {0.990} & {0.845} & {0.847} & {0.933} & {0.844}\\
			{GARCH-N} &  &  & {1.000} & {0.956} & {0.997} & {0.945} & {0.968} & {0.934} & {0.883} & {0.978} & {0.932}\\
			{GJR-ST} &  &  &  & {1.000} & {0.944} & {0.994} & {0.909} & {0.995} & {0.873} & {0.929} & {0.975}\\
			{GARCH-AL} &  &  &  &  & {1.000} & {0.936} & {0.973} & {0.921} & {0.883} & {0.977} & {0.921}\\
			{GJR-AL} &  &  &  &  &  & {1.000} & {0.903} & {0.995} & {0.880} & {0.921} & {0.969}\\
			{GARCH-AL-TVP} &  &  &  &  &  &  & {1.000} & {0.893} & {0.863} & {0.952} & {0.888}\\
			{GJR-AL-TVP} &  &  &  &  &  &  &  & {1.000} & {0.867} & {0.907} & {0.972}\\
			{GAS-1F} &  &  &  &  &  &  &  &  & {1.000} & {0.832} & {0.812}\\
			{SAV-CAViaR-ES} &  &  &  &  &  &  &  &  &  & {1.000} & {0.928}\\
			{AS-CAViaR-ES} &  &  &  &  &  &  &  &  &  &  & {1.000}\\
			\addlinespace
			\hline\hline
			\addlinespace
			& \multicolumn{11}{c}{Correlations of ES Forecasts} \\
			\addlinespace
			\cmidrule(lr){2-12}
			\addlinespace
			& {Hist} & {Risk} & GARCHN & GJR & {GARCH} & {GJR} & {GARCH} & {GJR} & & SAV  & AS \\
			& {Sim} & {Metrics} & {-N} & {-ST} & {-AL} & {-AL} & {-AL-TVP} & {-AL-TVP} & {GAS-1F} & {CAViaR-ES} & {CAViaR-ES} \\
			\addlinespace
			\hline 
			\addlinespace
			{Hist Sim} & {1.000} & {0.492} & {0.400} & {0.273} & {0.416} & {0.316} & {0.437} & {0.286} & {0.385} & {0.331} & {0.288}\\
			{RiskMetrics} &  & {1.000} & {0.944} & {0.853} & {0.954} & {0.861} & {0.991} & {0.846} & {0.836} & {0.901} & {0.838}\\
			{GARCH-N} &  &  & {1.000} & {0.948} & {0.998} & {0.944} & {0.968} & {0.935} & {0.878} & {0.924} & {0.921}\\
			{GJR-ST} &  &  &  & {1.000} & {0.936} & {0.991} & {0.897} & {0.994} & {0.874} & {0.864} & {0.953}\\
			{GARCH-AL} &  &  &  &  & {1.000} & {0.936} & {0.974} & {0.923} & {0.874} & {0.925} & {0.910}\\
			{GJR-AL} &  &  &  &  &  & {1.000} & {0.901} & {0.996} & {0.880} & {0.870} & {0.953}\\
			{GARCH-AL-TVP} &  &  &  &  &  &  & {1.000} & {0.891} & {0.855} & {0.920} & {0.880}\\
			{GJR-AL-TVP} &  &  &  &  &  &  &  & {1.000} & {0.876} & {0.855} & {0.954}\\
			{GAS-1F} &  &  &  &  &  &  &  &  & {1.000} & {0.807} & {0.837}\\
			{SAV-CAViaR-ES} &  &  &  &  &  &  &  &  &  & {1.000} & {0.891}\\
			{AS-CAViaR-ES} &  &  &  &  &  &  &  &  &  &  & {1.000}\\
			\addlinespace
			\hline 
			\hline
			\addlinespace
			\multicolumn{12}{p{0.9\linewidth}}{\textit{Notes:} This table contains the pairwise correlations for one-step ahead VaR and ES forecasts stemming from the eleven considered risk models.}
		\end{tabularx}
	\end{table}
\end{landscape}

\begin{landscape}
	\begin{table}[tbh]
		\footnotesize
		\centering
		\caption{Detailed Encompassing Test Results for One-Step Ahead Forecasts}
		\label{tab:JointOne}
		\begin{tabularx}{0.86\linewidth}{l @{\hspace{0.5cm}} cccccc }
			\hline \hline
			\addlinespace
			\multicolumn{7}{c}{Joint VaR and ES Test} \\
			\addlinespace
			\hline
			\addlinespace
			& {GJR-ST} & {GARCH-AL} & {GARCH-AL-TVP} & {GAS-1F} & {SAV-CAViaR-ES} & {AS-CAViaR-ES}\\
			\addlinespace
			\hline 
			\addlinespace
			{GJR-ST} & {} & {(0.89, 1.00)} & {(0.84, 1.00)} & {(1.00, 1.00)} & {(0.78, 1.00)} & {(0.70, 1.00)}\\
			\addlinespace
			{GARCH-AL} & {(0.11, 0.00){*}} & {} & {(0.90, 1.00)} & {(0.86, 0.97)} & {(0.68, 1.00)} & {(0.45, 0.46){*}}\\
			\addlinespace
			{GARCH-AL-TVP} & {(0.16, 0.00){*}} & {(0.10, 0.00){*}} & {} & {(0.74, 0.82)} & {(0.53, 0.45){*}} & {(0.37, 0.11){*}}\\
			\addlinespace
			{GAS-1F} & {(0.00, 0.00){*}} & {(0.014, 0.03){*}} & {(0.26, 0.18){*}} & {} & {(0.23, 0.33){*}} & {(0.06, 0.09){*}}\\
			\addlinespace
			{SAV-CAViaR-ES} & {(0.22, 0.00){*}} & {(0.32, 0.00){*}} & {(0.47, 0.55){*}} & {(0.77, 0.67){*}} & {} & {(0.23, 0.11){*}}\\
			\addlinespace
			{AS-CAViaR-ES} & {(0.30, 0.00){*}} & {(0.55, 0.54){*}} & {(0.63, 0.89)} & {(0.94, 0.91)} & {(0.77, 0.89)} & {}\\
			\addlinespace
			\hline 
			\hline
			\\
			\addlinespace
			\multicolumn{7}{c}{Auxiliary ES Test} \\
			\addlinespace
			\hline
			\addlinespace
			& {GJR-ST} & {GARCH-AL} & {GARCH-AL-TVP} & {GAS-1F} & {SAV-CAViaR-ES} & {AS-CAViaR-ES}\\
			\addlinespace
			\hline 
			\addlinespace
			{GJR-ST} & {} & {1.00} & {1.00} & {1.00} & {1.00} & {1.00}\\
			\addlinespace
			{GARCH-AL} & {0.00{*}} & {} & {1.00} & {0.97} & {1.00} & {0.46{*}}\\
			\addlinespace
			{GARCH-AL-TVP} & {0.00{*}} & {0.00{*}} & {} & {0.82} & {0.45} & {0.11{*}}\\
			\addlinespace
			{GAS-1F} & {0.00{*}} & {0.03{*}} & {0.18{*}} & {} & {0.33{*}} & {0.09{*}}\\
			\addlinespace
			{SAV-CAViaR-ES} & {0.00{*}} & {0.00{*}} & {0.55} & {0.67} & {} & {0.11{*}}\\
			\addlinespace
			{AS-CAViaR-ES} & {0.00{*}} & {0.54} & {0.89} & {0.91} & {0.89} & {}\\
			\addlinespace
			\hline 
			\hline 
			\addlinespace
			\multicolumn{7}{p{0.84\linewidth}}{\textit{Notes:} This table reports the estimates of the convex combination parameters $\left(\theta_{1},\theta_{2}\right)$ from the convex link function with intercepts for each pair of models.
				The symbol $^\ast$ indicates that the null hypothesis that the VaR and ES forecasts of a row-heading model jointly encompasses those of a column-heading model is rejected at the $5\%$ significance level.}
		\end{tabularx}
	\end{table}
\end{landscape}

\begin{table}[tbh]
	\footnotesize
	\centering
	\caption{Correlations of VaR and ES Multi-Step Forecasts}
	\label{tab:corr_ahead}
	\begin{tabularx}{0.88\linewidth}{l @{\hspace{0.2cm}}ccccccc}
		\hline\hline
		\addlinespace
		& {Risk} & {GARCH} & GJR & {GARCH} & {GJR} & {GARCH} & {GJR} \\
		& {Metrics} & {-N} & {-ST} & -AL  & -AL & {-AL-TVP} & {-AL-TVP}\\
		\addlinespace
		\hline 
		\addlinespace
		& \multicolumn{7}{c}{Correlations of 10-step Ahead VaR forecasts} \\
		\addlinespace
		\cmidrule(lr){2-8}
		\addlinespace
		{RiskMetrics} & {1.000} & {0.943} & {0.907} & {0.953} & {0.897} & {0.989} & {0.868}\\
		{GARCH-N} &  & {1.000} & {0.973} & {0.994} & {0.965} & {0.960} & {0.933}\\
		{GJR-ST} &  &  & {1.000} & {0.972} & {0.991} & {0.940} & {0.983}\\
		{GARCH-AL} &  &  &  & {1.000} & {0.964} & {0.969} & {0.937}\\
		{GJR-AL} &  &  &  &  & {1.000} & {0.928} & {0.977}\\
		{GARCH-AL-TVP} &  &  &  &  &  & {1.000} & {0.912}\\
		{GJR-AL-TVP} &  &  &  &  &  &  & {1.000}\\
		\addlinespace
		\hline\hline
		\addlinespace
		& \multicolumn{7}{c}{Correlations of 10-step Ahead ES forecasts} \\
		\addlinespace
		\cmidrule(lr){2-8}
		\addlinespace
		{RiskMetrics} & {1.000} & {0.943} & {0.898} & {0.952} & {0.897} & {0.989} & {0.869}\\
		{GARCH-N} &  & {1.000} & {0.964} & {0.993} & {0.964} & {0.960} & {0.934}\\
		{GJR-ST} &  &  & {1.000} & {0.966} & {0.987} & {0.934} & {0.987}\\
		{GARCH-AL} &  &  &  & {1.000} & {0.965} & {0.968} & {0.939}\\
		{GJR-AL} &  &  &  &  & {1.000} & {0.928} & {0.977}\\
		{GARCH-AL-TVP} &  &  &  &  &  & {1.000} & {0.912}\\
		{GJR-AL-TVP} &  &  &  &  &  &  & {1.000}\\
		\addlinespace
		\hline
		\hline
		\addlinespace
		& \multicolumn{7}{c}{Correlations of 10-step Aggregate VaR forecasts} \\
		\addlinespace
		\cmidrule(lr){2-8}
		\addlinespace
		{RiskMetrics} & {1.000} & {0.951} & {0.910} & {0.954} & {0.902} & {0.982} & {0.863}\\
		{GARCH-N} &  & {1.000} & {0.976} & {0.994} & {0.968} & {0.965} & {0.938}\\
		{GJR-ST} &  &  & {1.000} & {0.967} & {0.987} & {0.945} & {0.982}\\
		{GARCH-AL} &  &  &  & {1.000} & {0.969} & {0.962} & {0.924}\\
		{GJR-AL} &  &  &  &  & {1.000} & {0.934} & {0.966}\\
		{GARCH-AL-TVP} &  &  &  &  &  & {1.000} & {0.921}\\
		{GJR-AL-TVP} &  &  &  &  &  &  & {1.000}\\
		\addlinespace
		\hline
		\hline
		\addlinespace
		& \multicolumn{7}{c}{Correlations of 10-step Aggregate ES forecasts} \\
		\addlinespace
		\cmidrule(lr){2-8}
		\addlinespace
		{RiskMetrics} & {1.000} & {0.951} & {0.908} & {0.955} & {0.903} & {0.984} & {0.866}\\
		{GARCH-N} &  & {1.000} & {0.975} & {0.995} & {0.968} & {0.966} & {0.941}\\
		{GJR-ST} &  &  & {1.000} & {0.967} & {0.987} & {0.943} & {0.985}\\
		{GARCH-AL} &  &  &  & {1.000} & {0.968} & {0.965} & {0.929}\\
		{GJR-AL} &  &  &  &  & {1.000} & {0.935} & {0.971}\\
		{GARCH-AL-TVP} &  &  &  &  &  & {1.000} & {0.920}\\
		{GJR-AL-TVP} &  &  &  &  &  &  & {1.000}\\
		\addlinespace
		\hline
		\hline
		\addlinespace
		\multicolumn{8}{p{0.86\linewidth}}{\textit{Notes:} This table reports the pairwise correlations of the seven GARCH-type risk models for the VaR and ES 10-step \textit{ahead}  forecasts in the upper two panels and for the VaR and ES 10-step \textit{aggregate} forecasts in the lower two panels.}
	\end{tabularx}
\end{table}

\begin{landscape}
	\begin{table}[tbh]
		\footnotesize
		\centering
		\caption{Detailed Encompassing Test Results for 10-Day Ahead Forecasts}
		\label{tab:JointDetailsMultiAhead}
		\begin{tabularx}{0.9\linewidth}{l @{\hspace{0.5cm}} ccccccc }
			\hline \hline
			\addlinespace
			\multicolumn{8}{c}{Joint VaR and ES Test} \\
			\addlinespace
			\hline
			\addlinespace
			& {RiskMetrics} & {GARCH-N} & {GJR-ST} & {GARCH-AL} & {GJR-AL} & {GARCH-AL-TVP} & {GJR-AL-TVP} \\
			\addlinespace
			\hline 
			\addlinespace
			{RiskMetrics} & & (0.32, 0.26) & (0.45, 0.41) & (0.31, 0.63) & (0.75, 0.82) & (0.12, 0.16)* & (0.00, 0.00)* \\
			\addlinespace
			{GARCH-N} & (0.68, 0.74) & & (0.20, 0.63) & (0.33, 0.52) & (0.75, 0.82) & (0.25, 0.20)* & (0.06, 0.05)* \\
			\addlinespace
			{GJR-ST} & (0.55, 0.59)* & (0.80, 0.37) &  & (0.47, 0.55) & (0.86, 0.72) & (0.20, 0.11)* & (0.00, 0.00)* \\
			\addlinespace
			{GARCH-AL} & (0.69, 0.37) & (0.67, 0.48) & (0.53, 0.45) & & (0.86, 0.32) & (0.28, 0.25)* & (0.00, 0.00)* \\
			\addlinespace
			{GJR-AL} & (0.25, 0.18) & (0.25, 0.18)* & (0.14, 0.28)* & (0.14, 0.68)* & & (0.32,0.22)* & (0.00, 0.00)* \\
			\addlinespace
			{GARCH-AL-TVP} & (0.88, 0.84) & (0.75, 0.80) & (0.80, 0.89) & (0.72, 0.75) & (0.68, 0.78) &  & (0.00, 0.11)* \\
			\addlinespace
			{GJR-AL-TVP} & (1.00, 1.00) & (0.94, 0.95) & (1.00, 1.00) & (1.00, 1.00) & (1.00, 1.00) & (1.00, 0.89) &  \\
			\addlinespace
			\hline 
			\hline
			\\
			\addlinespace
			\multicolumn{8}{c}{Auxiliary ES Test} \\
			\addlinespace
			\hline
			\addlinespace
			& {RiskMetrics} & {GARCH-N} & {GJR-ST} & {GARCH-AL} & {GJR-AL} & {GARCH-AL-TVP} & {GJR-AL-TVP} \\
			\addlinespace
			\hline 
			\addlinespace
			{RiskMetrics} & & 0.26 & 0.41* & 0.63 & 0.82* & 0.16 & 0.00* \\
			\addlinespace
			{GARCH-N} & 0.74 & & 0.63 & 0.52 & 0.82 & 0.20* & 0.05* \\
			\addlinespace
			{GJR-ST} & 0.59 & 0.37 & & 0.55 & 0.72 & 0.11 & 0.00* \\
			\addlinespace
			{GARCH-AL} & 0.37 & 0.48 & 0.45* & & 0.32 & 0.25* & 0.00* \\
			\addlinespace
			{GJR-AL} & 0.18 & 0.18* & 0.28* & 0.68 & & 0.22* & 0.00* \\
			\addlinespace
			{GARCH-AL-TVP} & 0.84 & 0.80 & 0.89 & 0.75 & 0.78 & & 0.11 \\
			\addlinespace
			{GJR-AL-TVP} & 1.00 & 0.95 & 1.00 & 1.00 & 1.00 & 0.89 & \\
			\addlinespace
			\hline 
			\hline 
			\addlinespace
			\multicolumn{8}{p{0.88\linewidth}}{\textit{Notes:} This table reports the estimates of the convex combination parameters $\left(\theta_{1},\theta_{2}\right)$ from the convex link function with intercepts for each pair of models for 10-day ahead forecasts.
				The symbol $^\ast$ indicates that the null hypothesis that the VaR and ES forecasts of a row-heading model jointly encompasses those of a column-heading model is rejected at the $5\%$ significance level.}
		\end{tabularx}
	\end{table}
\end{landscape}

\begin{landscape}
	\begin{table}[tbh]
		\footnotesize
		\centering
		\caption{Detailed Encompassing Test Results for 10-Day Aggregate Forecasts}
		\label{tab:JointDetailsMultiAggregate}
		\begin{tabularx}{0.9\linewidth}{l @{\hspace{0.5cm}} ccccccc }
			\hline \hline
			\addlinespace
			\multicolumn{8}{c}{Joint VaR and ES Test} \\
			\addlinespace
			\hline
			\addlinespace
			& {RiskMetrics} & {GARCH-N} & {GJR-ST} & {GARCH-AL} & {GJR-AL} & {GARCH-AL-TVP} & {GJR-AL-TVP} \\
			\addlinespace
			\hline 
			\addlinespace
			{RiskMetrics} & & (0.12, 0.08)* & (0.08, 0.11)* & (0.07, 0.37)* & (0.32, 0.29)* & (0.05, 0.12)* & (0.05, 0.08)* \\
			\addlinespace
			{GARCH-N} & (0.88, 0.92) &  & (0.00, 0.26)* & (0.62, 0.88) & (0.73, 0.82) & (0.03, 0.18)* & (0.16, 0.10)* \\
			\addlinespace
			{GJR-ST} & (0.92, 0.89) & (1.00, 0.74) & & (1.00, 1.00) & (0.51, 0.44) & (0.26, 0.29) & (0.08, 0.00)* \\
			\addlinespace
			{GARCH-AL} & (0.93, 0.63) & (0.38, 0.12)* & (0.00, 0.00)* & & (0.43, 0.53)* & (0.00, 0.11)* & (0.17, 0.09)* \\
			\addlinespace
			{GJR-AL} &  (0.68, 0.71) & (0.27, 0.18) & (0.49, 0.56) & (0.57, 0.47)* & & (0.00, 0.08)* & (0.00, 0.17)* \\
			\addlinespace
			{GARCH-AL-TVP} &  (0.95, 0.88) & (0.97, 0.82) & (0.74, 0.71) & (1.00, 0.89) & (1.00, 0.92) & & (0.11, 0.15)* \\
			\addlinespace
			{GJR-AL-TVP} & (0.95, 0.92) & (0.84, 0.90) & (0.92, 1.00) & (0.83, 0.91) & (1.00, 0.93) & (0.89, 0.85) & \\
			\addlinespace
			\hline 
			\hline
			\\
			\addlinespace
			\multicolumn{8}{c}{Auxiliary ES Test} \\
			\addlinespace
			\hline
			\addlinespace
			& {RiskMetrics} & {GARCH-N} & {GJR-ST} & {GARCH-AL} & {GJR-AL} & {GARCH-AL-TVP} & {GJR-AL-TVP} \\
			\addlinespace
			\hline 
			\addlinespace
			{RiskMetrics} & & 0.08* & 0.11* & 0.37*  & 0.29* & 0.12* & 0.08* \\
			\addlinespace
			{GARCH-N} & 0.92 & & 0.26 & 0.88 & 0.82 & 0.18* & 0.10* \\
			\addlinespace
			{GJR-ST} & 0.89 & 0.74 & & 1.00 & 0.44 & 0.29* & 0.00* \\
			\addlinespace
			{GARCH-AL} & 0.63 & 0.12* & 0.00* & & 0.53 & 0.11* & 0.09* \\
			\addlinespace
			{GJR-AL} & 0.71 & 0.18 & 0.56 & 0.47 & & 0.08* & 0.17* \\
			\addlinespace
			{GARCH-AL-TVP} & 0.88 & 0.82 & 0.71 & 0.89 & 0.92 & & 0.15 \\
			\addlinespace
			{GJR-AL-TVP} & 0.92 & 0.90 & 1.00 & 0.91 & 0.93 & 0.85 & \\
			\addlinespace
			\hline 
			\hline 
			\addlinespace
			\multicolumn{8}{p{0.88\linewidth}}{\textit{Notes:} This table reports the estimates of the convex combination parameters $\left(\theta_{1},\theta_{2}\right)$ from the convex link function with intercepts for each pair of models for 10-day aggregate forecasts.
				The symbol $^\ast$ indicates that the null hypothesis that the VaR and ES forecasts of a row-heading model jointly encompasses those of a column-heading model is rejected at the $5\%$ significance level.}
		\end{tabularx}
	\end{table}
\end{landscape}

\begin{table}[t]
	\footnotesize
	\centering
	\caption{Backtesting Results for One-Step Ahead Forecasts}
	\label{tab:AbsEvaluation}
	\begin{tabular}{lclccccccc}
		\hline\hline
		\addlinespace
		{Models} & {Violation} & ESR & UC & CC & DQ & VQR & MF & BD & NZ\\
		& {Ratio} &   &   &   &   &   &   &   &  \\
		\addlinespace
		\hline 
		\addlinespace
		{Historical Sim} & {1.47} & {1.08} & {$<$0.001} & {$<$0.001} & {$<$0.001} & {$<$0.001} & {0.03} & {$<$0.001} & {$<$0.001}\\
		\addlinespace
		{RiskMetrics} & {1.68} & {1.18} & {$<$0.001} & {$<$0.001} & {0.12} & {$<$0.001} & {$<$0.001} & {$<$0.001} & {$<$0.001}\\
		\addlinespace
		{GARCH-N} & {1.57} & {1.12} & {$<$0.001} & {$<$0.001} & {0.39} & {$<$0.001} & {$<$0.001} & {$<$0.001} & {$<$0.001}\\
		\addlinespace
		{GJR-ST} & {1.20} & {0.98} & \textbf{0.09} & \textbf{0.17} & \textbf{0.96} & \textbf{0.36} & \textbf{0.55} & \textbf{0.23} & \textbf{0.14}\\
		\addlinespace
		{GARCH-AL} & {0.84} & {0.99} & \textbf{0.15} & \textbf{0.16} & \textbf{0.88} & \textbf{0.19} & \textbf{0.71} & \textbf{0.15} & \textbf{0.09}\\
		\addlinespace
		{GJR-AL} & {0.69} & {0.98} & {$<$0.001} & {0.01} & {0.89} & {$<$0.001} & {0.62} & {0.01} & {$<$0.001}\\
		\addlinespace
		{GARCH-AL-TVP} & {1.07} & {0.98} & \textbf{0.56} & \textbf{0.19} & \textbf{0.53} & \textbf{0.84} & \textbf{0.54} & \textbf{0.55} & \textbf{0.52}\\
		\addlinespace
		{GJR-AL-TVP} & {1.01} & {0.97} & {0.91} & {0.40} & {0.96} & {0.04} & {0.43} & {0.06} & {0.02}\\
		\addlinespace
		{GAS-1F} & {1.20} & {1.03} & \textbf{0.09} & \textbf{0.23} & \textbf{0.68} & \textbf{0.31} & \textbf{0.40} & \textbf{0.18} & \textbf{0.13}\\
		\addlinespace
		{SAV-CAViaR-ES} & {1.11} & {1.02} & \textbf{0.36} & \textbf{0.18} & \textbf{0.80} & \textbf{0.66} & \textbf{0.54} & \textbf{0.45} & \textbf{0.55}\\
		\addlinespace
		{AS-CAViaR-ES} & {1.15} & {1.02} & \textbf{0.21} & \textbf{0.29} & \textbf{0.87} & \textbf{0.25} & \textbf{0.50} & \textbf{0.26} & \textbf{0.30}\\
		\addlinespace
		\hline 
		\bottomrule 
		\addlinespace
		\multicolumn{10}{p{\linewidth}}{\textit{Notes:} The Violation Ratio is given by $\hat{\alpha}/\alpha$, where $\hat{\alpha}=T^{-1}\sum_{t\in \mathfrak{T}}\mathds{1}_{\{ Y_{t+1}<\hat{q}_{t} \}}$ and the empirical ES ratio is computed as $\text{ESR} =\sum_{t \in \mathfrak{T}} \left.\left[Y_{t+1}\right.\mathds{1}_{\{ Y_{t+1}<\hat{q}_{t} \}}\right] / \sum_{t \in \mathfrak{T}} \left[ \hat{e}_{t}\mathds{1}_{\{ Y_{t+1}<\hat{q}_{t} \}}\right]$.
			Both ratios are expected to equal one for correctly specified VaR and ES forecasts.
			The remaining columns report backtesting $p$-values for the unconditional coverage (UC) test of \cite{Kupiec1995}, the conditional coverage (CC) test of \cite{Christoffersen1998}, the dynamic quantile (DQ) test of \cite{Engle2004}, the VQR test of \cite{Gaglianone2011}, the ES backtest of \cite{McNeil2000} (MF), the regression-based ES backtest of \cite{BayerDimi2019} (BD), and for the calibration test of \cite{NoldeZiegel2017AAS} (NZ).
			Rows with p-values in bold indicate that for a respective model, the null hypotheses of all seven backtests cannot be rejected at the $5\%$ significance level.}
	\end{tabular}
\end{table}

\FloatBarrier
\clearpage

\begin{figure}[ph!]
	\includegraphics[width=\linewidth]{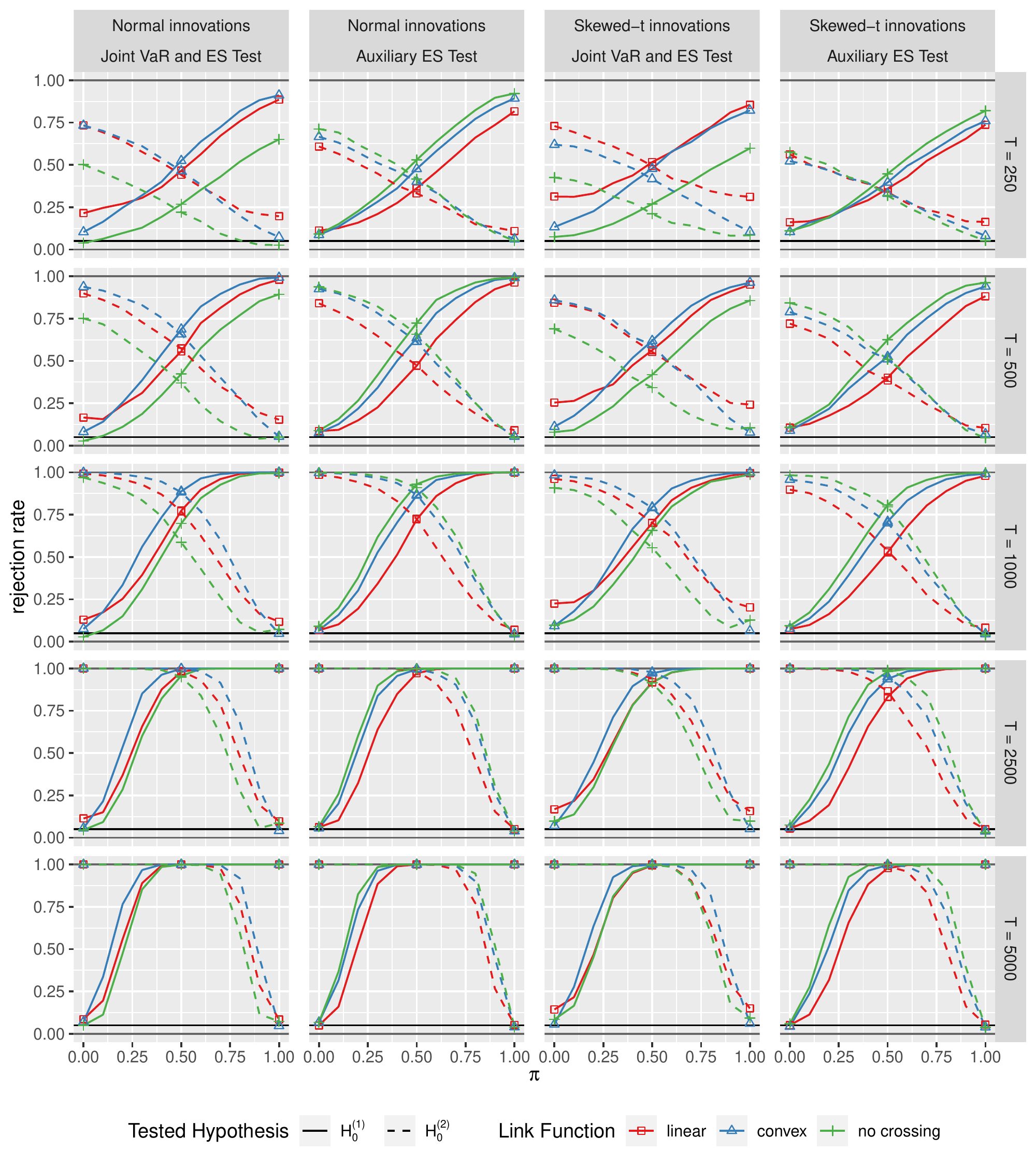}
	\caption{
		This figure shows raw power curves (empirical rejection frequencies) for the joint VaR and ES and the auxiliary ES encompassing tests with a nominal size of $5\%$.
		The employed link functions are indicated with the line color and symbol shape while the line type refers to the tested null hypothesis.
		The plot rows depict different sample sizes while the plot columns show results for the two innovation distributions described in \eqref{eqn:GARCHModel} - \eqref{eqn:GARCHresiduals} and for the joint and the auxiliary tests.
		An ideal test exhibits a rejection rate of $5\%$ for $\pi=0$ and for $\mathbb{H}_0^{(1)}$ (and inversely for $\pi=1$ and $\mathbb{H}_0^{(2)}$) and as sharply increasing rejection rates as possible for increasing (decreasing) values of $\pi$.}
	\label{fig:sim_VaRES_GARCH_rawpower}
\end{figure}

\begin{figure}[ph!]
	\includegraphics[width=\linewidth]{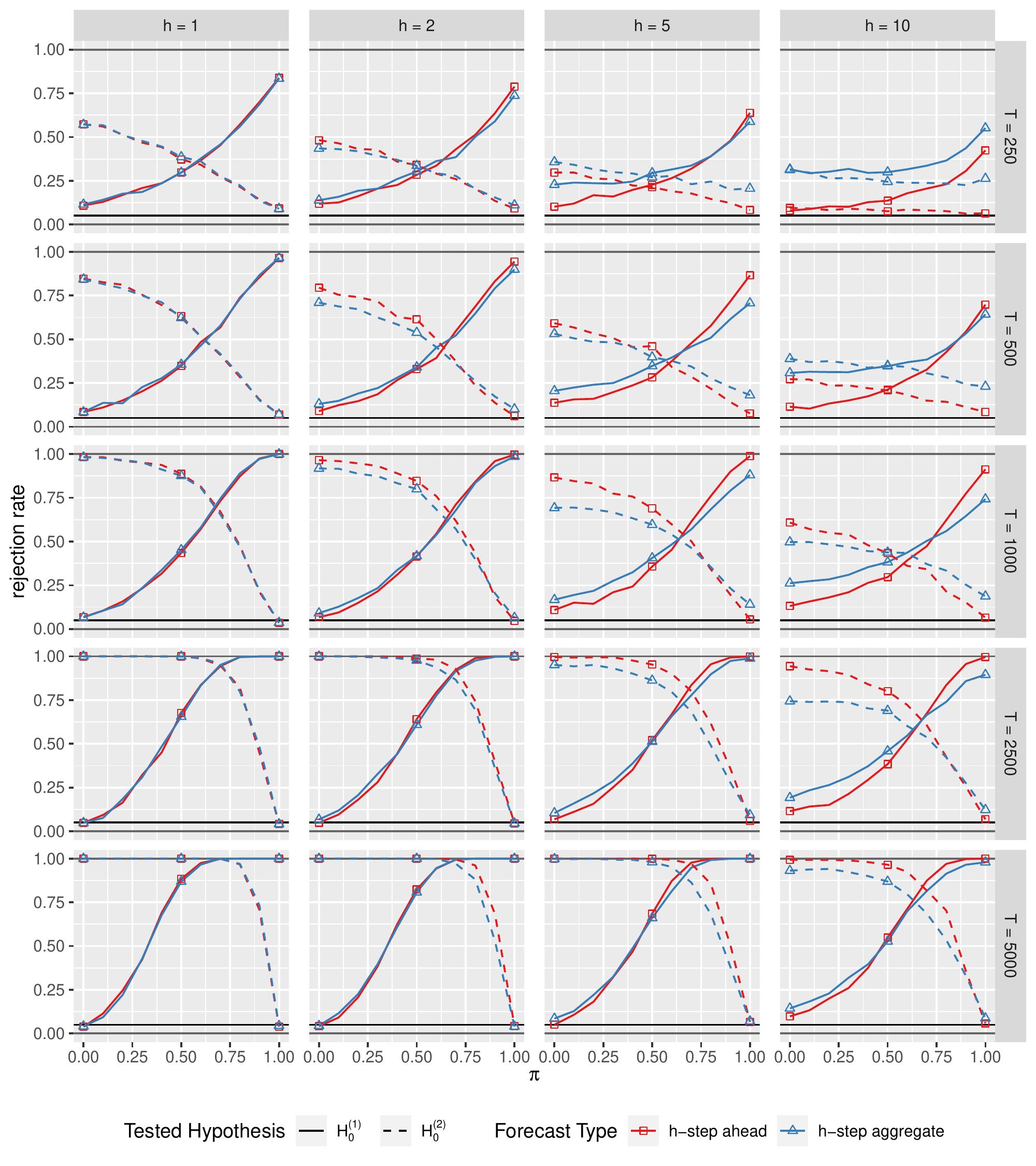}
	\caption{This figure shows raw power curves (empirical rejection frequencies) for the joint VaR and ES encompassing test with a nominal size of $5\%$, for $h$-step ahead and $h$-step aggregated forecasts indicated with different colors, and for the two tested null hypotheses indicated with different line types.
		The plot rows depict different sample sizes, while the plot columns refer to different forecast horizons $h$.
		An ideal test exhibits a rejection frequency of $5\%$ for $\pi=0$ and for $\mathbb{H}_0^{(1)}$ (and inversely for $\pi=1$ and $\mathbb{H}_0^{(2)}$) and as sharply increasing rejection rates as possible for increasing (decreasing) values of $\pi$. 
		Note that we use a Bernoulli draw based combination method in this section as opposed to the variance combination in Section \ref{sec:onestep_forecasts} and hence, the results of the one-step ahead forecasts are not necessarily identical.}
	\label{fig:sim_multi_VaRES_RawPower}
\end{figure}

\begin{figure}[!ht]
	\includegraphics[width=\linewidth]{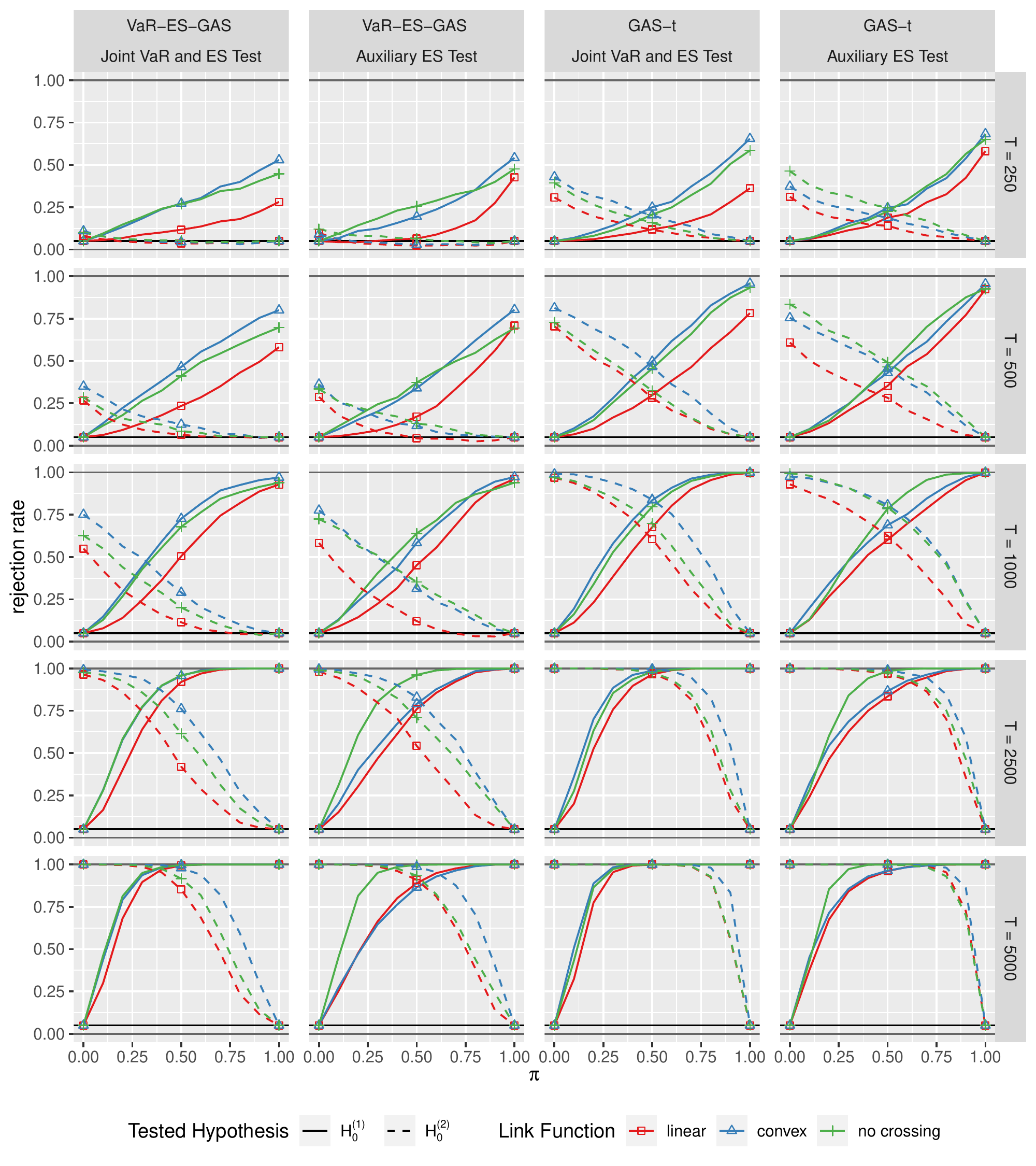}
	\caption{This figure shows size-adjusted power curves for the joint VaR and ES encompassing test and the auxiliary ES test with a nominal size of $5\%$ and for one-step ahead forecasts of the two GAS-based DGPs described in Section \ref{sec:add_simulation}.
		The plot rows depict different sample sizes, while the colors indicate the three different link functions and the line types refer to the two tested null hypotheses. The plot columns show results for the models described in \eqref{eqn:GAS1F}, \eqref{eqn:GAS2F} and \eqref{eqn:GAS_t} and for the joint and auxiliary tests.
		An ideal test exhibits a rejection frequency of $5\%$ for $\pi=0$ and for $\mathbb{H}_0^{(1)}$ (and inversely for $\pi=1$ and $\mathbb{H}_0^{(2)}$) and as sharply increasing rejection rates as possible for increasing (decreasing) values of $\pi$.}
	\label{fig:sim_other_SA}
\end{figure}

\begin{figure}[!ht]
	\includegraphics[width=\linewidth]{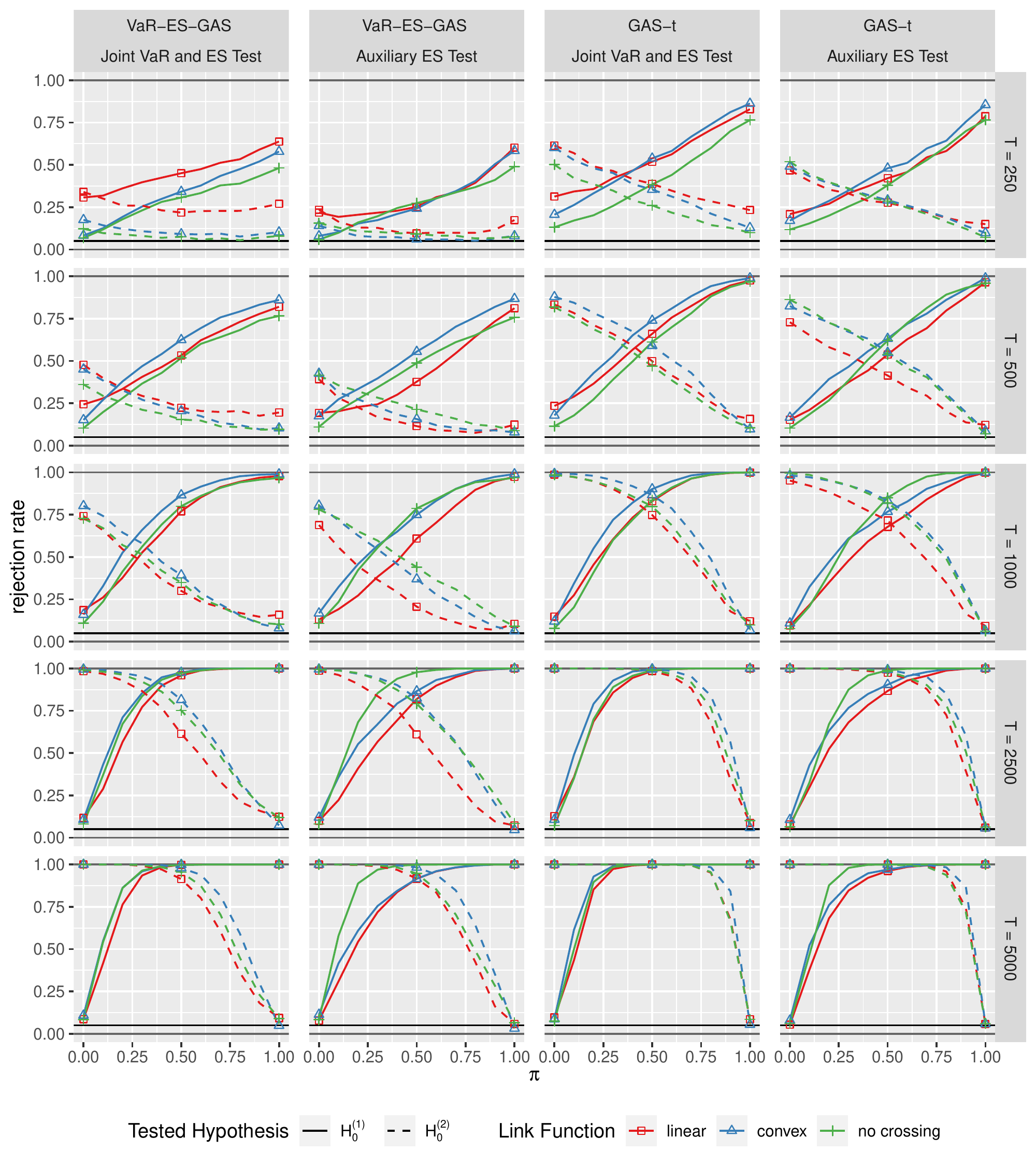}
	\caption{This figure shows raw power curves (empirical rejection frequencies) for the joint VaR and ES encompassing test and the auxiliary ES test with a nominal size of $5\%$ and for one-step ahead forecasts of the two GAS-based DGPs described in Section \ref{sec:add_simulation}.
		The plot rows depict different sample sizes, while the colors indicate the three different link functions and the line types refer to the two tested null hypotheses. The plot columns show results for the models described in \eqref{eqn:GAS1F}, \eqref{eqn:GAS2F} and \eqref{eqn:GAS_t} and for the joint and auxiliary tests.
		An ideal test exhibits a rejection frequency of $5\%$ for $\pi=0$ and for $\mathbb{H}_0^{(1)}$ (and inversely for $\pi=1$ and $\mathbb{H}_0^{(2)}$) and as sharply increasing rejection rates as possible for increasing (decreasing) values of $\pi$.}
	\label{fig:sim_other}
\end{figure}

\begin{figure}[htb]
	\includegraphics[width=\linewidth]{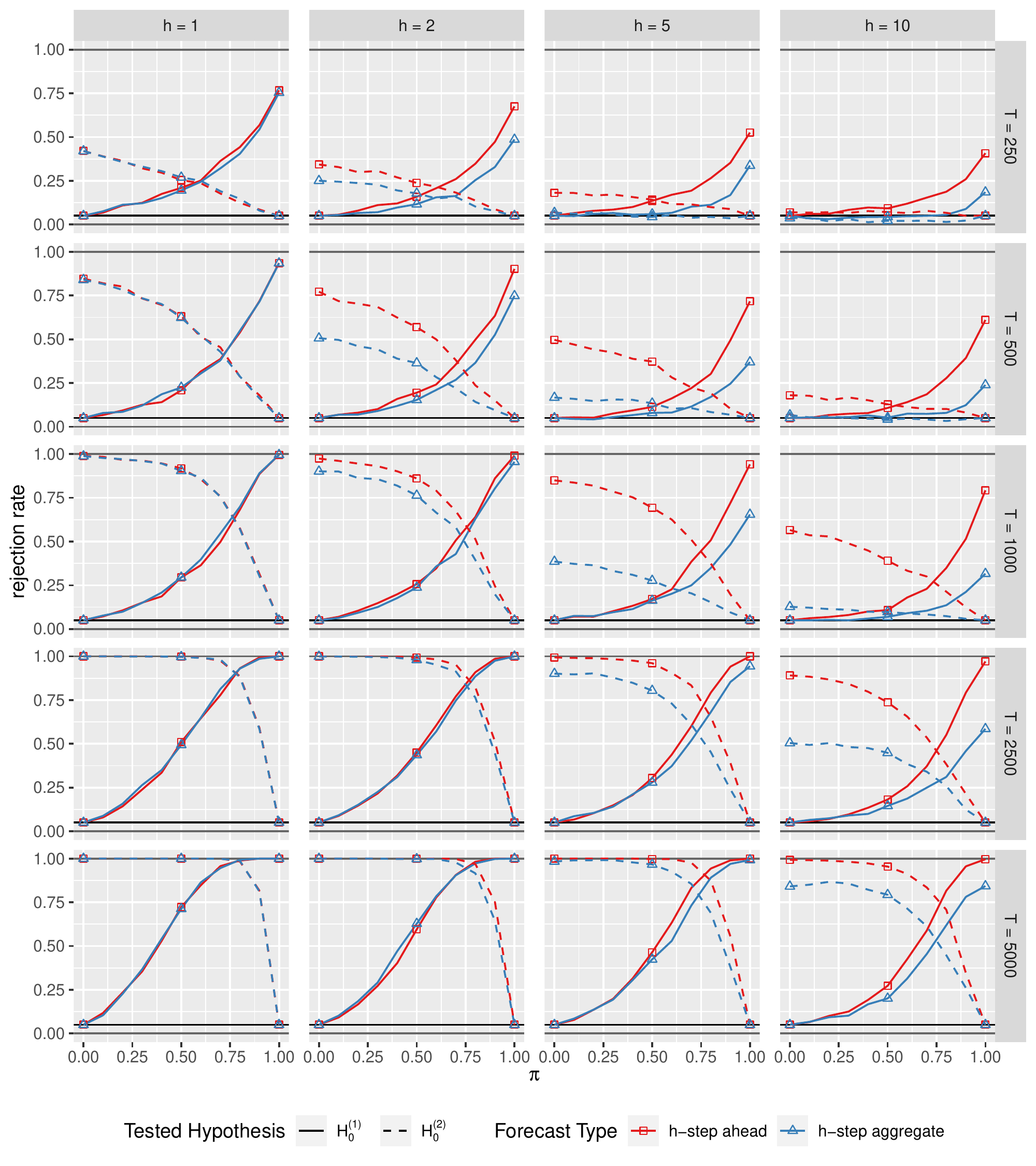}
	\caption{This figure shows size-adjusted power curves for the auxiliary ES encompassing test with a nominal size of $5\%$ for $h$-step ahead and $h$-step aggregate forecasts stemming from the GARCH process specifications in (\ref{eqn:GARCHModel}) - (\ref{eqn:GARCHresiduals}).
		The $h$-step ahead and aggregate forecasts are indicated by different colors and the two tested null hypotheses are indicated with different line types. 
		The plot rows depict different sample sizes, while the plot columns refer to different forecast horizons $h=1,2,5,10$.
		An ideal test exhibits a rejection frequency of $5\%$ for $\pi=0$ and for $\mathbb{H}_0^{(1)}$ (and inversely for $\pi=1$ and $\mathbb{H}_0^{(2)}$) and as sharply increasing rejection rates as possible for increasing (decreasing) values of $\pi$.} 
	\label{fig:sim_multi_AuxES_SA}
\end{figure}

\begin{figure}[htb]
	\includegraphics[width=\linewidth]{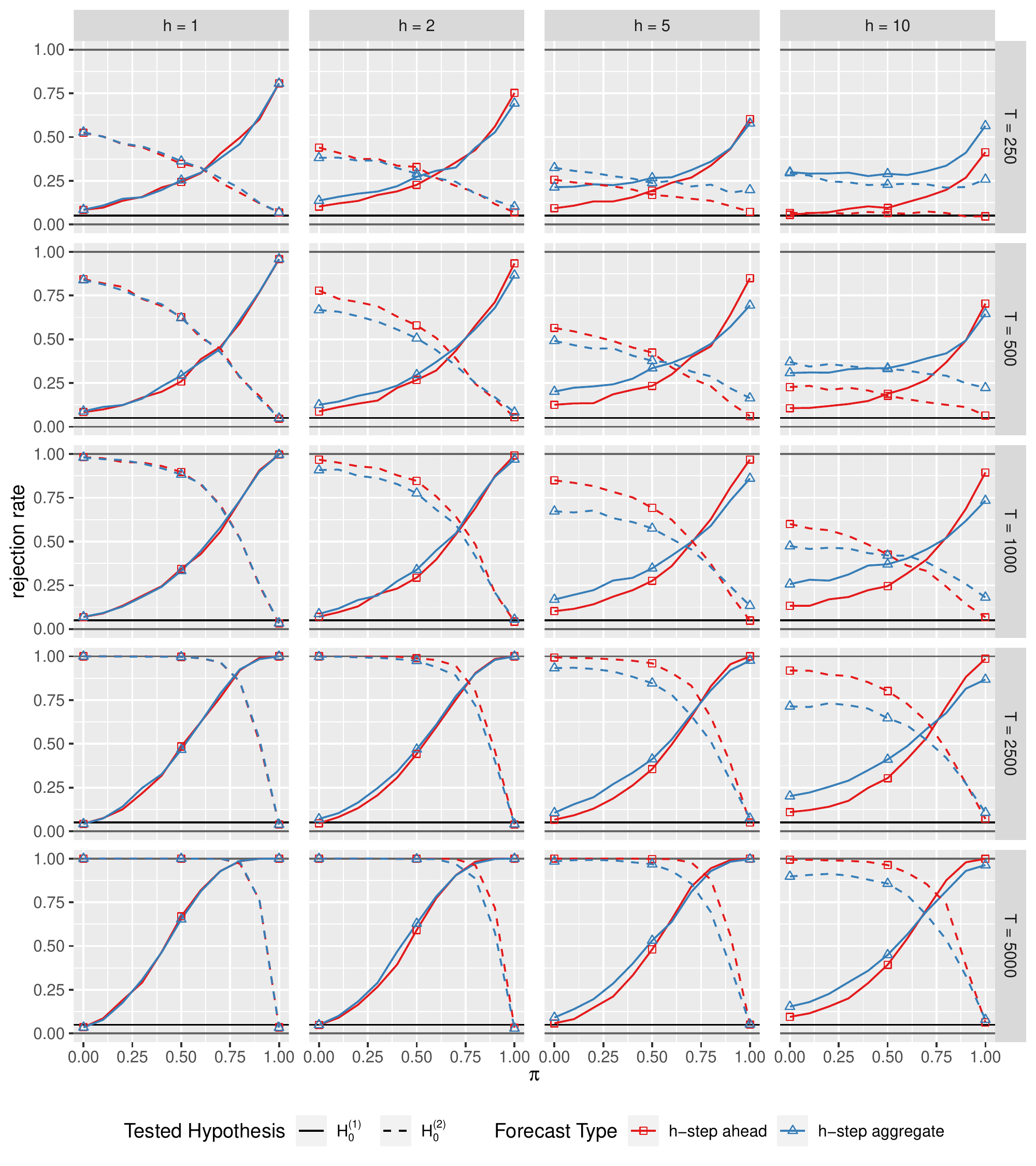}
	\caption{This figure shows raw power curves (empirical rejection frequencies) for the auxiliary ES encompassing test with a nominal size of $5\%$ for $h$-step ahead and $h$-step aggregate forecasts stemming from the GARCH process specifications in (\ref{eqn:GARCHModel}) - (\ref{eqn:GARCHresiduals}).
		The $h$-step ahead and aggregate forecasts are indicated by different colors and the two tested null hypotheses are indicated with different line types. 
		The plot rows depict different sample sizes, while the plot columns refer to different forecast horizons $h=1,2,5,10$.
		An ideal test exhibits a rejection frequency of $5\%$ for $\pi=0$ and for $\mathbb{H}_0^{(1)}$ (and inversely for $\pi=1$ and $\mathbb{H}_0^{(2)}$) and as sharply increasing rejection rates as possible for increasing (decreasing) values of $\pi$.} 
	\label{fig:sim_multi_AuxES}
\end{figure}

\begin{figure}[htb]
	\includegraphics[width=\linewidth]{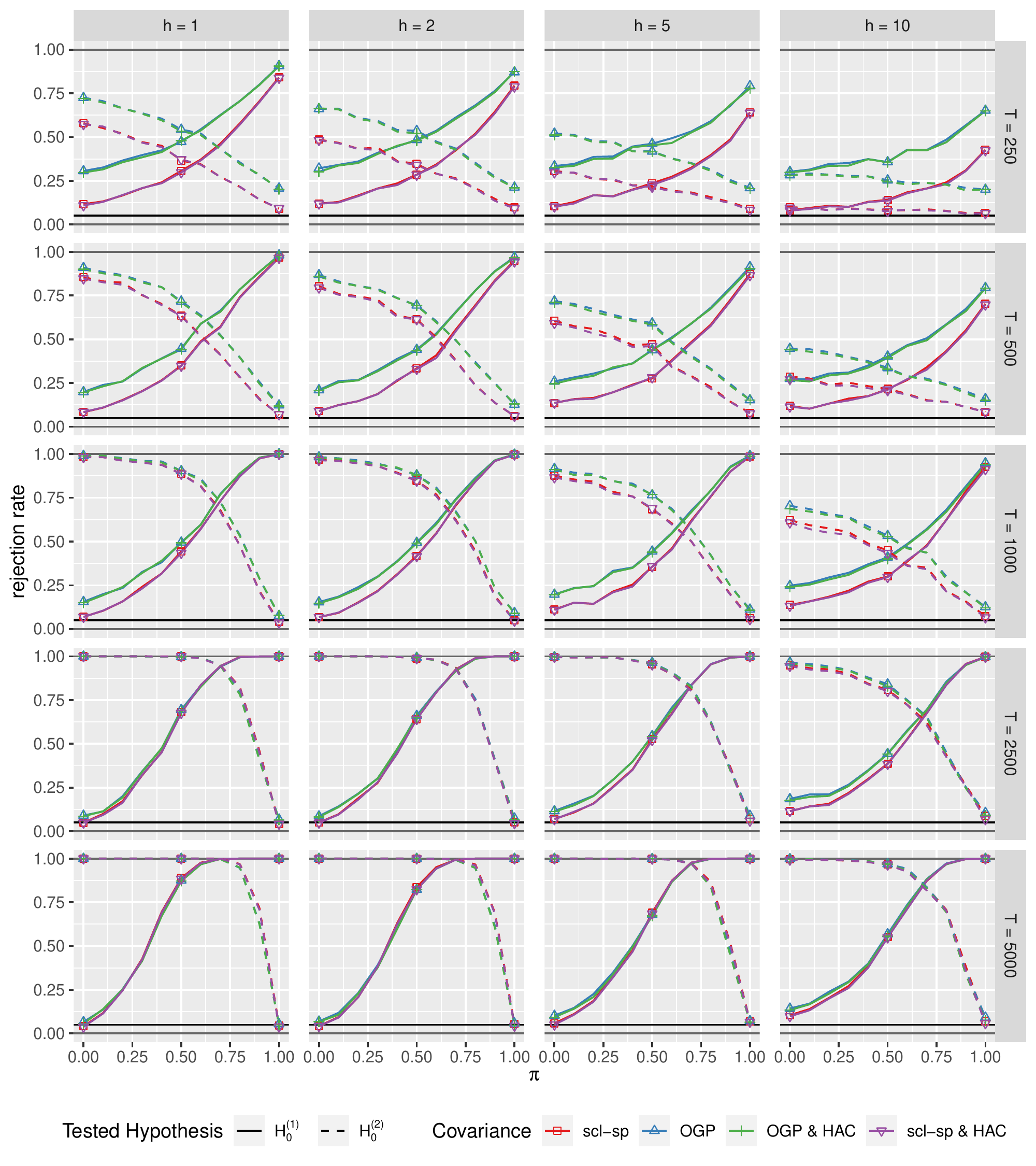}
	\caption{This figure shows raw power curves (empirical rejection frequencies) for the joint VaR and ES encompassing test with a nominal size of $5\%$ for $h$-step ahead forecasts stemming from the GARCH process specifications in (\ref{eqn:GARCHModel}) - (\ref{eqn:GARCHresiduals}).
		The plot rows depict different sample sizes, the plot columns show the different forecast horizons $h$, the colors indicate the different covariance estimators, and the line types refer to the two tested null hypotheses. 
		An ideal test exhibits a rejection frequency of $5\%$ for $\pi=0$ and for $\mathbb{H}_0^{(1)}$ (and inversely for $\pi=1$ and $\mathbb{H}_0^{(2)}$) and as sharply increasing rejection rates as possible for increasing (decreasing) values of $\pi$.}
	\label{fig:sim_multi_VaRES_iter}
\end{figure}

\begin{figure}[htb]
	\includegraphics[width=\linewidth]{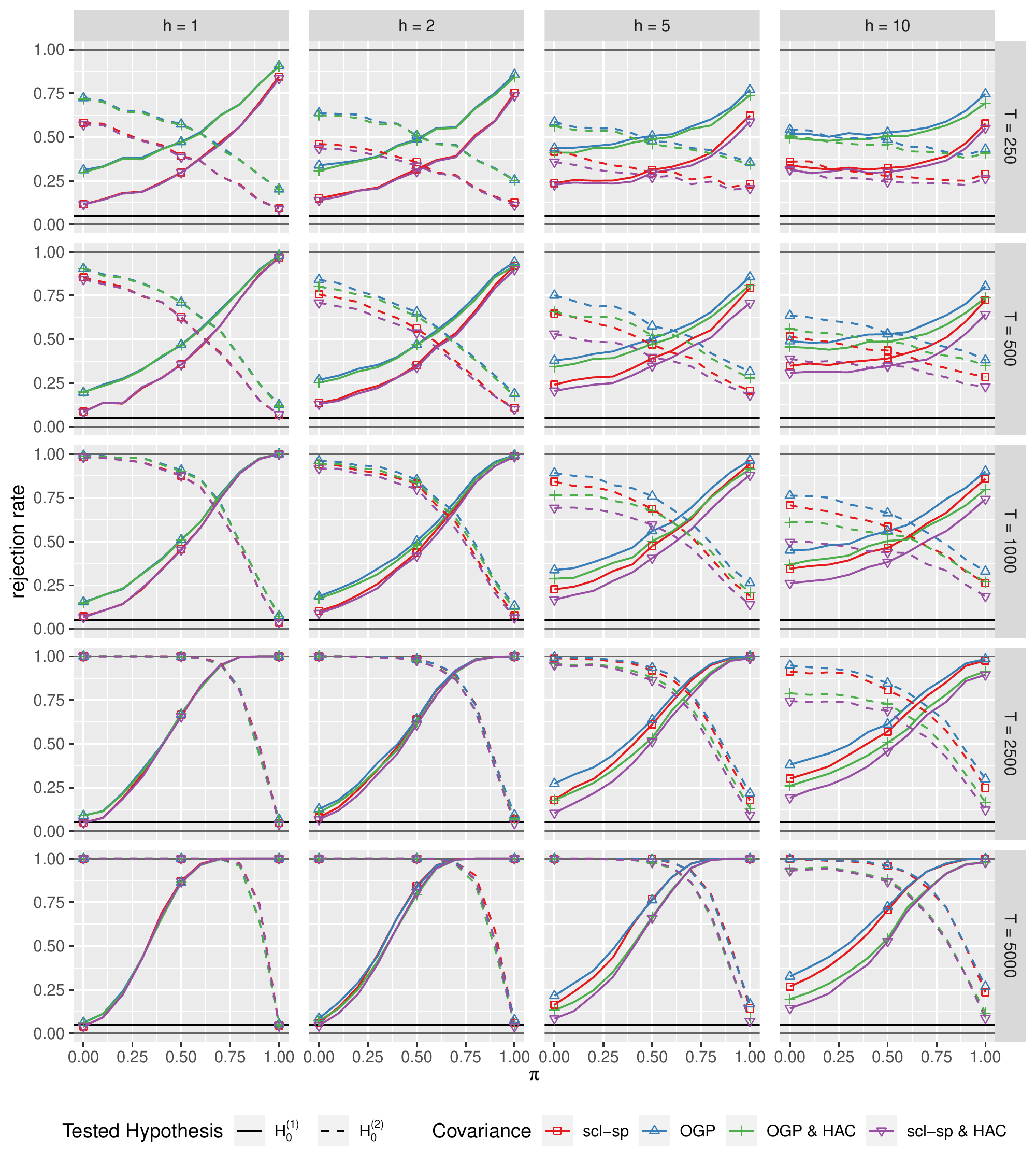}
	\caption{This figure shows raw power curves (empirical rejection frequencies) for the joint VaR and ES encompassing test with a nominal size of $5\%$ for $h$-step aggregate forecasts stemming from the GARCH process specifications in (\ref{eqn:GARCHModel}) - (\ref{eqn:GARCHresiduals}).
		The plot rows depict different sample sizes, the plot columns show the different forecast horizons $h$, the colors indicate the different covariance estimators, and the line types refer to the two tested null hypotheses. 
		An ideal test exhibits a rejection frequency of $5\%$ for $\pi=0$ and for $\mathbb{H}_0^{(1)}$ (and inversely for $\pi=1$ and $\mathbb{H}_0^{(2)}$) and as sharply increasing rejection rates as possible for increasing (decreasing) values of $\pi$.}
	\label{fig:sim_multi_VaRES_aggr}	
\end{figure}

\end{document}